\documentclass[reqno, 11pt]{amsart}

\usepackage[english]{babel}
\usepackage[T1]{fontenc}
\usepackage[utf8]{inputenc}

\usepackage{amsfonts,amssymb,amsbsy,amsmath,amsthm,enumerate}
\usepackage{hyperref}
\usepackage{eucal}
\usepackage{color}
\usepackage{mathrsfs}
\usepackage{esint}

\usepackage{lxfonts}

\theoremstyle{plain}
\newtheorem{theorem}{Theorem}[section]
\theoremstyle{plain}
\newtheorem{corollary}[theorem]{Corollary}%
\theoremstyle{plain}
\newtheorem{proposition}[theorem]{Proposition}%
\theoremstyle{plain}
\newtheorem{lemma}[theorem]{Lemma}%
\theoremstyle{definition}
\newtheorem{definition}[theorem]{Definition}%
\theoremstyle{definition}
\newtheorem{example}[theorem]{Example}%
\theoremstyle{definition}

\theoremstyle{definition}
\newtheorem{remark}[theorem]{Remark}%

\numberwithin{equation}{section}
\numberwithin{figure}{section}
\numberwithin{table}{section}

\newcommand{\R}{\mathbb{R}}
\newcommand{\N}{\mathbb{N}}
\newcommand{\C}{\mathbb{C}}                           %

\newcommand{\Z}{\mathbb{Z}}
\newcommand{\T}{\mathbb{T}}

\newcommand{\s}[1]{\CMcal{#1}}
\newcommand{\bb}[1]{\mathscr{#1}}
\newcommand{\rr}[1]{\mathfrak{#1}}
\newcommand{\n}[1]{\mathbb{#1}}

\newcommand{\nnorm}[1]{{\left\vert\kern-0.25ex\left\vert\kern-0.25ex\left\vert #1
    \right\vert\kern-0.25ex\right\vert\kern-0.25ex\right\vert}}
\newcommand{\bra}[1]{\langle#1|}
\newcommand{\ket}[1]{|#1\rangle}

\newcommand{\expo}[1]{\,\mathrm{e}^{#1}\,}                 %
\newcommand{\dd}{\,\mathrm{d}}
\newcommand{\ii}{\,\mathrm{i}\,}

\newcommand{\virg}[1]{\lq\lq#1\rq\rq}                %
\newcommand{\ie}{\textsl{i.\@e.\,}}
\newcommand{\eg}{\textsl{e.\@g.\,}}
\newcommand{\cf}{\textsl{cf}.\,}

\title[Quantization of Edge Currents along
Magnetic Interfaces]{Quantization of Edge Currents along
Magnetic Interfaces: A $K$-theory approach}

\author[G. De~Nittis]{Giuseppe De Nittis}

\address[G. De~Nittis]{Facultad de Matem\'aticas \& Instituto de F\'{\i}sica,
  Pontificia Universidad Cat\'olica de Chile,
  Santiago, Chile.}
\email{gidenittis@mat.uc.cl}

\author[E. Guti\'errez]{Esteban Guti\'errez}

\address[E. Guti\'errez]{Facultad de Matem\'aticas,
  Pontificia Universidad Cat\'olica de Chile,
  Santiago, Chile.}
\email{egutierreo@mat.uc.cl}

\thanks{\textbf{MSC2010}
Primary: 	81R60;
Secondary: 	58B34, 	81R15, 	46L80.}

\thanks{\textbf{Keywords.}
Magnetic interfaces, Iwatsuka Hamiltonian, edge currents, K-theory.}
\begin{document}

\begin{abstract}
The purpose of this paper is to investigate the propagation of \emph{topological currents} along  \emph{magnetic interfaces} (also known as magnetic walls) of a two-dimensional material. We consider tight-binding magnetic models associated to generic \emph{magnetic multi-interfaces} and describe the $K$-theoretical setting in which a \emph{bulk-interface duality} can be derived. Then, the (trivial) case of a \emph{localized} magnetic field and the 
(non trivial) case of the \emph{Iwatsuka} magnetic field are considered in full detail. This is a pedagogical preparatory work that aims to anticipate the study of more complicated multi-interface magnetic systems. 
 
\end{abstract}

\maketitle

\vspace{-5mm}
\tableofcontents

\section{Introduction}\label{sect:intro}

The purpose of this paper is to investigate the propagation of \emph{topological currents} along  \emph{magnetic interfaces} (also known as magnetic walls) of a two-dimensional material. A magnetic interface is a  thin region of  space (ideally a one-dimensional curve) which separates the material in two parts subjected to perpendicular uniform magnetic fields of distinct intensity.
The typical example of  magnetic interface is provided by the \emph{Iwatsuka magnetic field} \cite{iwatsuka-85}, which is a magnetic field oriented orthogonally to the plane  with a (monotone) strength function
$B_{\rm I}:\R^2\to\R$ that does not depend on the $y$-direction and such that $\lim_{x\to\pm\infty}B_{\rm I}(x)=b_\pm$ where $b_-$ and $b_+$ are distinct values. 
The response of the system to an Iwatsuka magnetic field is the production of a current flowing in the $y$ direction along the interface  (fixed around $x=0$). Such  a current  is carried by  extended states localized near the interface which are the quantum analog of the classical \virg{snake orbits} \cite[Fig. 6.1]{cycon-froese-kirsch-simon-94}. Such \virg{snake orbit  states} have received increasing attention
  in
the recent physics literature \cite{reijniers-peeters-00,reijniers-matulis-chang-peeters-vasilopoulos-02,ghosh-demartino-hausler-dellanna-08,oroszlany-rakyta-kormanyos-lambert-08,park-sim-08,liu-tiwari-brada-bruder-15}.

\medskip

A  very interesting and  important property of the current carried by the \virg{snake orbit  states}  along the interface is its topological quantization. For the  Iwatsuka magnetic field, this fact has been rigorously proved in \cite{dombrowski-germinet-raikov-11}  for continuous models and in \cite{kotani-schulzbaldes-villegasblas-14} for discrete (or tight-binding) models.
In both cases the crucial result can be stated  by considering the \virg{full} Hamiltonian
$\widehat{\rr{h}}$, which contains the full information about the Iwatsuka magnetic field,
and two \virg{asymptotic} Hamiltonians $\{\rr{h}_-,\rr{h}_+\}$ depending only on the constant magnetic fields $b_-$ and $b_+$ respectively. Accordingly
, the pair $\{\rr{h}_-,\rr{h}_+\}$ carries information only about the  asymptotic behavior  of the Iwatsuka magnetic field. 
Let $\Delta$ be an energy domain that sits inside a non-trivial gap of $\sigma(\rr{h}_-)\cup\sigma(\rr{h}_+)$ (asymptotic open gap hypothesis). Let $J_{\s{I}}(\Delta)$ be the interface current carried by the extended \virg{snake orbit  states} of $\widehat{\rr{h}}$ of energy $\Delta$
and $\sigma_{\s{I}}(\Delta):=eJ_{\s{I}}(\Delta)$ the associated \emph{interface conductance} ($e>0$ is the magnitude of the electron charge). Then, it holds true that
\begin{equation}\label{eq:int_01}
\sigma_{\s{I}}(\Delta)\;=\;\frac{e^2}{h}(N_+-N_-)\;,\qquad N_\pm\in\Z,
\end{equation}
where  $h$ is the  Planck's constant. The integers $N_\pm$ are related to the topology of the
asymptotic system. Let $\mu\in\Delta$ be a given \emph{Fermi energy} and $\rr{p}_{\mu,\pm}:=\chi_{(-\infty,\mu]}(\rr{h}_\pm)$ the  Fermi (spectral) projections in the gap of the asymptotic Hamiltonians. Then,  $N_\pm:={\rm Ch}(\rr{p}_{\mu,\pm})$ are the Chern numbers of such projectors. Since it is well known that the Chern numbers of the Fermi projections provide de values of the  Hall conductance (in unit of $e^2/h$) in the bulk of a quantum Hall system \cite{thouless-kohmoto-nightingale-nijs-82,kunz-87,bellissard-elst-schulz-baldes-94,avron-seiler-simon-94}, one can rewrite \eqref{eq:int_01} in the form
\begin{equation}\label{eq:int_02}
\sigma_{\s{I}}(\Delta)\;=\;\sigma_{b_+}(\mu)\;-\;\sigma_{b_-}(\mu)\;.
\end{equation}
where $\sigma_{b_\pm}(\mu)$ are the asymptotic \emph{Hall conductances} at energy $\mu$ (\cf equation \eqref{cond_cost_eq}). If one interprets the collection $\sigma_{\rm bulk}(\mu):=\{\sigma_{b_\pm}(\mu)\}$ as a description of the conductance of the 
\virg{bulk system}
described by the  Hamiltonians  $\{\rr{h}_-,\rr{h}_+\}$ (\cf Definition \ref{def:cond_bulk}), then equation \eqref{eq:int_02} expresses the \emph{bulk-interface duality} which  is a generalization, or rather a manifestation, of 
the celebrated bulk-boundary correspondence \cite{hatsugai-93,elbau-graf-02,kellendonk-richter-schulz-baldes-02}. The latter is a  phenomenon  ubiquitous  in the physics of \emph{topological insulators} \cite{hasan-kane-10,moore-10,kane-moore-11}. 

\medskip

The derivation of \eqref{eq:int_01} provided in  \cite{dombrowski-germinet-raikov-11} only considered
Landau operators and a restricted class of perturbations,
and relies  on techniques from functional analysis and spectral theory. The proof of \eqref{eq:int_01} contained in  \cite{kotani-schulzbaldes-villegasblas-14} works for a large family of discrete models with random perturbations and is based on an
index theorem of Noether-Gohberg-Krein type. In this work we will present  an alternative proof of \eqref{eq:int_01} which makes use of $K$-theory for $C^*$-algebras. Indeed, it is well established in the literature that the bulk-boundary correspondence in condensed matter can be successfully explained and studied inside a  $K$-theoretical framework \cite{schulz-baldes-kellendonk-richter-00,kellendonk-schulz-baldes-04,bourne-kellendonk-rennie-17,bourne-rennie-18,prodan-schulz-baldes-book}. For that reason it seems surprising that problems concerning magnetic interfaces have never been studied before with $K$-theoretical technique. One of the primary aims of this work is to try to fill this gap.

\medskip

We focus our study on magnetic tight-binding models without considering random perturbations. Accordingly, the setting in which we derive our results is quite close to that of \cite{kotani-schulzbaldes-villegasblas-14}. The fact that we are not considering random potentials does not constitute any loss of information as long as $K$-theory is used. In fact, random potentials are usually introduced by families of covariant operators parametrized by a configuration space $\Omega_{\rm random}$ which is generally assumed to be contractible. 
Therefore, in view of its contractibility, $\Omega_{\rm random}$ does not add any information to the underlying $K$-theory and can be replaced with a singleton. The latter fact is equivalent to consider only non-random systems. The decision to focus on tight-binding models is motivated by the desire to avoid the  technicalities  necessary to deal with continuous models in favor of 
bigger  clarity in the description of the algebraic and topological aspects. However, it is worth mentioning that the bulk-boundary correspondence can be   treated with $K$-theoretical methods even for continuous models \cite{kellendonk-schulz-baldes-04,bourne-rennie-18,ludewig-thiang-20}.

\medskip

It is important to  point out that the aim of this works goes beyond the goal of proving the equation \eqref{eq:int_01} by using a different  technique. Indeed, the main purpose  is to set a precise mathematical background for the analysis of general \emph{multi-interface} magnetic systems. In a nutshell
a magnetic system with a multi-interface of order $N$ is a two-dimensional material divided in $N+1$ (asymptotic) regions where the magnetic field has constant strength. The  Iwatsuka magnetic field provides an example of magnetic interface of order $N=1$, but more general examples can easily be imagined. One of the major contribution of this work is to show that
systems
  with magnetic multi-interfaces can be naturally studied by using  $K$-theory. A primary role in our analysis is played by
the \emph{magnetic flux} $f_B:=\expo{\ii B}$  associated with the magnetic field $B$. The magnetic flux, together with all its $\Z^2$-translations, defines a commutative $C^*$-algebra which, in view of the \emph{Gelfand isomorphism}, can be represented as the set of continuous functions $\s{C}(\Omega_B)$ over the compact Hausdorff space $\Omega_B$. The latter is known as  the \emph{magnetic hull}. The magnetic hull, endowed with the action of $\Z^2$, provides a topological dynamical system which encodes the structure of the   
  magnetic multi-interface  through its stationary points at infinity. It is worth pointing out that topology of  $\Omega_B$ is usually not trivial and influences substantially the $K$-theory of the magnetic systems under consideration (in contrast to the role of random perturbations). The idea of the use of  the   magnetic hull for the  study of  the topology of magnetic multi-interface systems seems to be  (to the best of our knowledge) a new contribution in the literature concerning topological insulators. 
On the other hand, similar ideas have already been used successfully in a totally different context  \cite{georgescu-iftimovici-02}.
  
  \medskip
  
This work is written with a pedagogical spirit. Along with original results and ideas, it contains an important amount of review material (\eg the rich  appendix section). The purpose of this \virg{stylistic} choice relies in the hope that this work can be  preparatory for further investigations. In fact, if on the one hand the description of the $K$-theoretical picture for general 
  multi-interface magnetic systems has been developed in a comprehensive form in this work, on the other hand the proof of the bulk-interface duality has been obtained only for the special case of the Iwatsuka magnetic field (other than that the trivial case of a \emph{localized} magnetic field). The study of more general interfaces will require a deeper and more specific analysis and probably will benefit from the \virg{pioneering} results described in \cite{thiang-20}.

\medskip
\noindent
{\bf Structure of the paper and overview of the results.}
{\bf Section \ref{sect:tb-magn}} is devoted to the construction of the \emph{magnetic $C^*$-algebra} associated to a generic magnetic field $B:\Z^2\to\R$ and to the study of the related  integro-differential  structure. In more details, in 
{\bf Section \ref{sect:MF&VP}} we introduce the notion of \emph{vector potential} $A_B$ associated to $B$, in {\bf Section \ref{sect:gen_MT}} we describe the \emph{magnetic translations} and the \emph{flux operator} associated to $A_B$ and in {\bf Section \ref{sect:magnetic_C_Al}} we define the magnetic $C^*$-algebra $\s{A}_{A_B}$ generated by the magnetic translations.
{\bf Section \ref{sect:four_theo}} describes the \emph{Fourier theory} for the algebra $\s{A}_{A_B}$ which is the natural
adaptation of the classical Fourier theory for tori.
{\bf Section \ref{sect:spa_deriv}} deals with the differential structure of $\s{A}_{A_B}$ and the description of its \emph{smooth} pre-$C^*$-algebra $\s{A}_{A_B}^\infty$. The new concept of the \emph{magnetic hull} $\Omega_B$ and the associated dynamical system are investigated in  
{\bf Section \ref{sect:integr}} along with the construction
of  possible \emph{integrals} (or traces) for $\s{A}_{A_B}$.
{\bf Section \ref{sect:NC-sobolev}} briefly introduces the notion of \emph{noncommutative Sobolev spaces} for $\s{A}_{A_B}$, although this concept is not relevant in this work.
{\bf Section \ref{sec:k-theo_gen_B_over}} contains the 
basic definitions and the main results for the  
 background in which the $K$-theory of the 
magnetic $C^*$-algebra $\s{A}_{A_B}$ can be studied.
The notions of the \emph{evaluation homomorphism}  and  the \emph{interface algebra} $\s{I}\subset \s{A}_{A_B}$ are introduced in {\bf Section \ref{sect:asym_ses}}. In {\bf Section \ref{sect:toep_inter}} it is shown that the {evaluation homomorphism} and the {interface algebra} fit into a short exact sequence (\cf Theorem \ref{theo:SES}) which is reminiscent of the classical 
\emph{Toeplitz extension} for the torus algebra.
In {\bf Section \ref{sect:dim_ev}} the  interconnection between
the structure at infinity of the dynamical system generated by the magnetic hull $\Omega_B$ and the existence of suitable 
evaluation homomorphisms is investigated (\cf Proposition \ref{prop:cond_ev_hom}). This analysis paves the way for the rigorous definition of \emph{magnetic multi-interface} and the associated \emph{bulk algebra} $\s{A}_{\rm bulk}$ (\cf Definition \ref{def:bulk_alg}). {\bf Section \ref{sect:K-theory}} contains the description of the \emph{six-term sequence} in $K$-theory for a 
magnetic multi-interface system described by the magnetic 
$C^*$-algebra $\s{A}_{A_B}$, the interface algebra $\s{I}$ and the 
{bulk algebra} $\s{A}_{\rm bulk}$. 
Finally, {\bf 
Section \ref{sec:Bulk_interface_currents}} provides the notion of 
\emph{bulk and interface currents} and 
the crucial formula 
for the \emph{bulk-interface correspondence}  (equation \eqref{ewq:int_condIIII}).
{\bf Section \ref{sect:iwatsuka}} contains the detailed study of the magnetic interface generated by the \emph{Iwatsuka magnetic field} and the proof of formula \eqref{eq:int_01}.
In more detail,  {\bf Section \ref{sect:toep_iwatsuka_MT}}
contains the presentation of the Toeplitz extension associated to the Iwatsuka magnetic field along with a useful characterization of the
{evaluation homomorphism} (\cf Remark \ref{rk:gen_lim}). 
{\bf Section \ref{sect:interf_iwatsuka_MT}}
 contains a precise description of  the {interface algebra} in terms of a relevant projection  (\cf Proposition \ref{pro_intAlg_iwa}) and {\bf Section \ref{sect:split_Toeplitz}}
provides a short discussion about the (non) splitting property of the  Toeplitz extension for the Iwatsuka magnetic field. The $K$-theory of the Iwatsuka $C^*$-algebra $\s{A}_{\rm I}$ is described in full detail in {\bf Section \ref{sect:K-Iwatsuka}} (\cf Theorem \ref{prop:PV_2}) and the six-term exact sequence 
associated to the Toeplitz extension of  $\s{A}_{\rm I}$ is completely solved in  {\bf Section \ref{sect:6t-K-Iwatsuka}}. This fact provides a  finer understanding of the $K$-theory of $\s{A}_{\rm I}$ (\cf Theorem \ref{prop:PV_2_II}). Finally, the 
{bulk-interface correspondence for the Iwatsuka $C^*$-algebra}
and the proof of the formula \eqref{eq:int_01} are contained in 
{\bf Section \ref{sect:bulk-int-Iwatsuka}} (\cf Theorem \ref{theo_main_Iw}).
Let us emphasize once again that this work  has been written with the aim of being pedagogical and self-consistent (as far as possible). For this reason a significant amount of supporting material has been included in the {\bf Appendices \ref{app:dis_Swar}, \ref{app:bloch-floquet}, \ref{sect:crossedproduct}, \ref{app:K-theoIMH}, \ref{app:PV}, \ref{app:Ktheo NCT}}. The experienced reader can certainly skip this part. Also, the \virg{study case} of a localized magnetic field $B_\Lambda$ has been discussed throughout the work in various examples as an aid to better fix the proposed concepts. As a pay-off we got the generalization of certain results described in \cite{denittis-schulz-baldes-16}.

\medskip
\noindent
{\bf Acknowledgements.} 
This research is supported
 by
the  grant \emph{Fondecyt Regular} -  1190204.
GD wants to thank Fran\c{c}ois Germinet for the suggestion to look at this problem received several years ago.

\section{The (tight-binding) magnetic \texorpdfstring{$C^*$-}-algebra}\label{sect:tb-magn}
In this section we will describe the algebra of operators on $\ell^2(\Z^2)$ that describe the kinematics of charged particles subjected to a generic orthogonal  magnetic field in the tight-binding approximation. This algebra, called \emph{magnetic $C^*$-algebra}, is a suitable generalization of the \emph{noncommutative torus} \cite{connes-94,davidson-96,gracia-varilly-figueroa-01}, and has the structure of a \emph{crossed product $C^*$-algebra} \cite{pedersen-79,connes-94} (\cf Appendix \ref{sect:crossedproduct}).

\subsection{Magnetic fields and vector potentials}\label{sect:MF&VP}
In the tight-binding approximation the (two-dimensional) position space is $\Z^2$.
An  {orthogonal} magnetic field is described by a function $B:\Z^2\to\R$. However, to work with magnetic fields   
we need to introduce  the notion of  vector potential in the discrete setting.
For that, let us fix first some notation. We will denote by $n:=(n_1,n_2)$
the generic point of $\Z^2$ and with $e_1:=(1,0)$ and $e_2:=(0,1)$ 
  the  canonical generators of $\Z^2$. Moreover, we will  view $\Z^2$ as the vertices of an oriented graph, with oriented edges given by the oriented line segments $(n-e_j, n)$ between nearest vertices.
\begin{definition}[Tight-binding vector potential]
Let $B:\Z^2\to\R$ be a magnetic field. A magnetic potential for $B$ is a real-valued function 
$$
A_B\;:\;\Z^2\times\Z^2\;\longrightarrow\;\R
$$ 
satisfying: 
\begin{itemize}
\item[(i)] $A_B(n,m)=0$ for $n,m\in\Z^2$ such that $|n-m| \neq1$;
\vspace{1mm}
\item[(ii)] $A_B(m,n)= -A_B(n,m)$ for all $n,m\in\Z^2$;
\vspace{1mm}
\item[(iii)] $B(n)=\rr{C}[A_B](n)$ for all $n\in\Z^2$ where
$$
\begin{aligned}
\rr{C}[A_B](n)\;:=\;&A_B(n,n-e_1)\;+\;A_B(n-e_1,n-e_1-e_2)\\
&+\;A_B(n-e_1-e_2,n-e_2)\;+\;A_B(n-e_2,n)
\end{aligned}
$$
is the (counterclockwise) \emph{circulation} of $A_B$ along the boundary of the unit cell of $\Z^2$ with right upper corner centered in $n$. 
\end{itemize}
 \end{definition}

\medskip

From \cite[Proposition 1]{denittis-schulz-baldes-16} we know that every 
 magnetic field $B$ admits infinite (gauge) equivalent vector potentials and every two potentials $A_B$ and $A'_B$ associated with the same magnetic field $B$ are related by a 
 \emph{gauge function} $G:\Z^2\to\R$ according to 
 \begin{equation}\label{eq:gaug_tr}
 A'_B(n,m)\;=\;A_B(n,m)\;+\;G(n)\;-\;G(m)\;,\qquad|n-m|=1\;.
\end{equation}
In this case a simple computation shows that $\rr{C}[A_B]=\rr{C}[A_B']$.

\medskip

The following examples will be 
will be considered throughout the work.

\begin{example}[Constant magnetic field]\label{Ex1:const_B}
A \emph{constant} magnetic field of strength $b$ is described by the function
\begin{equation}\label{eq:const_B}
B_{b}(n)\;:=\;b\;\qquad\quad \forall\;n\in\Z^2\;.
\end{equation}
Among the infinite vector potentials associated to the constant magnetic field $B_{b}$, there are two of special utility in practical applications. The first one is the so-called \emph{Landau potential} defined by
\begin{equation}\label{eq:const_A_0}
A^{\rm L}_b(n,n-e_j)\;:=\;\delta_{j,1}\;n_2\; b\;,\qquad \forall\;n\in\Z^2\;.
\end{equation}
The second is the 
 \emph{symmetric  potential}  defined by 
\begin{equation}\label{eq:watsuka_A_s-1}
A^{\rm s}_{b}(n,n-e_j)\;:=\;\left({\delta_{j,1}}\;n_2\;-{\delta_{j,2}}\;n_1\right)  \frac{b}{2}
\;,\qquad \forall\;n\in\Z^2\;.
\end{equation}
A straightforward computation shows that 
$$
\rr{C}[A^{\rm L}_b]\;=\;\rr{C}[A^{\rm s}_b]\;=\;B_b\;.
$$
Moreover, one can check that
$$
A^{\rm s}_b(n,n-e_j)\;=\;A^{\rm L}_b(n,n-e_j)\;+\;G_b(n)\;-\;G_b(n-e_j)
$$
where the gauge function is given by $G_b(n):=-n_1n_2\frac{b}{2}$.\hfill $\blacktriangleleft$
\end{example}
\begin{example}[Iwatsuka magnetic field]\label{Ex2:Iwatsuka}
The \emph{Iwatsuka magnetic field} \cite{iwatsuka-85} models 
 systems with a one-dimensional interface between two constant magnetic fields. In the tight-binding approximation it is defined by
\begin{equation}\label{eq:watsuka_B}
B_{\rm I}(n)\;:=\;b_-\;\delta_-(n)\;+\;b_0\;\delta_0(n)\;+\;
b_+\;\delta_+(n)\;,\qquad \forall\;n\in\Z^2
\end{equation}
where $b_-,b_0,b_+\in\R$ are real constants
and the
 three functions  $\delta_\pm,\delta_0$ are defined by
$$
\delta_\pm(n_1,n_2)\;:=\;
\left\{
\begin{aligned}
&1&\quad&\text{if}\; \pm n_1>0\\
&0&\quad&\text{otherwise}\\
\end{aligned}
\right.\;,
\quad\delta_0(n_1,n_2)\;:=\;
\left\{
\begin{aligned}
&1&\quad&\text{if}\; n_1=0\\
&0&\quad&\text{otherwise}\\
\end{aligned}
\right.\;.
$$
The magnetic field $B_{\rm I}$ describes a one-dimensional interface along the vertical line defined by $n_1=0$. Sometimes, it could be convenient to impose the \emph{monotonic} condition
$$
\min\{b_-,b_+\}\;\leqslant\;b_0\;\leqslant\;\max\{b_-,b_+\}\;.
$$
The simplest choice for a vector potential producing the magnetic field $B_{\rm I}$ is the so-called \emph{Landau-Iwatsuka potential} defined by
\begin{equation}\label{eq:watsuka_A_h}
A_{\rm I}(n,n-e_j)\;:=\;\delta_{j,1}\;n_2\;  B_{\rm I}(n)\;,\qquad \forall\; n\in\Z^2\;.
\end{equation}
A direct computation shows that 
$$
\rr{C}[A_{\rm I}](n)\;=\;n_2 B_{\rm I}(n)\;-\;(n_2-1)B_{\rm I}(n-e_2)\;=\;B_{\rm I}(n)
$$
in view of the fact that $B_{\rm I}(n-e_2)=B_{\rm I}(n)$ for all $n\in\Z^2$. Another interesting choice  is the 
\emph{symmetric Iwatsuka potential}
defined by
\begin{equation}\label{eq:watsuka_A_s}
A^{\rm s}_{\rm I}(n,n-e_j)\;=\;A^{{\rm s},0}_{\rm I}(n,n-e_j)\;+\;\Delta^{{\rm s}}_{\rm I}(n,n-e_j)
\end{equation}
where
$$
A^{{\rm s},0}_{\rm I}(n,n-e_j)\;:=\;\left(\frac{\delta_{j,1}}{2}\;n_2\;-\frac{\delta_{j,2}}{2}\;n_1\right)  B_{\rm I}(n)\;,\qquad \forall\;n\in\Z^2
$$
and 
$$
\Delta^{{\rm s}}_{\rm I}(n,n-e_j)\;:=\;\delta_{j,1}\frac{b_0-b_-}{2}\; n_2\delta^{(1)}_0(n) \;,\qquad \forall\;n\in\Z^2\;.
$$
Since
$\rr{C}[A^{{\rm s},0}_{\rm I}](n)=B_{\rm I}(n)$ only if $n_1\neq0$ and
$
\rr{C}[A^{{\rm s},0}_{\rm I}](0,n_2)=(b_0+b_-)\frac{1}{2}
$,
 the term $\Delta^{{\rm s}}_{\rm I}$ is needed to correct the mismatch. Notice that $\Delta^{{\rm s}}_{\rm I}=0$ when $b_0=b_-$. The symmetric Iwatsuka potential is related to  the vector potential \eqref{eq:watsuka_A_h} by the relation
$$
A^{{\rm s}}_{\rm I}(n,n-e_j)\;=\;A_{\rm I}(n,n-e_j)\;+\;G_{\rm I}(n)\;-\;G_{\rm I}(n-e_j)
$$ with 
gauge transformation 
 given by 
$
G_{\rm I}(n):=-\frac{n_1n_2}{2}B_{\rm I}(n) 
$.
Even though $A_{\rm I}$ and 
$A^{{\rm s}}_{\rm I}$ are equivalent potentials in view of the 
gauge transformation $G_{\rm I}$, the use of the 
Landau-Iwatsuka potential  \eqref{eq:watsuka_B}
will be preferred in this work.
Notice that in the case $b_-=b_0=b_+$ the Iwatsuka magnetic field reduces to the constant magnetic field described in Example \ref{Ex1:const_B}. \hfill $\blacktriangleleft$
\end{example}
\begin{example}[Localized magnetic field]\label{Ex3:loc}
Let $\s{P}_0(\Z^2)$ be the collection of bounded subsets of $\Z^2$. For every $\Lambda\in \s{P}_0(\Z^2)$ let $\delta_\Lambda$ be the characteristic function defined by 
$$
\delta_\Lambda(n)\;:=\;
\left\{
\begin{aligned}
&1&\quad&\text{if}\; n\in\Lambda\\
&0&\quad&\text{otherwise}\\
\end{aligned}
\right.\;.
$$
A \emph{localized} magnetic field of strength $b\in\R$ is described by the function
\begin{equation}\label{eq:loc_B}
B_{\Lambda}(n)\;:=\;b\;\delta_\Lambda(n)\;,\qquad \forall\;n\in\Z^2\;.
\end{equation}
Observe that 
$$
B_{\Lambda}\;=\;\sum_{\lambda\in\Lambda}B_{\{\lambda\}}
$$
where $B_{\{\lambda\}}:=b\delta_{\{\lambda\}}$ is the magnetic field localized on the singleton $\{\lambda\}\in \s{P}_0(\Z^2)$.
A simple vector potential for $B_{\{\lambda\}}$ is provided by the so-called \emph{half-line} potential
$$
A_{\{\lambda\}}(n,n-e_j)\;:=\;b\;\delta_{j,1}\; \sum_{t=0}^{+\infty}\delta_{\{\lambda+ t e_2\}}(n)\;.
$$
By linearity of the circulation one gets that 
$$
A_{\Lambda}\;:=\;\sum_{\lambda\in\Lambda}A_{\{\lambda\}}
$$
provides a  vector potential for the localized magnetic field 
$B_{\Lambda}$. Observe that $A_{\Lambda}$ is well defined in view of the  finiteness of the sum in $\lambda$. \hfill $\blacktriangleleft$
\end{example}

\subsection{The magnetic translations}\label{sect:gen_MT}
Let $\rr{s}_1$ and $\rr{s}_2$  be  the canonical \emph{shift operators} defined on the Hilbert space $\ell^2(\Z^2)$  by
$$
(\rr{s}_j\psi)(n)\;:=\;\psi(n-e_j)\;,\qquad j=1,2\;,\quad \psi\in \ell^2(\Z^2)\;.
$$
Let $B$ a magnetic field with associated vector potential $A_B$. The \emph{magnetic phases} in  the gauge  $A_B$ are the unitary operators $\rr{y}_{A_B,1}$ and $\rr{y}_{A_B,2}$ defined by
$$
(\rr{y}_{A_B,j}\psi)(n)\;:=\;\expo{\ii {A_B}(n,n-e_j)}\psi(n)\;,\quad \psi\in \ell^2(\Z^2)\;.
$$
The \emph{magnetic translations} (in the gauge  ${A_B}$) are the \virg{twisted} shift operators defined by $\rr{s}_{A_B,j}:=\rr{y}_{A_B,j}\;\rr{s}_j $, or more explicitly as
$$
(\rr{s}_{A_B,j}\psi)(n)\;:=\;\expo{\ii {A_B}(n,n-e_j)}\psi(n-e_j)\;, \quad \psi\in \ell^2(\Z^2)\;.
$$
The magnetic translations are composition of unitary operators, hence one has that   $\rr{s}_{A_B,j}^*=\rr{s}_{A_B,j}^{-1}$ with adjoint given by
$$
\big(\rr{s}_{A_B,j}^{\ast}\psi\big)(n)\;:=\;\expo{\ii {A_B}(n,n+e_j)}\psi(n+e_j)\;,\quad \psi\in \ell^2(\Z^2)\;.
$$
A direct computation shows that
\begin{equation}\label{eqgen_comm_rel_X}
\rr{s}_{A_B,1}\;\rr{s}_{A_B,2}\;\rr{s}_{A_B,1}^{\ast}\;\rr{s}_{A_B,2}^{\ast}\;=\;\rr{f}_B
\end{equation}
where the \emph{flux operator} $\rr{f}_B$ acts as
$$
(\rr{f}_B\psi)(n)\;:=\;\expo{\ii \rr{C}[A_B](n)}\psi(n)\;=\;\expo{\ii B(n)}\psi(n)\;,\qquad \; \psi\in \ell^2(\Z^2)\;.
$$
In the case $A_B=0$ the magnetic translations reduce to the {shift operators} and the {flux operator} becomes the identity ${\bf 1}$.

\medskip

It is worth noting that the flux operator $\rr{f}_B$ only depends on the magnetic field $B=\rr{C}[A_B]$ and not on the specific choice of the potential $A_B$. Let $A'_B$ be a second vector potential for  $B$ related to $A_B$ by the gauge transform $G$ according to \eqref{eq:gaug_tr}. Then, it is straightforward to check that the magnetic translations associated to $A_B$ and $A'_B$ are related by the unitary equivalence 
$$
\rr{s}_{A_B',j}\; =\;\expo{-\ii G}\; \rr{s}_{A_B,j}\;\expo{\ii G}\;,\qquad j=1,2\;.
$$ Therefore, the change of the gauge translates to a unitary equivalence at levels of the magnetic translations.

\begin{example}[Magnetic translations for a constant magnetic field]\label{Ex1:MTconst_B}
In the case of a constant magnetic field of strength $b$ the associated magnetic translations (in the Landau gauge) are given by
\begin{equation}\label{eq:MT_cost}
\left\{
\begin{aligned}
(\rr{s}_{b,1}\psi)(n)\;&:=\;\expo{\ii n_2b}\psi(n-e_1)\\
(\rr{s}_{b,2}\psi)(n)\;&:=\;\psi(n-e_2)\;
\end{aligned}
\right.\;,
\qquad \psi\in \ell^2(\Z^2)\;.
\end{equation}
Evidently,  $\rr{s}_{b,2}=\rr{s}_{2}$ coincides with the standard shift operator in the direction $e_2$. In this specific case equation \eqref{eqgen_comm_rel_X} can be rewritten as 
\begin{equation}\label{eq:cost_comm_rel}
\rr{s}_{b,1}\;\rr{s}_{b,2}\;=\;\expo{\ii b}\;\rr{s}_{b,2}\;\rr{s}_{b,1}\;
\end{equation}
and the flux operator coincides with the multiplication by the constant phase $\expo{\ii b}$.
Equation \eqref{eq:cost_comm_rel} provides  the canonical commutation rule for the \emph{noncommutative torus} (\cf \cite[Chapter 12]{gracia-varilly-figueroa-01}).
\hfill $\blacktriangleleft$
\end{example}
\begin{example}[Iwatsuka magnetic translations]\label{Ex2:MT_Iwatsuka}
The magnetic translations associated to a Iwatsuka magnetic field \eqref{eq:watsuka_B} can be easily described in the 
{Landau-Iwatsuka potential} \eqref{eq:watsuka_A_h}. With this choice one obtains the two twisted translations
\begin{equation}\label{eq:MT_iwa}
\left\{
\begin{aligned}
(\rr{s}_{{\rm I},1}\psi)(n)\;&:=\;\expo{\ii n_2B_{\rm I}(n) }\psi(n-e_1)\\
(\rr{s}_{{\rm I},2}\psi)(n)\;&:=\;\psi(n-e_2)
\end{aligned}
\right.\,,
\qquad \psi\in \ell^2(\Z^2)\;.
\end{equation}
Exactly as in the case of the constant magnetic field, one has that $\rr{s}_{{\rm I},2}=\rr{s}_{2}$ coincides with the standard shift operator in the direction $e_2$. Therefore, the difference between \eqref{eq:MT_cost} and \eqref{eq:MT_iwa} lies  totally in the translation along $e_1$. The 
flux operator
\begin{equation}\label{eq:flux_op}
(\rr{f}_{\rm I}\psi)(n_1,n_2)\;:=\;
\left\{
\begin{aligned}
&\expo{\ii b_-}\psi(n_1,n_2)&\quad&\text{if}\quad n_1<0\\
&\expo{\ii b_0}\psi(n_1,n_2)&\quad&\text{if}\quad n_1=0\\
&\expo{\ii b_+}\psi(n_1,n_2)&\quad&\text{if}\quad n_1>0\;
\end{aligned}
\right.
\end{equation}
implements the commutation relation 
\eqref{eqgen_comm_rel_X} for the 
Iwatsuka magnetic  translations.
\hfill $\blacktriangleleft$
\end{example}
\begin{example}[Magnetic translations for a localized field]\label{Ex3:MT_loc}
The magnetic translations associated to a localized magnetic field $B_{\Lambda}$ of the type \eqref{eq:loc_B} can be defined exactly as in the previous examples by using the vector potential 
$A_\Lambda$. Also in this case one obtains that the magnetic translation in the direction $e_2$ coincides with the the standard shift operator.
An explicit computation shows that in this case the flux operator is given by 
$$
\rr{f}_{\Lambda}\;:=\;(\expo{\ii b}-1)\rr{p}_{\Lambda}\;+\;{\bf 1},
$$
where the projection $\rr{p}_{\Lambda}$ is 
the multiplication operator by the function 
$\delta_{\Lambda}$  defined in Example \ref{Ex3:loc}.
\hfill $\blacktriangleleft$
\end{example}

\medskip

Let $\ell^\infty(\Z^2)$ be the von Neumann algebra of  {bounded}  sequences\footnote{\label{note:1}Due to the discreteness of $\Z^2$ every function on $\Z^2$ is automatically uniformly continuous. This provides the identification $\s{C}_{\rm b}(\Z^2)\equiv\ell^\infty(\Z^2)$ where the symbol $\s{C}_{\rm b}(X)$ denotes the algebra of continuous bounded functions on $X$.}  on $\Z^2$. 
The group $\Z^2$ acts on 
$\ell^\infty(\Z^2)$ by translations as follows: Let $\gamma=(\gamma_1,\gamma_2)\in\Z^2$ and $g\in\ell^\infty(\Z^2)$.
Then $\tau_\gamma(g)\in\ell^\infty(\Z^2)$ is defined by 
\begin{equation}\label{eq:tau_act}
\tau_\gamma(g)(n):=g(n-\gamma_1e_1-\gamma_2e_2)\;,\qquad \quad \forall\;n\in\Z^2\;.
\end{equation}
The map $\gamma\mapsto \tau_\gamma$ defines a representation of $\Z^2$ by automorphisms of $\ell^\infty(\Z^2)$.

\medskip

Let $\s{B}(\ell^2(\Z^2))$ be the algebra of bounded operators on $\ell^2(\Z^2)$. 
The algebra $\ell^\infty(\Z^2)$ is canonically identified with the von Neumann sub-algebra $\s{M}$ of bounded multiplication operators on $\ell^2(\Z^2)$. More precisely one has that
$$
\s{M}\;:=\;\left\{\rr{m}_g\; |\;g\in \ell^\infty(\Z^2)\right\}
$$
where the {multiplication operator} $\rr{m}_g\in\s{B}(\ell^2(\Z^2))$ is defined by 
$$
(\rr{m}_g\psi)(n)\;:=\;g(n)\psi(n)\;,\qquad \; \psi\in \ell^2(\Z^2)\;.
$$  
Observe that 
the  {magnetic phases}  $\rr{y}_{A_B,j}$ and the {flux operator} $\rr{f}_B$  are elements of $\s{M}$. In particular one has that  $\rr{f}_B=\rr{m}_{\expo{\ii B}}$.
 The following result follows from a direct computation:
\begin{lemma}\label{lemma:comput_01}
For all $\gamma\in\Z^2$ and $g\in\ell^\infty(\Z^2)$ it holds true that
$$
\begin{aligned}
\rr{m}_{\tau_\gamma(g)}\;&=\;(\rr{s}_{A_B,1})^{\gamma_1}(\rr{s}_{A_B,2})^{\gamma_2}\;\rr{m}_g\;
(\rr{s}_{A_B,2})^{-\gamma_2}(\rr{s}_{A_B,1})^{-\gamma_1}\\
&=\;(\rr{s}_{A_B,2})^{\gamma_2}(\rr{s}_{A_B,1})^{\gamma_1}\;\rr{m}_g\;(\rr{s}_{A_B,1})^{-\gamma_1}
(\rr{s}_{A_B,2})^{-\gamma_2}
\end{aligned}
$$
independently of the vector potential $A_B$.
\end{lemma}

\medskip

Lemma \ref{lemma:comput_01} states that, independently of the magnetic field,  the magnetic translations can be used to implement the \emph{induced} action (still denoted with the same symbol) 
$$
\Z^2\;\ni\;\gamma\;\longmapsto\;\tau_\gamma\;\in\;{\rm Aut}(\s{M})
$$ 
defined by
$$
\tau_\gamma(\rr{m}_g)\;:=\;\rr{m}_{\tau_{\gamma}(g)}
$$
for every $g\in \ell^\infty(\Z^2)$.

\subsection{Construction of the  magnetic \texorpdfstring{$C^*$-}-algebra}
\label{sect:magnetic_C_Al}
Let  $\rr{s}_{A_B,1}$ and $\rr{s}_{A_B,2}$ be the magnetic translations associated
to the magnetic field  $B$  through the vector potential $A_B$. 
\begin{definition}[The tight-binding magnetic $C^*$-algebra]\label{def:magnetic-algebra}
The \emph{magnetic $C^*$-algebra} $\s{A}_{A_B}$ of the magnetic field $B$
(in the gauge $A_B$) is the unital $C^*$-subalgebra of $\s{B}(\ell^2(\Z^2))$ generated by $\rr{s}_{A_B,1}$ and $\rr{s}_{A_B,2}$, \ie
$$
\s{A}_{A_B}\;:=\;C^*\left(\rr{s}_{A_B,1}\;,\;\rr{s}_{A_B,2}\right)\;.
$$
\end{definition}

\medskip

In more detail  $\s{A}_{A_B}$ is  constructed by closing  the collection of the Laurent polynomials 
in $\rr{s}_{A_B,1}$ and $\rr{s}_{A_B,2}$ with respect to the operator norm. 
Since $\s{A}_{A_B}$ is unital we will use the convention $(\rr{s}_{A_B,})^0={\bf 1}$, $j=1,2$.

\medskip

\begin{example}[Noncommutative torus, Iwatsuka algebra and flux tubes]\label{ex_B=0}
In the case of a constant magnetic field of strength $b$ the associated magnetic $C^*$-algebra will be denoted simply with
$$
\s{A}_{b}\;:=\;C^*\left(\rr{s}_{b,1},\rr{s}_{b,2}\right)
$$
where the magnetic translations $\rr{s}_{b,1}$ and $\rr{s}_{b,2}$ are described in Example \ref{Ex1:MTconst_B}. In the case of a vanishing magnetic field ($b=0$) the vector potential can be chosen to be constantly zero and the magnetic translations reduce to the standard shift operators $\rr{s}_j$ on $\ell^2(\Z^2)$.
The  Fourier transform provides a unitary isomorphism between $\ell^2(\Z^2)$ and
$L^2(\T^2)$ where $\n{T}^2\simeq[0,2\pi)^2$ is the two-dimensional torus. This unitary transform  
maps the shift operators into the multiplication by the phases $\expo{\ii k_j}$, where
$k:=(k_1,k_2)$ are the coordinates of $\n{T}^2$. As a consequence one obtains the isomorphism $\s{A}_{b=0}\simeq \s{C}(\T^2)$ where on the right-hand side one has the $C^*$-algebra of the continuous functions on $\n{T}^2$. In this sense, the 
magnetic $C^*$-algebras $\s{A}_{b}$ are noncommutative deformations of the algebra 
$\s{C}(\T^2)$.
Indeed $\s{A}_{b}$ turns out to be   a faithful representation
of the noncommutative torus $\n{A}_{\theta_b}$ with deformation parameter 
$\theta_b:=b(2\pi)^{-1}$ on the Hilbert space $\ell^2(\Z^2)$ (\cf \cite[Chapter 12]{gracia-varilly-figueroa-01}). The magnetic $C^*$-algebra associated to the Iwatsuka magnetic field $B_{\rm I}$ \eqref {eq:watsuka_B} is defined by
$$
\s{A}_{\rm I}\;:=\;C^*\left(\rr{s}_{{\rm I},1},\rr{s}_{{\rm I},2}\right)
$$
where the magnetic translations $\rr{s}_{{\rm I},1}$ and $\rr{s}_{{\rm I},2}$ are defined in Example \ref{Ex2:MT_Iwatsuka}. We will refer to $\s{A}_{\rm I}$ as the \emph{Iwatsuka $C^*$-algebra}. The magnetic $C^*$-algebra associated to a localized magnetic field is
$$
\s{A}_{\Lambda}\;:=\;C^*\left(\rr{s}_{\Lambda,1},\rr{s}_{\Lambda,2}\right)
$$
where the magnetic translations $\rr{s}_{\Lambda,1}$ and $\rr{s}_{\Lambda,2}$ are defined through the vector potential $A_\Lambda$ described in Example \ref{Ex3:loc}. The special case $\Lambda=\{\lambda\}$ has been considered in 
\cite{denittis-schulz-baldes-16} and,  by analogy, one could refer to $\s{A}_{\Lambda}$ as the (extended) \emph{flux-tube $C^*$-algebra}.
\end{example}

\medskip

The $C^*$-algebra $\s{A}_{A_B}$ contains a relevant commutative $C^*$-subalgebra $\s{F}_B$. From 
\eqref{eqgen_comm_rel_X} we know that the 
flux operator $\rr{f}_B=\rr{m}_{\expo{\ii B}}$
is an element of $\s{A}_{A_B}$. In view of 
Lemma \ref{lemma:comput_01}, also every   translated element $\tau_\gamma(\rr{f}_B):=\rr{m}_{\tau_{-\gamma}(\expo{\ii B})}$, with $\gamma\in\Z^2$, belongs to $\s{A}_{A_B}$. As a result one can define the $C^*$-algebra generated by these elements, \ie
\begin{equation}\label{eq:def_F_B}
\s{F}_B\;:=\;C^*\left(\tau_\gamma(\rr{f}_B)\;,\gamma\in\Z^2\right)\;\subset\; \s{A}_{A_B}\;.
\end{equation}
It turns out that $\s{F}_B$ is a commutative unital $C^*$-algebra. Therefore, in view of the Gelfand isomorphism, one has  $\s{F}_B\simeq\s{C}(\Omega_B)$, where 
$\Omega_B$ is a compact Hausdorff space. We will refer to $\Omega_B$ as the \emph{hull} of the magnetic field $B$ (\cf Section \ref{sect:integr} for a more accurate description of this space). It is worth pointing out that  the  commutative  $C^*$-algebra
$\s{F}_B$ only depends on the magnetic field $B$
and not on the particular vector potential $A_B$. More precisely, if $A_B$ and $A_B'$ are two distinct vector potentials for the same magnetic field $B$, then the $C^*$-algebras 
$\s{A}_{A_B}$ and $\s{A}_{A_B'}$ are in general different, even though they are gauge equivalent  $\s{A}_{A_B'}=\expo{-\ii G}(\s{A}_{A_B})\expo{\ii G}$. Nevertheless $\s{F}_B\subset\s{A}_{A_B}\cap\s{A}_{A_B'}$ since the elements in $\s{F}_B$ are gauge invariant. In this sense the commutative $C^*$-subalgebra $\s{F}_B$ encodes all the information about the magnetic field $B$.

\medskip

A finite monomial in $\s{A}_{A_B}$ is an element of the type 
\begin{equation}\label{eq:mon_01}
\rr{u}\;:=\;(\rr{s}_{A_B,1})^{r_1}\;(\rr{s}_{A_B,2})^{s_1}\;\ldots\;(\rr{s}_{A_B,1})^{r_d}\;(\rr{s}_{A_B,2})^{s_d}
\end{equation}
with $d\in\N$ and $r_1,s_1,\ldots,r_d,s_d\in\Z$. In view of \eqref{eqgen_comm_rel_X} and Lemma \ref{lemma:comput_01} the monomial 
$\rr{u}$ can always be rearranged in  the form
\begin{equation}\label{eq:mon_02}
\rr{u}\;:=\;\rr{f}_{r,s}\;(\rr{s}_{A_B,1})^{r}\;(\rr{s}_{A_B,2})^{s}
\end{equation}
where $r=r_1+\ldots+r_d$, $s=s_1+\ldots +s_d$ and
$\rr{f}_{r,s}$ is a suitable (unitary) element of $\s{F}_B$. From its very definition it follows that $\s{A}_{A_B}$ is linearly generated by the  family of  monomials  \eqref{eq:mon_01}, or equivalently \eqref{eq:mon_02}. This observation allows to define several relevant 
dense $\ast$-subalgebras of $\s{A}_{A_B}$.

\medskip

The first dense $\ast$-subalgebra is denoted with  $\s{A}_{A_B}^0$ and is defined as the  collection of the \emph{finite}  linear combination of   monomials of the type \eqref{eq:mon_01}, or equivalently of type \eqref{eq:mon_02}. We will refer to $\s{A}_{A_B}^0$
 as the subalgebra of the \emph{noncommutative polynomials}.

\medskip

For the second dense $\ast$-subalgebra we need to introduce the operator-valued Schwartz space
$\s{S}(\Z^2,\s{F}_B)$ made by  the \emph{rapidly descending} sequences 
$$
\{\rr{g}_{r,s}\}\;:=\;\left\{\rr{g}_{r,s}\;\big|\;(r,s)\in\Z^2\right\}\;\subset\; \s{F}_B
$$
such that 
\begin{equation}\label{eq:fre_top}
r_k\big(\{\rr{g}_{r,s}\}\big)^2\;:=\;\sup_{(r,s)\in\Z^2}\left(1+r^2+s^2\right)^k\|\rr{g}_{r,s}\|^2\;<\;\infty\;,
\end{equation}
for all $k\in\N_0$. It turns out that the system of seminorms \eqref{eq:fre_top} endows 
$\s{S}(\Z^2,\s{F}_B)$ with the structure of  a Fr\'{e}chet space (\cf Appendix \ref{app:dis_Swar}).
In the (standard) case $\s{F}_B=\C$ we will use the short notation $\s{S}(\Z^2)$ instead of $\s{S}(\Z^2,\C)$.
The \emph{smooth} $\ast$-subalgebra
$\s{A}_{A_B}^\infty$ is then defined as
\begin{equation}\label{eq:schw_spac}
\s{A}_{A_B}^\infty\;:=\; \left\{\rr{a}_{\rr{g}}:=\sum_{(r,s)\in\Z^2}\rr{g}_{r,s}\;(\rr{s}_{A_B,1})^{r}\;(\rr{s}_{A_B,2})^{s}\;\Big|\;\{\rr{g}_{r,s}\}\in\s{S}(\Z^2,\s{F}_B)\right\}\;.
\end{equation}
One can easily check the chain of inclusions $
\s{A}_{A_B}^0\subset\s{A}_{A_B}^\infty\subset\s{A}_{A_B}$. 
\begin{proposition}\label{prop_fre_top}
$\s{A}_{A_B}^\infty$ is a Fr\'{e}chet space with respect to the topology induced by the system of norms
$$
|||\rr{a}_{\rr{g}}|||_k\;:=\;r_k\big(\{\rr{g}_{r,s}\}\big)\;,\qquad k\in\N_0\;.
$$
\end{proposition}
\proof
The map $\rr{g}_{r,s}\mapsto \rr{a}_{\rr{g}}$
which associates to a rapidly descending sequence in $\s{S}(\Z^2,\s{F}_B)$ an element of $\s{A}_{A_B}^\infty$ is a bijection. Surjectivity follows from the definition while injectivity is a consequence of Theorem \ref{theo:four-ces}.
Therefore, the map $\rr{g}_{r,s}\mapsto \rr{a}_{\rr{g}}$ defines an isomorphism of  Fr\'{e}chet spaces once 
$\s{A}_{A_B}^\infty$ is endowed with the topology induced by the norms $|||\cdot |||_k$.
\qed

\subsection{Fourier theory}\label{sect:four_theo}
In the  case $B=0$  the magnetic algebra is isomorphic to $\s{C}(\T^2)$  (\cf Example \ref{ex_B=0}) and for this algebra a very rich Fourier theory is available \cite{katznelson-04,grafakos-14}. In this section we will show that some of the results of the classical 
 Fourier theory extend to the  algebra $\s{A}_{A_B}$. Similar results are also described in \cite[Section VIII.2]{davidson-96}.

\medskip

Let $\rr{n}_1$ and $\rr{n}_2$ be the \emph{position} operators acting on $\ell^2(\Z^2)$ 
by
$$
(\rr{n}_j\psi)(n_1,n_2)\;:=\;n_j\;\psi(n_1,n_2)\;,\qquad j=1,2\;,\quad \psi\in \ell^2(\Z^2)\;.
$$
Let $\theta:=(\theta_1,\theta_2)$ be the coordinates of $\n{T}^2\simeq[0,2\pi)^2$ and consider the unitary operator $\rr{w}_\theta:=\expo{-\ii \theta\cdot \rr{n}}=\expo{-\ii (\theta_1 \rr{n}_1+\theta_2\rr{n}_2)}$  defined by
$$
(\rr{w}_\theta\psi)(n_1,n_2)\;:=\;\expo{-\ii (\theta_1 {n}_1+\theta_2{n}_2)}\;\psi(n_1,n_2)\;,\quad \psi\in \ell^2(\Z^2)\;.
$$
The unitary operators $\rr{w}_\theta$ allow to define a family of \emph{inner} automorphisms $\alpha_\theta:\s{B}(\ell^2(\Z^2))\to \s{B}(\ell^2(\Z^2))$ given by
\begin{equation}\label{eq:gru_act_01}
\alpha_\theta(\rr{a})\;:=\;\rr{w}_\theta\;\rr{a}\;\rr{w}_\theta^*\;,\quad \rr{a}\in \s{B}(\ell^2(\Z^2))\;.
\end{equation}
A  straightforward check shows that the map $\n{T}^2\ni\theta\mapsto\alpha_\theta\in{\rm Aut}(\s{B}(\ell^2(\Z^2)))$ defines an action
of the group $\n{T}^2$ on $\s{B}(\ell^2(\Z^2))$. However, this action is \emph{not} continuous on the whole algebra $\s{B}(\ell^2(\Z^2))$. For that consider the parity operator $\wp$ defined by
$(\wp\psi)(n)=\psi(-n)$. Since $\alpha_\theta(\wp)=\rr{w}_{2\theta}\wp$, it follows that 
$\|\alpha_\theta(\wp)-\wp\|=\|\rr{w}_{2\theta}-{\bf 1}\|=2$ whenever one of $\theta_1$ and   
$\theta_2$ are irrational. Things go differently if the action of $\n{T}^2$ is restricted to $\s{A}_{A_B}$.
\begin{proposition}\label{prop:group_act}
The prescription \eqref{eq:gru_act_01} defines a continuous group action 
$\n{T}^2\ni\theta\mapsto\alpha_\theta\in{\rm Aut}(\s{A}_{A_B})$.
\end{proposition}
\proof
A direct computation shows that 
\begin{equation}\label{eq:gru_act_02}
\alpha_\theta(\rr{g})\;=\;\rr{g}\quad\forall\; \rr{g}\in\s{F}_B\;,
\end{equation}
independently of $\theta=(\theta_1,\theta_2)\in\T^2$, and
\begin{equation}\label{eq:gru_act_02.1}
\alpha_\theta\big((\rr{s}_{A_B,1})^{r}\;(\rr{s}_{A_B,2})^{s}\big)\;=\;\expo{-\ii (r\theta_1 +s\theta_2)}\;(\rr{s}_{A_B,1})^{r}\;(\rr{s}_{A_B,2})^{s}
\end{equation}
for all $(r,s)\in\Z^2$.
The relations \eqref{eq:gru_act_02} and \eqref{eq:gru_act_02.1} along with the definition \eqref{eq:schw_spac} of $\s{A}_{A_B}^\infty$  imply that $\alpha_\theta(\s{A}_{A_B}^\infty)=\s{A}_{A_B}^\infty$ for all $\theta\in\T^2$.
Finally, the density of $\s{A}_{A_B}^\infty$ and the fact that $\alpha_\theta$ is norm-preserving imply that  $\alpha_\theta(\s{A}_{A_B})=\s{A}_{A_B}$, namely $\alpha_\theta\in{\rm Aut}(\s{A}_{A_B})$ for all $\theta\in\T^2$. Let us prove now the continuity of the group action.
Let $\rr{a}=\sum_{(r,s)\in\Z^2}\rr{g}_{r,s}\;(\rr{s}_{A_B,1})^{r}\;(\rr{s}_{A_B,2})^{s}$
according to \eqref{eq:schw_spac}. Then 
$$
\|\alpha_\theta(\rr{a})\;-\;\rr{a}\|\;\leqslant\;\sum_{(r,s)\in\Z^2}\big|\expo{-\ii (r\theta_1 +s\theta_2)}-1\big|\; \|\rr{g}_{r,s}\|\;\leqslant\;2\sum_{(r,s)\in\Z^2} \|\rr{g}_{r,s}\|
$$
and from the dominated convergence theorem (for series) it follows that 
\begin{equation}\label{eq:lim_0_01}
\lim_{\theta\to0}\|\alpha_\theta(\rr{a})\;-\;\rr{a}\|\;\leqslant\;\sum_{(r,s)\in\Z^2}\lim_{\theta\to0}\big|\expo{-\ii (r\theta_1 +s\theta_2)}-1\big|\; \|\rr{g}_{r,s}\|\;=\;0
\end{equation}
for all $\rr{a}\in \s{A}_{A_B}^\infty$. Now, let $\rr{b}\in \s{A}_{A_B}$ be a generic element and $\varepsilon>0$. By density it exists a $\rr{a}\in \s{A}_{A_B}^\infty$ such that $\|\rr{b}-\rr{a}\|<\frac{\varepsilon}{2}$. Moreover,
$$
\begin{aligned}
\|\alpha_\theta(\rr{b})\;-\;\rr{b}\|\;&\leqslant\;\|\alpha_\theta(\rr{a})\;-\;\rr{a}\|\;+\;\|\alpha_\theta(\rr{b}-\rr{a})\;-\;(\rr{b}-\rr{a})\|\\
&<\;\|\alpha_\theta(\rr{a})\;-\;\rr{a}\|+\varepsilon\;.
\end{aligned}
$$
Therefore, from \eqref{eq:lim_0_01} it follows that
$\lim_{\theta\to0}\|\alpha_\theta(\rr{b})\;-\;\rr{b}\|<\varepsilon
$,
independently of $\varepsilon>0$ and for all $\rr{b}\in \s{A}_{A_B}$. This concludes the proof.
\qed

\medskip

Let 
$$
{\rm Inv}_{\T^2}(\s{A}_{A_B})\;:=\;\big\{\rr{a}\in\s{A}_{A_B} \;\big|\; \alpha_\theta(\rr{a})=\rr{a}\;,\;\;\forall\;\theta\in\T^2 \big\}
$$
 be the set of invariant elements of  $\s{A}_{A_B}$. From \eqref{eq:gru_act_02} one gets that $\s{F}_B\subseteq {\rm Inv}_{\T^2}(\s{A}_{A_B})$. The next goal is to characterize ${\rm Inv}_{\T^2}(\s{A}_{A_B})$. For that let us denote with $\dd\mu(\theta):=(2\pi)^{-2}\dd\theta$  the normalized Haar measure on $\T^2$ and consider the averaging
$$
\langle\rr{a}\rangle\;:=\;\int_{\n{T}^2}\dd\mu(\theta)\; \alpha_\theta(\rr{a})\;,\qquad \rr{a}\in \s{A}_{A_B}
$$
where the integral is  meant in the Bochner sense. From the invariance of the Haar measure it follows that $\langle\rr{a}\rangle\in {\rm Inv}_{\T^2}(\s{A}_{A_B})$ by construction. Moreover, 
$\langle\rr{a}\rangle=\rr{a}$ if and only if $\rr{a}\in{\rm Inv}_{\T^2}(\s{A}_{A_B})$.
This means that every element of ${\rm Inv}_{\T^2}(\s{A}_{A_B})$ can be always represented as the
averaging of some element in $\s{A}_{A_B}$. The next result   characterizes the set of invariant elements.
\begin{lemma}\label{lemma:four_01}
It holds true that 
$$
{\rm Inv}_{\T^2}(\s{A}_{A_B})\;=\;\s{F}_B\;.
$$
\end{lemma}
\proof

Since we already know that $\s{F}_B\subseteq {\rm Inv}_{\T^2}(\s{A}_{A_B})$ we only need to prove the opposite inclusion. Since every element in ${\rm Inv}_{\T^2}(\s{A}_{A_B})$ can be always represented as the
averaging of some element in $\s{A}_{A_B}$ it is enough to prove that $\langle\rr{a}\rangle\in \s{F}_B$
for all $\rr{a}\in\s{A}_{A_B}$. Since the map $\rr{a}\mapsto \langle\rr{a}\rangle$ is  continuous, \ie $\|\langle\rr{a}\rangle\|\leqslant\|{\rr{a}}\|$, and $\s{F}_B$ is closed, it is sufficient to prove that the 
averaging of the monomials \eqref{eq:mon_02} takes value in $\s{F}_B$. Based on  \eqref{eq:gru_act_02} and \eqref{eq:gru_act_02.1}, a direct computation shows that
\begin{equation}\label{eq:aux_0010}
\begin{aligned}
\langle\rr{g}\;(\rr{s}_{A_B,1})^{r}\;(\rr{s}_{A_B,2})^{s}\rangle\;&=\;\rr{g}\;(\rr{s}_{A_B,1})^{r}\;(\rr{s}_{A_B,1})^{s}\int_{\n{T}^2}\dd\mu(\theta)\; \expo{-\ii (r\theta_1 +s\theta_2)}\\
&=\;\rr{g}\;\delta_{r,0}\;\delta_{s,0}
\end{aligned}
\end{equation}
for all $\rr{g}\in \s{F}_B$
 and for all $(r,s)\in\Z^2$. This completes the proof.
\qed

\medskip

We are now in position to prove that every element of   $\s{A}_{A_B}$ can be represented as a
Fourier-type series in the generating monomials \eqref{eq:mon_02}. To make precise this statement, we need to introduce some notation.
Given a $\rr{a}\in \s{A}_{A_B}$ let us define the 
$\s{F}_B$-valued coefficients
\begin{equation}\label{eq:oxoco}
\begin{aligned}
\hat{\rr{a}}_{r,s}\;:&=\;\langle\rr{a}(\rr{s}_{A_B,2})^{-s}(\rr{s}_{A_B,1})^{-r}\rangle\\
&=\;\left(\int_{\n{T}^2}\dd\mu(\theta)\; \expo{\ii (r\theta_1 +s\theta_2)}\;\alpha_\theta(\rr{a})\right)(\rr{s}_{A_B,2})^{-s}(\rr{s}_{A_B,1})^{-r}\;.
\end{aligned}
\end{equation}
To every box $\Lambda_N:=[-N,N]^2\cap\Z^2$ with $N\in\N$, we associate the   \emph{Ces\`aro mean}
\begin{equation}\label{eq:oxoco_02}
\sigma_{N}(\rr{a})\;:=\;\sum_{(r,s)\in\Lambda_N}\left(1-\frac{|r|}{N+1}\right)
\left(1-\frac{|s|}{N+1}\right)\hat{\rr{a}}_{r,s}\;(\rr{s}_{A_B,1})^{r}\;(\rr{s}_{A_B,2})^{s}\;.
\end{equation}

\begin{theorem}[Fourier expansion - Ces\`aro mean]\footnote{The proof of this theorem is adapted from \cite[Theorem 5.5.7]{weaver-01}.}\label{theo:four-ces}
For every element $\rr{a}\in \s{A}_{A_B}$ it holds true that 
$$
\lim_{N\to\infty}\|\sigma_N(\rr{a})-\rr{a}\|\;=\;0\;.
$$

\end{theorem}
\proof
By combining \eqref{eq:oxoco} and \eqref{eq:oxoco_02} one gets
$$
\begin{aligned}
\sigma_N(\rr{a})\;=\;\int_{\n{T}^2}\dd\mu(\theta)\; K_N(\theta)\;\alpha_\theta(\rr{a})
\end{aligned}
$$
where
$$
\begin{aligned}
K_N(\theta)\;:&=\;\sum_{(r,s)\in\Lambda_N}\left(1-\frac{|r|}{N+1}\right)
\left(1-\frac{|s|}{N+1}\right)\expo{\ii (r\theta_1 +s\theta_2)}\\
&=\;F_N(\theta_1)\;F_N(\theta_2)
\end{aligned}
$$
and 
$$
F_N(\theta_j)\;:=\;\sum_{k=-N}^N\left(1-\frac{|k|}{N+1}\right)
\expo{\ii k\theta_j}\;=\;\frac{1}{N+1}\left(\frac{\sin\left(N\theta_j+\frac{\theta_j}{2}\right)}{\sin\left(\frac{\theta_j}{2}\right)}\right)^2
$$
is the Fej\'er kernel, with $j=1,2$ \cite[Chapter I, Section 2.5]{katznelson-04} or \cite[Chapter I, Section 3.1.3]{grafakos-14}.
Since $(2\pi)^{-1}\int_0^{2\pi}\dd\theta_jF_N(\theta_j)=1$, and consequently
$\int_{\n{T}^2}\dd\mu(\theta) K_N(\theta)=1$, one gets that
$$
\sigma_N(\rr{a})\;-\;\rr{a}\;=\;\int_{\n{T}^2}\dd\mu(\theta)\; K_N(\theta)\;\big[\alpha_\theta(\rr{a})-\rr{a}\big]\;.
$$
Using the identity $\alpha_\theta(\rr{a})-\rr{a}=\alpha_{(\theta_1,0)}(\alpha_{(0,\theta_2)}(\rr{a})-\rr{a}+\rr{a}-
\alpha_{(-\theta_1,0)}(\rr{a}))$ and the fact that the $\T^2$-action is isometric one gets
$$
\begin{aligned}
\|\sigma_N(\rr{a})\;-\;\rr{a}\|\;\leqslant\;&\int_{0}^{2\pi}\frac{\dd\theta_1}{2\pi}\; F_N(\theta_1)\;\|\alpha_{(\theta_1,0)}(\rr{a})-\rr{a}\|\\
&+\;\int_{0}^{2\pi}\frac{\dd\theta_2}{2\pi}\; F_N(\theta_2)\;\|\alpha_{(0,\theta_2)}(\rr{a})-\rr{a}\|\;.
\end{aligned}
$$
Since 
the functions $f_1(\theta_1):=\|\alpha_{(\theta_1,0)}(\rr{a})-\rr{a}\|$ and $f_2(\theta_2):=\|\alpha_{(0,\theta_2)}(\rr{a})-\rr{a}\|$ are continuous with $f_1(0)=0=f_2(0)$
and the the Fej\'er kernel is a summability kernel \cite[Chapter I, Section 2.2]{katznelson-04}
one obtains that the two integrals on the right go to zero when $N\to \infty$. This concludes the proof.
\qed

\medskip

Theorem \ref{theo:four-ces} states that every element of  $\rr{a}\in \s{A}_{A_B}$ can be approximated by the sequence $\sigma_N(\rr{a})\in \s{A}_{A_B}^0$ obtained from the \virg{Fourier coefficients} $\hat{\rr{a}}_{r,s}$.
It follows that two elements with the same $\s{F}_B$-valued coefficients are identical. 
\begin{corollary}\label{cor_uni_exp}
Let $\rr{a}\in \s{A}_{A_B}$. Then $\rr{a}=0$ if and only if 
$\hat{\rr{a}}_{r,s}=0$ for all $(r,s)\in\Z^2$.
\end{corollary}
\begin{remark}[Ces\`aro vs. uniform convergence]\label{rk:ces_unif}
 By observing that 
$$
K_N(\theta)\;=\;\frac{1}{(N+1)^2}\sum_{n_1=0}^N\sum_{n_2=0}^N D_{(n_1,n_2)}(\theta)\;,
$$
where
$$
 D_{(n_1,n_2)}(\theta)\;:=\;\sum_{(r,s)\in\Lambda_{(n_1,n_2)}}\expo{\ii (r\theta_1 +s\theta_2)}\;=\;\frac{\sin\left(n_1\theta_1+\frac{\theta_1}{2}\right)\sin\left(n_1\theta_2+\frac{\theta_2}{2}\right)}{\sin\left(\frac{\theta_1}{2}\right)\sin\left(\frac{\theta_2}{2}\right)}
$$
is the \emph{Dirichlet kernel} of the rectangular domain $\Lambda_{(n_1,n_2)}:=([-n_1,n_1]\times[-n_2,n_2])\cap\Z^2$, one can rewrite \eqref{eq:oxoco_02} in the form
\begin{equation}\label{eq:ces_01}
\sigma_N(\rr{a})\;=\;\frac{1}{(N+1)^2}\sum_{(n_1,n_2)\in\Lambda_N} S_{(n_1,n_2)}(\rr{a})
\end{equation}
where
$$
\begin{aligned}
S_{(n_1,n_2)}(\rr{a})\;:&=\;\sum_{(r,s)\in\Lambda_{(n_1,n_2)}}\widehat{\rr{a}}_{r,s}\;(\rr{s}_{A_B,1})^{r}\;(\rr{s}_{A_B,2})^{s}\\
\end{aligned}
$$
is the  partial Fourier-type expansion of $\rr{a}$. Therefore, Theorem \ref{theo:four-ces} provides   
a justification of the series representation 
$$
\rr{a}\;\stackrel{\sigma}{=}\;\lim_{(n_1,n_2)\to\infty}S_{(n_1,n_2)}(\rr{a})\;:=\;\sum_{(r,s)\in\Z^2}\widehat{\rr{a}}_{r,s}\;(\rr{s}_{A_B,1})^{r}\;(\rr{s}_{A_B,2})^{s}
$$
where the symbol $\stackrel{\sigma}{=}$ means that the limit must be understood in the sense of 
Ces\`aro, as given by  equation \eqref{eq:ces_01}. This is the best that one can generally hope for a generic element $\rr{a}\in \s{A}_{A_B}$. Indeed, let $f\in\s{C}(\T)$ be the Fej\'er-type function constructed as in \cite[Chapter II, Section 2.1]{katznelson-04}. 
Then, the sequence of the partial Fourier-type expansions of the  element  $f(\rr{s}_1^{A_B})\in \s{A}_{A_B}$ cannot be convergent in norm. 
\hfill $\blacktriangleleft$
\end{remark}

\medskip

It is useful to characterize the collection of elements of  
$\s{A}_{A_B}$ having an absolutely convergent Fourier series of $\s{F}_B$-valued coefficients.
More precisely, let us introduce the space
$$
\s{A}_{A_B}^{\rm a.c.}\;:=\;\left\{\rr{a}\in\s{A}_{A_B}\;\left|\;\|\rr{a}\|_{\ell^1}:=\sum_{(r,s)\in\Z^2}\|\widehat{\rr{a}}_{r,s}\|<\infty\right.
\right\}
$$
where the coefficients $\widehat{\rr{a}}_{r,s}$ are defined by \eqref{eq:oxoco}. Since
$\s{A}_{A_B}^\infty\subset \s{A}_{A_B}^{\rm a.c.}\subset\s{A}_{A_B}$
it follows that $\s{A}_{A_B}^{\rm a.c.}$ is  dense in $\s{A}_{A_B}$. The main properties of 
$\s{A}_{A_B}^{\rm a.c.}$ are described in the next result.
\begin{proposition}
The space $\s{A}_{A_B}^{\rm a.c.}$, endowed with the norm $\|\;\|_{\ell^1}$, is a Banach $\ast$-algebra isomorphic to $\ell^1(\Z^2,\s{F}_B)$. In particular every $\rr{a}\in \s{A}_{A_B}^{\rm a.c.}$ agrees with its Fourier-type expansion, \ie
$$
\rr{a}\;=\;\sum_{(r,s)\in\Z^2}\rr{a}_{r,s}\;(\rr{s}_{A_B,1})^{r}\;(\rr{s}_{A_B,2})^{s}
$$
\end{proposition}
\proof
Every $\rr{a}\in \s{A}_{A_B}^{\rm a.c.}$ defines an element $\{\widehat{\rr{a}}_{r,s}\}\in \ell^1(\Z^2,\s{F}_B)$ by definition. Moreover, the map $\rr{a}\mapsto \{\widehat{\rr{a}}_{r,s}\}$ is injective in view of Corollary \ref{cor_uni_exp}. The surjectivity follows by observing that every $\{\widehat{\rr{a}}_{r,s}\}\in \ell^1(\Z^2,\s{F}_B)$ defines an element 
$$
\rr{a}\;:=\;\lim_{N\to\infty}\sum_{(r,s)\in\Lambda_{N}}\widehat{\rr{a}}_{r,s}\;(\rr{s}_{A_B,1})^{r}\;(\rr{s}_{A_B,2})^{s}\;\in\;\s{A}_{A_B}^{\rm a.c.}\;.
$$
with $\s{F}_B$-valued coefficients $\{\widehat{\rr{a}}_{r,s}\}$.
Finally, a straightforward computation as in \cite[Chapter I, Section 6.1]{katznelson-04} shows that $\s{A}_{A_B}^{\rm a.c.}$ is closed under the operations inherited by the $\ast$-algebraic structure of $\s{A}_{A_B}$.
\qed

\subsection{Spatial derivations and differential structure}\label{sect:spa_deriv}
Proposition  \ref{prop:group_act} can be  reinterpreted
in the jargon of the theory of $C_0$-(semi)groups \cite[Chapter 3]{bratteli-robinson-87}  by saying that  the map $\theta\mapsto \alpha_\theta$ defines a strongly continuous $\n{T}^2$-action 
on the $C^*$-agebra $\s{A}_{A_B}$ by automorphisms \cite[Corollary 3.1.8]{bratteli-robinson-87}. This allows us to introduce the \emph{infinitesimal generators} ${\nabla}_1$ and 
${\nabla}_2$ defined by
$$
\begin{aligned}
{\nabla}_1(\rr{a})\;&:=\;\lim_{\theta_1\to 0}\frac{\alpha_{(\theta_1,0)}(\rr{a})-\rr{a}}{\theta_1}\\
{\nabla}_2(\rr{a})\;&:=\;\lim_{\theta_2\to 0}\frac{\alpha_{(0,\theta_2)}(\rr{a})-\rr{a}}{\theta_2}
\end{aligned}
$$
for suitable elements $\rr{a}\in\s{A}_{A_B}$ \cite[Definition 3.1.5]{bratteli-robinson-87}. Indeed,   ${\nabla}_1$ and 
${\nabla}_2$ are unbounded linear maps on $\s{A}_{A_B}$, defined on dense domains $\s{D}({\nabla}_1)$ and $\s{D}({\nabla}_2)$, respectively \cite[Proposition 3.1.6]{bratteli-robinson-87}. Moreover, they are \emph{(symmetric) derivations} \cite[Definition 3.2.21]{bratteli-robinson-87}, in the sense that 
\begin{equation}\label{eq:leib}
\begin{aligned}
{\nabla}_j(\rr{a}^*)\;&=\;{\nabla}_j(\rr{a})^*;\,\\
{\nabla}_j(\rr{a}\rr{b})\;&=\;\rr{a}\;{\nabla}_j(\rr{b})\;+\; {\nabla}_j(\rr{a})\;\rr{b}\;,
\end{aligned}\qquad \rr{a},\rr{b}\in\s{D}({\nabla}_j)\;,\;\;j=1,2\;.
\end{equation}
Since the subalgebra $\s{F}_B$ is invariant under the action $\alpha_\theta$ it follows that
\begin{equation}\label{eq:deriv_01}
{\nabla}_1(\rr{g})\;=\;{\nabla}_2(\rr{g})\;=\;0\;,\qquad\forall\; \rr{g}\in\s{F}_B\;.
\end{equation}
Moreover, a direct computation shows
\begin{equation}\label{eq:deriv_02}
\begin{aligned}
&{\nabla}_1\big((\rr{s}_{A_B,1})^{r}\;(\rr{s}_{A_B,2})^{s}\big)\;=\;-\ii r\;(\rr{s}_{A_B,1})^{r}\;(\rr{s}_{A_B,2})^{s}&&\\
&{\nabla}_2\big((\rr{s}_{A_B,1})^{r}\;(\rr{s}_{A_B,2})^{s}\big)\;=\;-\ii s\;(\rr{s}_{A_B,1})^{r}\;(\rr{s}_{A_B,2})^{s}\;,&&
\end{aligned}\qquad\forall\; (r,s)\in\Z^2\;.
\end{equation}
In particular, one can check  that
\begin{equation}\label{eq:deriv_022}
{\nabla}_j\big(\rr{g}\;(\rr{s}_{A_B,1})^{r}\;(\rr{s}_{A_B,2})^{s}\big)\;=\;\ii\big[\rr{g}\;(\rr{s}_{A_B,1})^{r}\;(\rr{s}_{A_B,2})^{s},\rr{n}_j\big]\;,
\end{equation}
where $[\;,\;]$ denotes the commutator. Indeed equation \eqref{eq:deriv_022} is a special case
of a more general result \cite[Definition 3.2.55]{bratteli-robinson-87}, which justifies the name of
 \emph{spatial derivation} for 
 ${\nabla}_1$ and 
${\nabla}_2$.

\medskip

From \eqref{eq:deriv_01} and \eqref{eq:deriv_02} it follows that
$$
\s{A}_{A_B}^0\;\subset\;\s{A}_{A_B}^\infty\;\subset\;\s{D}({\nabla}_1)\;\cap\;\s{D}({\nabla}_2)\;.
$$
Moreover, the elements of $\s{A}_{A_B}^\infty$ support several iterated derivations. Let  ${\nabla}_j^a:={\nabla}_j\circ\ldots\circ{\nabla}_j$ be the $a$-times iteration of the derivation ${\nabla}_j$.
Since the group $\n{T}^2$ is abelian, it follows that  ${\nabla}_1\circ{\nabla}_2={\nabla}_2\circ{\nabla}_1$ whenever the product of the derivatives is well defined. It follows that the expression  ${\nabla}_1^a{\nabla}_2^b$, for  $a,b\in\N_0$,
 is not ambiguous in suitable domains  like $\s{A}_{A_B}^0$ or $\s{A}_{A_B}^\infty$. 
Let us introduce the spaces
$$
\s{C}^k(\s{A}_{A_B})\;:=\;\overline{\s{A}_{A_B}^0}^{\;\|\;\|_k}\;,
$$
obtained by closing the noncommutative polynomials with respect to the norm
$$
\|\rr{a}\|_k\;:=\;\sum_{i=0}^k\sum_{a+b=i}\|{\nabla}_1^a{\nabla}_2^b(\rr{a})\|\;.
$$
A standard argument shows that $\rr{a}\in \s{C}^k(\s{A}_{A_B})$ if and only if ${\nabla}_1^a{\nabla}_2^b(\rr{a})\in \s{A}_{A_B}$ is well defined for all $a,b\in\N_0$ such that $a+b\leqslant k$, namely
$$
\s{C}^k(\s{A}_{A_B})\;=\;\left\{\rr{a}\in \s{A}_{A_B}\; \big|\; \theta\mapsto \alpha_\theta(\rr{a})\;\;\;\text{is}\; k\text{-differentiable} \right\}\;.
$$
The regularity of an element is reflected on the decay property of its $\s{F}_B$-valued coefficients. This is the content of the next result.
\begin{lemma}\label{lem:diff-abs}\footnote{
The results provided in
 Lemma \ref{lem:diff-abs} are not optimal, in general. For instance, in the case  of a zero magnetic field
 described in Example \ref{ex_B=0} one can replace \eqref{eq:coef_bound_0} with $(1+r^2+s^2)^k\|\widehat{\rr{a}}_{r,s}\|^2\to0$ when $(r,s)\to\infty$ \cite[Theorem 3.3.9]{grafakos-14}. Moreover, the absolute convergence of the series of coefficients is generally guaranteed by a degree of regularity weaker than $k>2$ \cite[Theorem 3.3.16]{grafakos-14}. However, for the purposes of this work we will not need such a  kind of generalization.}
Let $\rr{a}\in \s{C}^k(\s{A}_{A_B})$ then
\begin{equation}\label{eq:coef_bound_0}
\sup_{(r,s)\in\Z^2}\left(1+r^2+s^2\right)^k\|\widehat{\rr{a}}_{r,s}\|^2\;<\;\infty
\end{equation}
where the $\widehat{\rr{a}}_{r,s}$ are defined  by \eqref{eq:oxoco}. In particular
$$
\s{C}^k(\s{A}_{A_B})\;\subset\; \s{A}_{A_B}^{\rm a.c.}
$$
for all $k>2$.
\end{lemma}
\proof
Let $a,b\in\N_0$ such that $a+b\leqslant k$. Then ${\nabla}_1^a{\nabla}_2^b(\rr{a})\in \s{A}_{A_B}$ and we can calculate the $\s{F}_B$-valued coefficients according to \eqref{eq:oxoco}.
An iterated integration by parts provides
$$
\begin{aligned}
\widehat{{\nabla}_1^a{\nabla}_2^b(\rr{a})}_{r,s}\;&=\;\left(\int_{\n{T}^2}\dd\mu(\theta)\; \expo{\ii (r\theta_1 +s\theta_2)}\;\alpha_\theta({\nabla}_1^a{\nabla}_2^b(\rr{a}))\right)(\rr{s}_{A_B,2})^{-s}(\rr{s}_{A_B,1})^{-r}\\
&=\;(-\ii)^{a+b}r^as^b\left(\int_{\n{T}^2}\dd\mu(\theta)\; \expo{\ii (r\theta_1 +s\theta_2)}\;\alpha_\theta(\rr{a})\right)(\rr{s}_{A_B,2})^{-s}(\rr{s}_{A_B,1})^{-r}\\
&=\;(-\ii)^{a+b}r^as^b\;\widehat{\rr{a}}_{r,s}\;.
\end{aligned}
$$
Since $\|\widehat{\nabla_1^a\nabla_2^b(\rr{a})}_{r,s}\|\leqslant\|\nabla_1^a\nabla_2^b(\rr{a})\|=:C_{a,b}$ for all $(r,s)\in\Z^2$, 
we can define $C:=\max_{a+b=k}\{C_{a,b}\}$. It
follows that
$r^{2a}s^{2b}\|\widehat{\rr{a}}_{r,s}\|^2\leqslant C^2$ for all $a,b$ such that $a+b=k$. Then, by using the   formula for the  binomial expansion
one gets
\begin{equation}\label{eq:coef_bound}
(r^2+s^2)^{k}\|\widehat{\rr{a}}_{r,s}\|^2\;\leqslant\;2^kC^2\;.
\end{equation}
From \eqref{eq:coef_bound}, a second application of the formula for the  binomial expansion provides \eqref{eq:coef_bound_0} with bound given by $4^kC^2$. From \eqref{eq:coef_bound} one gets
$$
\begin{aligned}
\sum_{(r,s)\in\Z^2}\|\widehat{\rr{a}}_{r,s}\|\;&\leqslant\;\|\widehat{\rr{a}}_{0,0}\|\,+\;2^kC^2\sum_{(r,s)\in\Z^2\setminus(0,0)}\frac{1}{\left({r^2+s^2}\right)^{\frac{k}{2}}}\\
&=\;\|\widehat{\rr{a}}_{0,0}\|\,+\;2^{k+1}C^2\left( 2\sum_{r=1}^{+\infty}\sum_{s=1}^{+\infty}\frac{1}{\left({r^2+s^2}\right)^{\frac{k}{2}}}+\sum_{r=1}^{\infty}\frac{1}{r^k}+\sum_{s=1}^{\infty}\frac{1}{s^k}\right)
\\
&=\;\|\widehat{\rr{a}}_{0,0}\|\,+\;2^{k+2}C^2\left( \sum_{r=1}^{+\infty}\sum_{s=1}^{+\infty}\frac{1}{\left({r^2+s^2}\right)^{\frac{k}{2}}}+\sum_{r=1}^{+\infty}\frac{1}{r^k}\right)
\\
&\leq\;\|\widehat{\rr{a}}_{0,0}\|\,+\;2^{\frac{k}{2}+2}C^2\left( \left(\sum_{r=1}^{+\infty}\frac{1}{r^{\frac{k}{2}}}\right)\left(\sum_{s=1}^{+\infty}\frac{1}{s^{\frac{k}{2}}}\right)+\sum_{r=1}^{+\infty}\frac{1}{r^k}\right)
\end{aligned}
$$
where in the last inequality we used $2rs\leqslant r^2+s^2$. This concludes the proof.
\qed

\medskip

The space of the smooth elements is defined by
$$
\s{C}^\infty(\s{A}_{A_B})\;:=\;\bigcap_{k\in\N_0}\s{C}^k(\s{A}_{A_B})\;.
$$
For $\rr{a}\in\s{C}^\infty(\s{A}_{A_B})$ the 
 map $\theta\mapsto \alpha_\theta(a)$ turns out to be  smooth.
\begin{proposition}
The dense subalgebra $\s{A}_{A_B}^\infty$ defined by \eqref{eq:schw_spac} 
coincides with the algebra of the  
smooth elements with respect to the $\n{T}^2$-action, \ie
$$
\s{A}_{A_B}^\infty\;=\;\s{C}^\infty(\s{A}_{A_B})\;.
$$
\end{proposition}
\proof
Let $\rr{a}\in \s{A}_{A_B}^\infty$. Then 
the computation of the $\s{F}_B$-valued coefficients of $\nabla_1^a\nabla_2^b(\rr{a})$ provided in the proof of  Lemma \ref{lem:diff-abs} shows that
$\nabla_1^a\nabla_2^b(\rr{a})\in \s{A}_{A_B}^{\rm a.c.}$ for all $a,b\in\N_0$. This implies that $\s{A}_{A_B}^\infty\subset\s{C}^k(\s{A}_{A_B})$ for all $k\in\N_0$, and so
$\s{A}_{A_B}^\infty\subseteq\s{C}^\infty(\s{A}_{A_B})$. On the other 
hand it is also true that $\s{C}^\infty(\s{A}_{A_B})\subseteq\s{A}_{A_B}^\infty$. In fact, if $\rr{a}\in \s{C}^\infty(\s{A}_{A_B})$ then \eqref{eq:coef_bound_0} applies for all $k\in\N_0$, showing that $\rr{a}\in \s{A}_{A_B}^\infty$. This concludes the proof.
\qed

\medskip
The last result justifies the name of \emph{smooth algebra} for $\s{A}_{A_B}^\infty$.
Let us recall that a  pre-$C^*$-algebra is a dense subalgebra of a $C^*$-algebra which is stable under holomorphic functional calculus (see \cite[Definition 3.26]{gracia-varilly-figueroa-01}). 
\begin{proposition}\label{prop:pre-C-ast}
The smooth algebra $\s{A}_{A_B}^\infty$ defined by \eqref{eq:schw_spac} is a  unital Fréchet  pre-$C^*$-algebra of $\s{A}_{A_B}$.
\end{proposition}
\proof
Since $\n{T}^2$ is a Lie group, the criterion established in \cite[Proposition 3.45]{gracia-varilly-figueroa-01} applies proving the claim.
\qed

\medskip

The  Fréchet topology of  the pre-$C^*$-algebra $\s{A}_{A_B}^\infty$ is provided by the system of norms
described in Proposition \ref{prop_fre_top}.

\subsection{Magnetic hull and integration theory}\label{sect:integr}
The first task of this section is to define the \emph{magnetic hull} by following the construction
 sketched in \cite[Section 2.4]{becku-bellissard-denittis-18}. 
Then, we will construct \emph{traces} on the magnetic algebra by following the ideas of \cite[VIII.3]{davidson-96}.

\medskip

Let $B:\Z^2\to\R$ be a magnetic field and $f_B\in\ell^\infty(\Z^2)$ the function defined by $f_B(n):=\expo{\ii B(n)}$ for all $n\in\Z^2$. 
The natural discrete topology of $\Z^2$ implies that 
$\ell^\infty(\Z^2)= \s{C}_{\rm b}(\Z^2)$ 
as anticipated in Note \ref{note:1}.
The $C^*$-algebra $\s{C}_{\rm b}(\Z^2)$ carries the $\Z^2$-action  defined by 
\eqref{eq:tau_act}. Since elements of $\s{C}_{\rm b}(\Z^2)$  are uniformly continuous one has that
$$
\lim_{\gamma\to0}\|\tau_{\gamma}(g)-g\|_\infty\;=\;\lim_{\gamma\to0}\left(\sup_{n\in\Z^2}|g(n-\gamma)-g(n)|\right)\;=\;0\;,
$$
for all $g\in \s{C}_{\rm b}(\Z^2)$. This means that the $\Z^2$-action $\gamma\mapsto\tau_\gamma$
acts continuously on $\s{C}_{\rm b}(\Z^2)$. It is worth recalling that  the Gelfand-Na\v{\i}mark theorem \cite[Theorem 1.4]{gracia-varilly-figueroa-01} provides the  isomorphism 
$\s{C}_{\rm b}(\Z^2)\simeq\s{C}(\beta\Z^2)$ where $\beta\Z^2$ is the \emph{Stone-\v{C}ech compactification} of $\Z^2$ \cite[Section 1.3]{gracia-varilly-figueroa-01}. In particular, one has a canonical inclusion  $\Z^2\hookrightarrow\beta\Z^2$, which identifies the lattice $\Z^2$  with an open and dense subset of $\beta\Z^2$.

\medskip

Let 
\begin{equation}\label{eq:def_Cfb}
\s{C}(f_B,\Z^2)\;:=\;C^*\left(\tau_\gamma(f_B)\;,\gamma\in\Z^2\right)
\end{equation}
be the $C^*$-subalgebra of $\s{C}_{\rm b}(\Z^2)$ generated by the $\Z^2$-translated of $f_B$ and its complex conjugated.
It turns out that there is an isomorphism $\s{C}(f_B,\Z^2)\simeq \s{F}_B$ under the identification of $\s{C}_{\rm b}(\Z^2)$ with 
the von Neumann algebra $\s{M}$
of bounded  multiplication operators described in Section \ref{sect:gen_MT}.
As anticipated in Section \ref{sect:magnetic_C_Al},
 the Gelfand-Na\v{\i}mark theorem provides the isomorphism   $\s{C}(f_B,\Z^2)\simeq\s{C}(\Omega_B)$
where $\Omega_B$ is  a compact Hausdorff space. Since $\s{C}(f_B,\Z^2)$ is generated by a countable family, it follows that  it is separable (\ie it has a countable and dense subset) and in turn $\Omega_B$ is second countable  and metrizable as a separable complete metric space\footnote{Indeed, $\Omega_B$  is a compact \emph{Polish space}.} \cite[Proposition 1.11]{gracia-varilly-figueroa-01} (see also \cite[Section 2.2]{arveson-76}). We will refer to the  topological space $\Omega_B$ 
as the \emph{hull} of the magnetic field $B$, or the \emph{magnetic hull} for short.

\medskip

Actually, $\Omega_B$ is built as the \emph{Gelfand spectrum} of  $\s{C}(f_B,\Z^2)$, namely the set of characters defined as the $\ast$-homomorphisms $\omega:\s{C}(f_B,\Z^2)\to\C$.
As a consequence, $\Z^2$  acts by duality on $\Omega_B$. For every $\gamma\in\Z^2$ let ${\tau}^*_\gamma:\Omega_B\to\Omega_B$ be the map defined on $\omega\in \Omega_B$ as ${\tau}^*_\gamma(\omega)(g):=\omega(\tau_{-\gamma}(g))$ for all $g\in \s{C}(f_B,\Z^2)$. It is straightforward to show that ${\tau}^*_\gamma\in{\rm Homeo}(\Omega_B)$
are homeomorphisms of $\Omega_B$ and that the mapping $\gamma\mapsto {\tau}^*_\gamma$ provides a continuous $\Z^2$-action by homeomorphisms. As a result $(\Omega_B,{\tau}^*,\Z^2)$ is a genuine
\emph{topological dynamical system} (see \eg \cite[Chapter 5]{walters-82}). In $\Omega_B$ there is a  remarkable point $\omega_0$, called the \emph{evaluation at $0$}, defined by 
$\omega_0(g):=g(0)$ for all $g\in \s{C}(f_B,\Z^2)$. Let $\omega_\gamma:={\tau}^*_\gamma(\omega_0)=\omega_0\circ {\tau}_{-\gamma}$ be the $\gamma$-translated of $\omega_0$ and  ${\rm Orb}(\omega_0):=\{\omega_\gamma\in \Omega_B\;|\; \gamma\in\Z^2\}$ the \emph{$\Z^2$-orbit} of $\omega_0$. The next result provides a relevant property of the dynamical system  $(\Omega_B,{\tau}^*,\Z^2)$.
\begin{proposition}
\label{prop:dens_orb}
The $\Z^2$-orbit of $\omega_0$ is dense, \ie
$$
\overline{{\rm Orb}(\omega_0)}\;=\; \Omega_B\;.
$$
\end{proposition}
\proof
In view of the
Gelfand-Na\v{\i}mark isomorphism 
$\s{C}_{\rm b}(\Z^2)\simeq\s{C}(\beta\Z^2)$,  
the Gelfand spectrum of $\s{C}_{\rm b}(\Z^2)$ can be  identified  
 with the Stone-\v{C}ech compactification $\beta\Z^2$. The inclusion $\jmath: \s{C}(f_B,\Z^2) \hookrightarrow \s{C}_{\rm b}(\Z^2)$ provides, by duality, a continuous map ${\jmath}':\beta\Z^2\to  \Omega_B$ defined by ${\jmath}'(\tilde{\omega}):=\tilde{\omega}\circ \jmath$ where $\tilde{\omega}\in \beta\Z^2$ is a  character of $\s{C}_{\rm b}(\Z^2)$. More precisely, ${\jmath}'(\tilde{\omega})$ is by definition the restriction of the character $\tilde{\omega}$ to the subalgebra $\s{C}(f_B,\Z^2)$. 
On the other hand, every character $\omega$ of $\s{C}(f_B,\Z^2)$ admits an (not necessarily unique) extension $\tilde{\omega}$ to a character of $\s{C}_{\rm b}(\Z^2)$ \cite[Proposition 2.3.24]{bratteli-robinson-87}.  As a result, it turns out that ${\jmath}'$ is a continuous surjection. 
Therefore, if $X\subset\beta\Z^2$ is dense in $\beta\Z^2$ then ${\jmath}'(X)\subset  \Omega_B$ is dense in $\Omega_B$.
In view of the Riesz-Markov-Kakutani representation theorem \cite[Theorem IV.14]{reed-simon1},   the Gelfand spectrum of $\s{C}_{\rm b}(\Z^2)$ consists of the evaluations (Dirac measures) at the points of  $\beta\Z^2$. Since $\Z^2$ can be identified with a dense open subset of $\beta\Z^2$, it follows that the set of characters $\{\tilde{\omega}_\gamma\;|\;\gamma\in\Z^2\}$, defined by
$\tilde{\omega}_\gamma(f):=f(\gamma)$ for $f\in\s{C}_{\rm b}(\Z^2)$, is dense in the  Gelfand spectrum of $\s{C}_{\rm b}(\Z^2)$.  On the other hand, it holds true that $\omega_\gamma={\jmath}'(\tilde{\omega}_\gamma)$, and consequently
$$
{\jmath}'\big(\{\tilde{\omega}_\gamma\;|\;\gamma\in\Z^2\}\big)\;=\;{\rm Orb}(\omega_0)\;.
$$
The last equality proves the  density of ${\rm Orb}(\omega_0)$.
\qed

\medskip

For a given $g\in \s{C}(f_B,\Z^2)$ the Gelfand transform $\hat{g}\in \s{C}(\Omega_B)$ is defined by the equation $\hat{g}(\omega):=\omega(g)$ for all $\omega\in \Omega_B$. The density of the orbit of 
$\omega_0$ implies that the Gelfand transform is entirely defined by the equation $g(\gamma)=\omega_\gamma(g)=\hat{g}({\tau}^*_\gamma(\omega_0))$ for all $\gamma\in\Z^2$.

\begin{remark}[Topological transitivity]\label{rk:top_trans}
Proposition \ref{prop:dens_orb} can be rephrased by saying that the dynamical system $(\Omega_B,{\tau}^*,\Z^2)$ is  \emph{topologically transitive} \cite[Definition 5.6]{walters-82}.  
As a consequence, every invariant element of $\s{C}(\Omega_B)$ is automatically constant 
\cite[Theorem 5.14]{walters-82}.
In our specific setting ($\Omega_B$ compact and second countable) the notion of 
topological transitivity for $(\Omega_B,{\tau}^*,\Z^2)$ is equivalent to the following property: Whenever $U$ and $V$ are nonempty open subsets of $\Omega_B$, then there exists a $\gamma\in\Z^2$ such that ${\tau}^*_{\gamma}(U)\cap V\neq\emptyset$ \cite[Theorem 5.8]{walters-82}.
The latter, is the usual definition of topological transitivity in the context of the general theory of topological dynamical systems (see \eg \cite{kolyada-snoha-97,akin-carlson-12} and references therein). hfill $\blacktriangleleft$
\end{remark}

The subsets ${\rm Orb}(\omega_0)$ and $\partial \Omega_B:=\Omega_B\setminus{\rm Orb}(\omega_0)$ are disjoint and ${\tau}^*$-invariant by construction. Moreover, 
$\partial \Omega_B$ is nowhere dense \cite[Theorem 5.8]{walters-82} and  is contained in the subset of \emph{non-wandering points} of the dynamical system \cite[Theorem 5.6]{walters-82}.
Let ${\rm Mes}_{1,{\tau}^*}( \Omega_B)$ be the set of the normalized and ${\tau}^*$-invariant regular Borel measures\footnote{By the Riesz-Markov-Kakutani representation theorem \cite[Theorem IV.14]{reed-simon1}, ${\rm Mes}_{1,{\tau}^*}( \Omega_B)$
provides the space of  ${\tau}^*$-invariant states of the Abelian $C^*$-algebra $\s{C}(\Omega_B)$.} of the dynamical system $(\Omega_B,{\tau}^*,\Z^2)$.
It is well known that  ${\rm Mes}_{1,{\tau}^*}( \Omega_B)$ is a non-empty, convex and compact set (\ie a Choquet simplex) whose extreme points  are exactly the \emph{ergodic} measures \cite[Corollary 6.9.1 \& Theorem 6.10]{walters-82}. Let ${\rm Erg}( \Omega_B)$ be the subset of the ergodic probability measures of  $(\Omega_B,{\tau}^*,\Z^2)$. It is worth recalling that  an ergodic measure  $\n{P}\in{\rm Erg}( \Omega_B)$  is  characterized by the dichotomy $\n{P}(X)=1$ or $\n{P}(X)=0$ for every  ${\tau}^*$-invariant subset $X\subseteq \Omega_B$.
 A $\n{P}\in{\rm Erg}( \Omega_B)$ such that $\n{P}(\partial \Omega_B)=1$ will be called a \emph{measure at infinity}.

\begin{example}[Magnetic hull for a constant magnetic field]
\label{ex:cost_hull}
In the case of a constant magnetic field of strength $b$ 
one has that $f_b(n):=\expo{\ii b}$ for all $n\in\Z^2$ (see Example \ref{Ex1:const_B}) and accordingly
 $\s{C}(f_b,\Z^2)=\C$. Therefore the associated magnetic hull $\Omega_b\simeq\{\omega_0\}$ is a singleton (or one point set) on which the $\tau^*$-action is trivial. The unique normalized ergodic measure on $\Omega_b$ is entirely specified  by $\n{P}(\{\omega_0\})=1$.\hfill $\blacktriangleleft$
\end{example}

\begin{example}[Iwatsuka magnetic hull]
\label{ex:hull_Iwatsuka}
In the case of the Iwatsuka  magnetic field \eqref{eq:watsuka_B} 
one has 
$$
f_{\rm I}\;:=\;\expo{\ii B_{\rm I}}\;=\;\expo{\ii b_-}\delta_-\;+\;\expo{\ii b_0}\delta_0\;+\;
\expo{\ii b_+}\delta_+\;.
$$
Let us assume $b_-\neq b_+$.
From the definition one has that $\tau_{(0,q e_2)}(f_{\rm I})=f_{\rm I}$ 
and $\tau_{(q e_1,0)}(f_{\rm I})\neq f_{\rm I}$
for all $q\in\Z\setminus \{0\}$. This means that the $\Z^2$-action on $f_{\rm I}$ indeed reduces to a $\Z$-action.
It follows that the Gelfand isomorphism $\s{C}(f_
{\rm I},\Z^2)\simeq\s{C}(\Omega_{\rm I})$ is provided by the 
\emph{Iwatsuka magnetic hull}
\begin{equation}\label{eq:dec_hull_iwa}
\Omega_{\rm I}\;\simeq\;\Z\;\cup\;\{-\infty\}\;\cup\;\{+\infty\}
\end{equation}
given by the two-point compactification of $\Z$. The inclusion $\Z\ni q \mapsto \omega_q\in \Omega_{\rm I}$ is given by the evaluation at \emph{finite distance}   
 defined by 
 $$
 \omega_q\big(\tau_\gamma(f_{\rm I})\big)\;:=\;f_{\rm I}\big((q-\gamma_1) e_1-\gamma_2 e_2\big)
 $$ for every $\gamma=(\gamma_1,\gamma_2)\in\Z^2$. The two limit points $\{\pm \infty\}$ are identified with the evaluations at infinity $\omega_{\pm\infty}\in \Omega_{\rm I}$ defined by
 $$
 \omega_{\pm\infty}\big(\tau_\gamma(f_{\rm I})\big)\;:=\;\expo{\ii b_{\pm}}
 $$ for every $\gamma\in\Z^2$. 
From the construction it follows that $\Z\simeq {\rm Orb}(\omega_0)$ and in turn $\{\pm\infty\}\simeq \partial \Omega_{\rm I}$. Therefore, equation \eqref{eq:dec_hull_iwa} provides a decomposition of $\Omega_{\rm I}$ in three invariant subsets.
 Since $\Z^2$ acts on ${\rm Orb}(\omega_0)$ as a one dimensional shift it follows that ${\rm Orb}(\omega_0)$ is made by wondering points \cite[Definition 5.5]{walters-82}. As a consequence every ergodic measure $\n{P}\in{\rm Erg}( \Omega_{\rm I})$ necessarily must satisfy $\n{P}({\rm Orb}(\omega_0))=0$
 \cite[Theorem 6.15]{walters-82}.
 This implies that the set ${\rm Erg}( \Omega_{\rm I})=\{\n{P}_{\pm\infty}\}$ is made by two ergodic measures at infinity specified by the condition $\n{P}_{\pm\infty}(\{\pm\infty\})=1$.\hfill $\blacktriangleleft$
\end{example}

\begin{example}[Magnetic hull for a localized magnetic field]
\label{ex:magn_hull_loc}
In the case of a localized  magnetic field \eqref{eq:loc_B} 
one has 
$$
f_{\Lambda}\;:=\;\expo{\ii B_{\Lambda}}\;=\;(\expo{\ii b}-1)\;\delta_\Lambda\;+\;1\;.
$$
Observe that $\tau_\gamma(f_{\Lambda})=f_{\gamma+\Lambda}\neq f_{\Lambda}$ for every $\gamma\in\Z^2\setminus\{0\}$.
In this case the Gelfand isomorphism $\s{C}(f_
{\Lambda},\Z^2)\simeq\s{C}(\Omega_{\Lambda})$ is given by the 
\emph{localized magnetic hull}
\begin{equation}\label{eq:dec_hull}
\Omega_{\Lambda}\;\simeq\;\Z^2\;\cup\;\{\infty\}\end{equation}
given by the one-point compactification of $\Z^2$. The inclusion $\Z^2\ni \xi \mapsto \omega_\xi\in \Omega_{\Lambda}$ is given by the  evaluation at \emph{finite distance}   
 defined by 
 $$
 \omega_\xi\big(\tau_\gamma(f_{\Lambda})\big)\;:=\;f_{\Lambda}\big(\xi-\gamma\big)
 $$ for every $\gamma\in\Z^2$. The  limit point $\{ \infty\}$ is identified with the evaluation at infinity $\omega_{\infty}\in \Omega_{\Lambda}$ given by
 $$
 \omega_{\infty}\big(\tau_\gamma(f_{\Lambda})\big)\;:=\;1 $$ for every $\gamma\in\Z^2$. 
 From the construction it follows that $\Z^2\simeq {\rm Orb}(\omega_0)$ and  $\{\infty\}\simeq \partial \Omega_{\Lambda}$ are the two invariant subsets of $\Omega_{\Lambda}$.
Since ${\rm Orb}(\omega_0)$ is made of wondering points under the action of $\Z^2$ it follows that ${\rm Erg}( \Omega_{\Lambda})=\{\n{P}_{\infty}\}$ where the measure at infinity $\n{P}_{\infty}$ is specified by $\n{P}_{\infty}(\{\infty\})=1$. \hfill $\blacktriangleleft$
\end{example}

\medskip

The ergodic measures of  $(\Omega_B,{\tau}^*,\Z^2)$ play a crucial role for the construction of the integration theory of the magnetic algebra $\s{A}_{A_B}$. 
Let us start by the following fact.
\begin{lemma}\label{lemma:inv_prop}
Under the isomorphism $\iota : \s{F}_B\rightarrow\s{C}(\Omega_B)$
 every invariant measure $\n{P}\in {\rm Mes}_{1,{\tau}^*}( \Omega_B)$ defines a trace $t_{\n{P}}$ on  $\s{F}_B$ through the formula
$$
t_{\n{P}}(\rr{g})\;:=\;\int_{\Omega_B}\dd\n{P}(\omega)\; \iota(\rr{g})(\omega)\;,\qquad\quad \rr{g}\in\s{F}_B\;.
$$
The trace $t_{\n{P}}$ is $\Z^2$-invariant in the sense that
$$
t_{\n{P}}\left(\tau_\gamma(\rr{g})\right)\;=\;t_{\n{P}}(\rr{g})\;,\qquad \forall\; \gamma=(\gamma_1,\gamma_2)\in\Z^2
$$ 
where $\tau_\gamma(\rr{g}):=(\rr{s}_{A_B,1})^{\gamma_1}(\rr{s}_{A_B,2})^{\gamma_2}\rr{g}
(\rr{s}_{A_B,2})^{-\gamma_2}(\rr{s}_{A_B,1})^{-\gamma_1}$
\end{lemma}
\proof
The claim follows immediately by noting that the integration with respect to $\n{P}$ defines an invariant trace on  $\s{C}(\Omega_B)$.
\qed

\medskip

We are now in position to construct the integration theory of the magnetic algebra $\s{A}_{A_B}$. Let $\n{P}\in {\rm Mes}_{1,{\tau}^*}( \Omega_B)$ be an invariant measure and define the map $\bb{T}_{\n{P}}:\s{A}_{A_B}\to \C$ by
\begin{equation}\label{eq:tra_01}
\bb{T}_{\n{P}}(\rr{a})\;:=\;t_{\n{P}}(\hat{\rr{a}}_{0,0})\;,\qquad\quad \rr{a}\in\s{A}_{A_B}
\end{equation}
where the $\s{F}_B$-valued coefficient $\hat{\rr{a}}_{0,0}$ is defined by \eqref{eq:oxoco}.
\begin{proposition}\label{prop:trace_properties}
The map $\bb{T}_{\n{P}}:\s{A}_{A_B}\to \C$ defined by \eqref{eq:tra_01} is a tracial state of the $C^*$-algebra $\s{A}_{A_B}$. Moreover, it holds true that:
\begin{itemize}
\item[(i)] $\bb{T}_{\n{P}}(\nabla_j(\rr{a}))=0$ for all $\rr{a}\in\s{C}^1(\s{A}_{A_B})$ and $j=1,2$;
\vspace{1mm}
\item[(ii)] $\bb{T}_{\n{P}}(\rr{b}\nabla_j(\rr{a}))=-\bb{T}_{\n{P}}(\rr{a}\nabla_j(\rr{b}))$ for all $\rr{a},\rr{b}\in \s{C}^1(\s{A}_{A_B})$ and $j=1,2$.
\end{itemize}
\end{proposition}
\proof
The map $\bb{T}_{\n{P}}$ is evidently linear (composition of linear maps) and normalized, \ie 
$\bb{T}_{\n{P}}({\bf 1})=1$. The positivity follows by observing that
$$
(\widehat{\rr{a}^*\rr{a}})_{0,0}\;=\;\int_{\n{T}^2}\dd\mu(\theta)\; \alpha_\theta(\rr{a}^*)\;\alpha_\theta(\rr{a})\;\geqslant\;0
$$
and consequently $\bb{T}_{\n{P}}(\rr{a}^*\rr{a})=t_{\n{P}}((\widehat{\rr{a}^*\rr{a}})_{0,0})\geqslant0$ since $t_{\n{P}}$ is a trace state (hence positive) on $\s{F}_B$.
Since $\bb{T}_{\n{P}}$ is linear and positive then it is automatically continuous \cite[Proposition 2.3.11]{bratteli-robinson-87}. To prove the cyclic property of the trace let us consider two monomials $\rr{u}_j :=\rr{g}_{j}\;(\rr{s}_{A_B,1})^{r_j}(\rr{s}_{A_B,2})^{s_j}$ 
with $\rr{g}_{j}\in \s{F}_B$ and $j=1,2$. Observe that 
$$
\rr{u}_1\rr{u}_2\;=\;\rr{g}_{1}\tau_{(r_1,s_1)}(\rr{g}_{2})\;(\rr{s}_{A_B,1})^{r_1}(\rr{s}_{A_B,2})^{s_1}(\rr{s}_{A_B,1})^{r_2}(\rr{s}_{A_B,2})^{s_2}
$$
where $\tau_{(r_1,s_1)}(\rr{g}_{2}):=(\rr{s}_{A_B,1})^{r_1}(\rr{s}_{A_B,2})^{s_1}\rr{g}_{2}(\rr{s}_{A_B,1})^{-r_1}(\rr{s}_{A_B,2})^{-s_1}$ and by mimicking the computation of 
\eqref{eq:aux_0010} one gets
$$
(\widehat{\rr{u}_1\rr{u}_2})_{0,0}\;=\;\rr{g}_{1}\;\tau_{(r_1,s_1)}(\rr{g}_{2})\;(\rr{s}_{A_B,1})^{r_1}(\rr{s}_{A_B,2})^{s_1}(\rr{s}_{A_B,1})^{-r_1}(\rr{s}_{A_B,2})^{-s_1}\;\delta_{r_1,-r_2}\delta_{s_1,-s_2}\;.
$$
A similar argument provides
$$
(\widehat{\rr{u}_2\rr{u}_1})_{0,0}\;=\;\tau_{(-r_1,-s_1)}(\rr{g}_{1})\;\rr{g}_{2}\;(\rr{s}_{A_B,1})^{-r_1}(\rr{s}_{A_B,2})^{-s_1}(\rr{s}_{A_B,1})^{r_1}(\rr{s}_{A_B,2})^{s_1}\;\delta_{r_1,-r_2}\delta_{s_1,-s_2}\;.
$$
An iterated application of the commutation relation \eqref{eqgen_comm_rel_X} provides
$$
\begin{aligned}
(\rr{s}_{A_B,1})^{r_1}(\rr{s}_{A_B,2})^{s_1}(\rr{s}_{A_B,1})^{-r_1}(\rr{s}_{A_B,2})^{-s_1}\;&=\;\prod_{j=0}^{r_1-1}\prod_{i=0}^{s_1-1}\tau_{(j,i)}(\rr{f}_{B})\\
&=:\;\rr{k}_{(r_1,s_1)}\;\in\;\s{F}_B\;.
\end{aligned}
$$
This implies
$$
(\widehat{\rr{u}_1\rr{u}_2})_{0,0}\;=\;\rr{k}_{(r_1,s_1)}
\;\rr{g}_{1}\;\tau_{(r_1,s_1)}(\rr{g}_{2})\;\delta_{r_1,-r_2}\delta_{s_1,-s_2}
$$
and
$$
\begin{aligned}
(\widehat{\rr{u}_2\rr{u}_1})_{0,0}\;&=\;\tau_{(-r_1,-s_1)}(\rr{k}_{(r_1,s_1)})
\;\tau_{(-r_1,-s_1)}(\rr{g}_{1})\;\rr{g}_{2}\;\delta_{r_1,-r_2}\delta_{s_1,-s_2}\\
&=\;\tau_{(-r_1,-s_1)}((\widehat{\rr{u}_1\rr{u}_2})_{0,0})\;.
\end{aligned}
$$
From the invariance property of Lemma \ref{lemma:inv_prop} it follows that
$$
t_{\n{P}}((\widehat{\rr{u}_1\rr{u}_2})_{0,0})\;=\;t_{\n{P}}((\widehat{\rr{u}_2\rr{u}_1})_{0,0})
$$
and in turn $\bb{T}_{\n{P}}(\rr{u}_1\rr{u}_2)=\bb{T}_{\n{P}}(\rr{u}_2\rr{u}_1)$
for all pair of monomials $\rr{u}_1,\rr{u}_2$.
It turns out that $\bb{T}_{\n{P}}$ satisfies the cyclic property of the trace on the dense subalgebra $\s{A}_{A_B}^0$ of the noncommutative polynomials, and by continuity
on the whole algebra $\s{A}_{A_B}$.
Property (i) follows from the computation at the beginning of of the proof of Lemma \ref{lem:diff-abs} which provides 
$\widehat{\nabla_j(\rr{a})}_{0,0}=0$ for $j=1,2$. Property (ii) follows by the application of property (i) along with the Leibniz's rule \eqref{eq:leib}.
\qed

\medskip

The trace property of the map $\bb{T}_{\n{P}}$
is guaranteed by the invariance property of the measure $\n{P}$. The ergodicity of $\n{P}$ plays a role for the physical interpretation of 
$\bb{T}_{\n{P}}$. For the next result we need to introduce some notation. 
Let $\{\Lambda_i\}_{i\in\N}\subset \s{P}_0(\Z^2)$ a sequence of bounded subsets of  cardinality $|\Lambda_i|$.  The family $\{\Lambda_i\}_{i\in\N}$ is  a \emph{F{\o}lner sequence}  \cite{greenleaf-69} if: (i) it is increasing, \ie $\Lambda_i\subseteq \Lambda_{i+1}$ for all $i\in\N$; (ii) it is exhaustive, \ie  $\Lambda_i\nearrow\Z^2$; (iii) it meets the F{\o}lner condition. \ie 
$$
\lim_{i\to \infty}\frac{|(\gamma+\Lambda_i)\triangle\Lambda_i|}{|\Lambda_i|}\;=\;0\;,\qquad\quad \forall\;\gamma\in\Z^2\;,
$$
where $\gamma+\Lambda_i$ is the $\gamma$-translated  of  $\Lambda_i$ and $\triangle$ is the  symmetric difference.

\medskip
\medskip

Let $\n{P}\in{\rm Erg}( \Omega_B)$ be an ergodic measure and $\{\Lambda_i\}_{i\in\N}$ a 
F{\o}lner sequence.  The Birkhoff's Ergodic Theorem \cite[Lemma 6.13]{walters-82}
assures that there exists a Borelian subset $Y\subseteq \Omega_B$ such that $\n{P}(Y)=1$ and
$$
t_{\n{P}}(\rr{g})\;=\;\lim_{i\to\infty}\frac{1}{|\Lambda_i|}\sum_{\gamma\in \Lambda_i}\iota(\rr{g})\big(\tau^*_\gamma(\omega)\big)\;,\qquad \forall\;  \omega\in Y\;,\;\;\forall\;\rr{g}\in\s{F}_B\;.
$$
 By observing that 
$\iota(\rr{g})\circ\tau^*_\gamma=\iota(\tau_{-\gamma}(\rr{g}))$
and recalling the definition of the trace $\bb{T}_{\n{P}}$ given by
\eqref{eq:tra_01} one gets
$$
\bb{T}_{\n{P}}(\rr{a})\;=\;\lim_{i\to\infty}\frac{1}{|\Lambda_i|}\sum_{\gamma\in \Lambda_i}\iota\big(\tau_{-\gamma}\big(\hat{\rr{a}}_{0,0}\big)\big)(\omega)\;,\qquad \forall\;  \omega\in Y\;,\;\;\forall\;\rr{a}\in \s{A}_{A_B}\;.
$$
 Finally, by observing that the  extraction of the $\s{F}_B$-valued coefficient commutes with the translations one gets 
$$
\bb{T}_{\n{P}}(\rr{a})\;=\;\lim_{i\to\infty}\frac{1}{|\Lambda_i|}\sum_{\gamma\in \Lambda_i}\iota\left(\widehat{\tau_{-\gamma}(\rr{a})}_{0,0}\right)(\omega)\;,\qquad \forall\;  \omega\in Y\;,\;\;\forall\;\rr{a}\in \s{A}_{A_B}\;.
$$
The latter formula becomes physically meaningful in the special case $\Omega_B=\{\omega_0\}$ as for the constant magnetic field (\cf Example \ref{ex:cost_hull}).
\begin{proposition}[Trace per unit volume]
\label{prop:tr_u_vol}
Assume that $\Omega_B=\{\omega_0\}$ and let $\n{P}$ be the (ergodic) measure supported on $\{\omega_0\}$.
Let $\{\Lambda_i\}_{i\in\N}$ be a 
F{\o}lner sequence and for every 
$\Lambda_i$ let $\rr{p}_{\Lambda_i}$ be the associated projection defined by $(\rr{p}_{\Lambda_i}\psi)(n)=\delta_{\Lambda_i}(n)\psi(n)$ for all $\psi\in\ell^2(\Z^2)$.
Then, it holds true that
$$
\bb{T}_{\n{P}}(\rr{a})\;:=\;\lim_{i\to\infty}\frac{1}{|\Lambda_i|}{\rm Tr}_{\ell^2(\Z^2)}\big(\rr{p}_{\Lambda_i}\;\rr{a}\;\rr{p}_{\Lambda_i}\big)\;,\qquad\quad \forall\; \rr{a}\in \s{A}_{A_B}\;.
$$
\end{proposition}
\proof
Let $\psi_\gamma\in\ell^2(\Z^2)$ be the normalized vector defined by $\psi_\gamma(n):=\delta_{n,\gamma}$. Then, it holds true that
$\imath(\hat{\rr{a}}_{0,0})(\omega_0)=\langle\psi_0,\rr{a}\psi_0\rangle$ for all $\rr{a}\in \s{A}_{A_B}$
where  $\omega_0$ can be identified with the evaluation at $0\in\Z^2$.
 Indeed, from \eqref{eq:oxoco_02} one gets that $\imath(\hat{\rr{a}}_{0,0})(\omega_0)=\langle\psi_0,\sigma_N(\rr{a})\psi_0\rangle$ for all $N\in\N$ and the continuity of the scalar product concludes the argument. Therefore after some straightforward computation one obtains
$$
\begin{aligned}
\sum_{\gamma\in \Lambda_i}\imath\left(\widehat{\tau_{-\gamma}(\rr{a})}_{0,0}\right)(\omega_0)\;&=\;\sum_{\gamma\in \Lambda_i}\langle\psi_0,\tau_{-\gamma}(\rr{a})\psi_0\rangle\\
&=\;\sum_{\gamma\in \Lambda_i}\langle\psi_\gamma,\rr{a}\psi_\gamma\rangle\;=\;
{\rm Tr}_{\ell^2(\Z^2)}\big(\rr{p}_{\Lambda_i}\;\rr{a}\;\rr{p}_{\Lambda_i}\big)\;,
\end{aligned}
$$
and this concludes the proof
\qed

\subsection{Noncommutative Sobolev spaces}\label{sect:NC-sobolev}
The existence of a trace $\bb{T}_{\n{P}}$ for the magnetic algebra $\s{A}_{A_B}$ allows to define the $L^p$-norms
$$
\|\rr{a}\|_{L^p}\;:=\;\bb{T}_{\n{P}}\left(|\rr{a}|^p\right)^{\frac{1}{p}}\;,\qquad \rr{a}\in \s{A}_{A_B}\;,\;\; p>0
$$
where $|\rr{a}|=\sqrt{\rr{a}\rr{a}^*}$, and the noncommutative $L^p$-spaces
$$
\s{L}_{A_B}^p\;:=\;\overline{ \s{A}_{A_B}}^{\;\|\;\|_{L^p}}
$$
For a review of the theory of the noncommutative $L^p$-spaces
 we refer to \cite{segal-53,nelson-13,terp-93} (see also~\cite[Section
3.2]{denittis-lein-book} and references therein).

\medskip

The noncommutative Sobolev spaces are defined by
$$
\s{W}^{k,p}_{A_B}\;:=\;\overline{\s{A}_{A_B}^0}^{\;\|\;\|_{k,L^p}}\;,
$$
obtained by closing the noncommutative polynomials with respect to the Sobolev norms
$$
\|\rr{a}\|_{k,L^p}\;:=\;\sum_{i=0}^k\sum_{a+b=i}\|\nabla_1^a\nabla_2^b(\rr{a})\|_{L^p}\;.
$$

\section{Magnetic interfaces, Toeplitz extensions and \texorpdfstring{$K$-}-theory}
\label{sec:k-theo_gen_B_over}
In this section we will describe shows that magnetic $C^*$-algebras
are  Toeplitz-type extensions of its \emph{interfaces} subalgebras. 
This observation will be used to study the $K$-theory of the magnetic $C^*$-algebras.

\subsection{Evaluation homomorphisms and interface algebra}\label{sect:asym_ses}
Let $B_1$ and $B_2$ be two magnetic fields with associated vector potentials $A_{B_1}$ and $A_{B_2}$, respectively. In this section we will  study a  family of $C^\ast$-homomorphisms between the magnetic algebras $\s{A}_{A_{B_1}}$ and $\s{A}_{A_{B_2}}$ which will be of central importance in the rest of the work.
\begin{definition}[Evaluation homomorphisms]
A $C^\ast$-homomorphism 
${\rm ev}:\s{A}_{A_{B_1}}\to\s{A}_{A_{B_2}}$ such that 
$$
\begin{aligned}
{\rm ev}\big(\rr{s}_{A_{B_1},1}\big)\;:&=\;\rr{s}_{A_{B_2},1}\\
{\rm ev}\big(\rr{s}_{A_{B_1},2}\big)\;:&=\;\rr{s}_{A_{B_2},2}\;.
\end{aligned}
$$
will be called an \emph{evaluation homomorphism} from $\s{A}_{A_{B_1}}$ to $\s{A}_{A_{B_2}}$. 
\end{definition}

\medskip

Let $\rr{f}_{B_1}$ and $\rr{f}_{B_2}$ be the  flux operators of the magnetic algebras $\s{A}_{A_{B_1}}$ and $\s{A}_{A_{B_2}}$, respectively. If 
${\rm ev}:\s{A}_{A_{B_1}}\to\s{A}_{A_{B_2}}$ is an evaluation homomorphism then from \eqref{eqgen_comm_rel_X} it follows that
$
{\rm ev}(\rr{f}_{B_1})=\rr{f}_{B_2}
$.
More in general, one can check that 
\begin{equation}\label{eq:surj_F_B_01}
{\rm ev}\big(\tau_\gamma(\rr{f}_{B_1})\big)\;=\;\tau_\gamma(\rr{f}_{B_2})\;,\qquad\quad \forall\;\gamma\in\Z^2
\end{equation}
where, with a little abuse of notation,  $\tau_\gamma$ denotes the $\Z^2$-action described in Section \ref{sect:magnetic_C_Al},  for both algebras. 
\begin{lemma}\label{lemm:surj_ev_01}
Let ${\rm ev}:\s{A}_{A_{B_1}}\to\s{A}_{A_{B_2}}$ be an  evaluation homomorphism. Then
$$
{\rm ev}|_{\s{F}_{B_1}}\;:\;\s{F}_{B_1}\;\longrightarrow\;\s{F}_{B_2}
$$
restricts to a surjective $C^\ast$-homomorphism.
\end{lemma}
\proof
Let $\s{F}_{B_j}^0\subset\s{F}_{B_j}$, $j=1,2$, be the dense subalgebra generated by the finite polynomials in the generators $\tau_\gamma(\rr{f}_{B_j})$. From \eqref{eq:surj_F_B_01} it follows that 
$\s{F}^0_{B_2}\subseteq{\rm ev}(\s{F}^0_{B_1})\subset\s{F}_{B_2}$.  Since by assumption ${\rm ev}$ is  a 
$C^\ast$-homomorphism and $\s{F}^0_{B_1}$ is dense, one gets that 
$$
\s{F}^0_{B_2}\;\subseteq\;{\rm ev}(\s{F}^0_{B_1})\;\subset\;{\rm ev}(\s{F}_{B_1})\;\subseteq\;\overline{{\rm ev}\big(\s{F}^0_{B_1}\big)}\;\subseteq\;
\s{F}_{B_2}\;.
$$
From the chain of inclusions above it follows that ${\rm ev}(\s{F}_{B_1})$ is a $C^*$-subalgebra of $\s{F}_{B_2}$ \cite[Proposition 2.3.1]{bratteli-robinson-87} which contains the dense set $\s{F}^0_{B_2}$.
This implies that: (i) the restriction  ${\rm ev}|_{\s{F}_{B_1}}$ is well defined, and (ii)
${\rm ev}(\s{F}_{B_1})=\s{F}_{B_2}$, \ie the surjectivity of the map.
\qed

\medskip

Since ${\rm ev}|_{\s{F}_{B_1}}$ is a well defined $C^\ast$-homomorphism between $\s{F}_{B_1}$ and $\s{F}_{B_2}$ it follows that 
${\rm Ker}({\rm ev}|_{\s{F}_{B_1}})\subset \s{F}_{B_1}$ is a closed (two-sided) ideal.
\begin{definition}[Interface algebra]\label{def:interf}
Let ${\rm ev}:\s{A}_{A_{B_1}}\to\s{A}_{A_{B_2}}$ be an evaluation homomorphism. The \emph{interface algebra} $\s{I}\subset \s{A}_{A_{B_1}}$ is the closed two-sided ideal generated in $\s{A}_{A_{B_1}}$ by ${\rm Ker}({\rm ev}|_{\s{F}_{B_1}})$.
\end{definition}

\medskip

In other words $\s{I}$ coincides with the $C^*$-subalgebra of $\s{A}_{A_{B_1}}$ generated by elements of the type $\rr{a}\rr{g}\rr{b}$ with $\rr{g}\in {\rm Ker}({\rm ev}|_{\s{F}_{B_1}})$ and $\rr{a},\rr{b}\in \s{A}_{A_{B_1}}$. This justifies the following notation
$$
\s{I}\;:=\;\s{A}_{A_{B_1}}\;{\rm Ker}({\rm ev}|_{\s{F}_{B_1}})\;\s{A}_{A_{B_1}}\;.
$$ 

\medskip

Let $\s{K}(\s{H})$ be the closed ideal of compact operators on the separable Hilbert space $\s{H}$. It is worth recalling that $\s{K}(\s{H})$ is an \emph{essential ideal} in $\s{B}(\s{H})$ \cite[Example 3.1.2]{murphy-90}.  Let $\s{C}(\n{S}^1)$ be the $C^*$-algebra of the continuous function on the unit circle $\n{S}^1\simeq\R/2\pi\Z$.
\begin{definition}[Localized and straight-line interfaces]\label{def:typ_inter}
Let $\s{I}$ be the interface algebra associated to a given evaluation homomorphism. We will say that the interface is \emph{localized} if $\s{I}=\s{K}(\ell^2(\Z^2))$.
The interface will be called
 \emph{straight-line} if $\s{I}\simeq\s{C}(\n{S}^1)\otimes\s{K}(\ell^2(\Z))$ up to a unitary transformation.
\end{definition}

\medskip

The motivation for the terminology introduced in Definition \ref{def:typ_inter} will be clarified in part 
 in the next example and in part in Section \ref{sect:toep_iwatsuka_MT}.

\begin{example}[Interface algebra for a localized magnetic field]
\label{ex:int_alg_loc}
According to the notations introduced in Example \ref{Ex3:loc}, Example \ref{Ex3:MT_loc} and Example \ref{ex_B=0},
let $\s{A}_\Lambda$ be the magnetic algebra associated to a localized magnetic field $B_\Lambda$ and $\s{A}_0$ be the algebra associated to a (constant) zero magnetic field. The map
 defined by ${\rm ev}(\rr{s}_{\Lambda,j})=\rr{s}_j$, where $\rr{s}_j$  are the canonical shift operators, extends to a $C^*$-homomorphism ${\rm ev}:\s{A}_\Lambda\to \s{A}_0$ (see the proof of Proposition \ref{prop:cond_ev_hom}). Therefore, from \eqref{eq:surj_F_B_01} one has that ${\rm ev}(\rr{f}_\Lambda)={\bf 1}$, and in turn
$$
{\rm ev}(\rr{p}_\Lambda)\;=\;{\rm ev}\left((\expo{\ii b}-1)^{-1}(\rr{f}_\Lambda-{\bf 1})\right)\;=\;0\;,\qquad b\notin2\pi\Z\;.
$$ 
This shows that $\rr{p}_\Lambda\in\s{I}$ is an element of the interface algebra. Moreover, by acting with the magnetic translations $\rr{s}_{\Lambda,j}$ one obtains that also  $\rr{p}_{\gamma+\Lambda}\in\s{I}$ for every $\gamma\in\Z^2$. This fact 
is the key observation to conclude that $\s{I}=\s{K}(\ell^2(\Z^2))$. Let us start by the simplest case of a flux tube supported on a single site  $\Lambda:=\{\lambda_0\}$. In this case the latter  observation implies that every rank-one 
projection $\rr{p}_{\{\gamma\}}$ is in $\s{I}$ and this is  enough to conclude that $\s{I}$ is the $C^*$-algebra of compact operators (see \cite[Appendix B]{denittis-schulz-baldes-16} for the details). In the more general case $|\Lambda|>1$ one can show that  it is always possible to build a projector in $\s{I}$ supported in a single point. Let $\lambda_0,\lambda\in\Lambda$ be two distinct points and $\gamma_0:=\lambda_0-\lambda$. Then $\lambda_0\in\gamma_0+\Lambda$
and $\rr{p}_{\Lambda}({\bf 1}-\rr{p}_{\gamma_0+\Lambda})$ is a projection in $\s{I}$ which projects on a subset $\Lambda'\subset\Lambda$ where the strict inclusion is justified by the fact that $\lambda_0\notin\Lambda'$. By iterating the procedure  a sufficient number of times one ends with a projection on a singleton. In summary, a localized magnetic field always provides a localized interface in the sense of Definition \ref{def:typ_inter}.\hfill $\blacktriangleleft$
\end{example}

\subsection{Toeplitz extensions by an interface}\label{sect:toep_inter}
Let $\s{A}, \s{B}$ and $\s{C}$ three $C^*$-algebras fitting into the short exact sequence
\begin{equation}\label{eq:ex_seq}
0\;\stackrel{}{\longrightarrow}\;\s{A}\;\stackrel{\alpha}{\longrightarrow}\;\s{B}\;\stackrel{\beta}{\longrightarrow}\s{C}\;\stackrel{}{\longrightarrow}0\;.
\end{equation}
In such a case we will say that $\s{B}$ is the \emph{Toeplitz extension} of $\s{A}$ by $\s{C}$. For a simple and complete review of the theory of extension of $C^*$-algebras we refer to
\cite[Chapter 3]{wegge-olsen-93}. It is worth pointing out that we are proposing the use of the expression Toeplitz extension
in a somehow (ultra) generalized sense. Indeed the original notion of Toeplitz extension refers to a very specific example of extension of $C^*$-algebras, see \eg \cite[Section 3.5]{murphy-90} or \cite[Exercise 3.F]{wegge-olsen-93}. However, such a generalized use of the name Toeplitz extension  is becoming standard in condensed matter community (see \eg \cite{arici-mesland-19})
and we decided to adhere to this use.

\medskip

The main aim of this section is to show that an evaluation homomorphism automatically provides a Toeplitz extension.

\begin{theorem}\label{theo:SES}
Every evaluation homomorphism ${\rm ev}:\s{A}_{A_{B_1}}\to\s{A}_{A_{B_2}}$ fits into the short exact sequence
\begin{equation}\label{eq:ex_seq_021_0_int}
0\;\stackrel{}{\longrightarrow}\;\s{I}\;\stackrel{\imath}{\longrightarrow}\;\s{A}_{A_{B_1}}\;\stackrel{\rm ev}{\longrightarrow}\s{A}_{A_{B_2}}\;\stackrel{}{\longrightarrow}0\;
\end{equation}
where $\s{I}$ is the related interface algebra
and $\imath$ is the (natural) inclusion map.
\end{theorem}
\proof
The map $\imath$ is injective by definition. Therefore, to complete the proof we need to prove that the evaluation homomorphism   is   surjective and that 
${\rm Ker}({\rm ev})=\s{I}$. The surjectivity is a consequence of Lemma \ref{lemm:surj_ev_01}
which ensures  ${\rm ev}(\s{A}_{A_{B_1}}^0)=\s{A}_{A_{B_2}}^0$ (or equivalently
${\rm ev}(\s{A}_{A_{B_1}}^\infty)=\s{A}_{A_{B_2}}^\infty$). Then, as in the proof of Lemma \ref{lemm:surj_ev_01}, the chain of inclusions
$$
\s{A}_{A_{B_2}}^0\;=\;{\rm ev}(\s{A}_{A_{B_1}}^0)\;\subset\;{\rm ev}(\s{A}_{A_{B_1}})\;\subseteq\;\overline{{\rm ev}(\s{A}_{A_{B_1}}^0)}\;\subseteq\;
\s{A}_{A_{B_2}}\;
$$
implies  ${\rm ev}(\s{A}_{A_{B_1}})=\s{A}_{A_{B_2}}$. The description of the kernel of 
${\rm ev}$ is a consequence of Corollary \ref{cor_uni_exp}  which guarantees that
  $\rr{a}\in {\rm Ker}({\rm ev})$ if and only if all the $\s{F}_{B_2}$-coefficients of 
${\rm ev}(\rr{a})$ are   zero. From 
the definition \eqref{eq:oxoco}, the linearity of the integral and the fact that the evaluation homomorphism ${\rm ev}$ commutes (by construction) with the family of automorphisms $\alpha_\theta$, it follows that 
$\widehat{{\rm ev}(\rr{a})}_{r,s}={\rm ev}(\widehat{\rr{a}}_{r,s})$. Then $\rr{a}\in{\rm Ker}({\rm ev})$ if and only if ${\rm ev}(\widehat{\rr{a}}_{r,s})=0$ for all $(r,s)\in\Z^2$. 
This implies that $\rr{a}\in{\rm Ker}({\rm ev})$ if and only if
$\sigma_N(\rr{a})\in\s{I}$ for every $N\in\N$, where $\sigma_N(\rr{a})$ is the {Ces\`aro mean} \eqref{eq:oxoco_02} which converges to $\rr{a}$. Since $\s{I}$ is a closed ideal it follows that $\rr{a}\in{\rm Ker}({\rm ev})$ if and only if $\rr{a}\in\s{I}$.
\qed

\medskip

By using the terminology introduced at the beginning of this section we will say that $\s{A}_{A_{B_1}}$ is the Toeplitz extension of the \emph{interface} $\s{I}$ by the final (or \emph{bulk}\footnote{The use of this terminology will be clarified in Definition \ref{def:bulk_alg}.}) algebra $\s{A}_{A_{B_2}}$. 
\begin{corollary}\label{cor:triv_inter}
${\bf 1}\in \s{I}$ if and only if ${\rm ev}=0$.
\end{corollary}
\proof
Since $\s{I}$ is an ideal one has that ${\bf 1}\in \s{I}$ if and only if $\s{I}=\s{A}_{A_{B_1}}$.
\qed

\begin{example}[Localized interface and discrete spectrum]
\label{ex:loc_int}
In the case of a localized interface $\s{I}=\s{K}(\ell^2(\Z^2))$ (like in Example \ref{ex:int_alg_loc}) the short exact sequence \eqref{eq:ex_seq_021_0_int} provides the isomorphism
$$
\s{A}_{A_{B_2}}\;\simeq\;\s{A}_{A_{B_1}}/\s{K}(\ell^2(\Z^2))\;.
$$
This means that the elements of $\s{A}_{A_{B_1}}$ are compact perturbations of elements of the (bulk) algebra $\s{A}_{A_{B_2}}$. In this case the algebra $\s{A}_{A_{B_2}}$ is known as the \emph{corona} of $\s{A}_{A_{B_1}}$
and the  short exact sequence \eqref{eq:ex_seq_021_0_int}
 is called \emph{essential} \cite[Definition 3.2.1]{wegge-olsen-93}. The isomorphism above  is useful to analyze the spectrum of elements $\rr{a}\in \s{A}_{A_{B_1}}$. In fact it holds true that the  evaluation ${\rm ev}(\rr{a})\in  \s{A}_{A_{B_2}}$ contains the information about the essential spectrum  $\sigma_{\rm ess}(\rr{a})$ while the discrete spectrum  $\sigma_{\rm d}(\rr{a})$ is generated by the part of $\rr{a}$ which belongs to the interface. Usually, the discrete spectrum of $\rr{a}$ is located in the gaps of the spectrum of  ${\rm ev}(\rr{a})$.
\hfill $\blacktriangleleft$
\end{example}

\medskip

The short exact sequence  \eqref{eq:ex_seq} is called \emph{split exact} if there exists a $C^*$-homomorphism (the lifting map)
$\beta':\s{C}\to\s{B}$ such that $\beta\circ\beta'={\rm Id}_{\s{C}}$. In such a case 
both $\alpha(\s{A})$ and $\beta'(\s{C})$
 are $C^*$-subalgebras of $\s{B}$ and 
 $\s{B}=\alpha(\s{A})+\beta'(\s{C})$ is the Banach space \emph{direct sum} of these two $C^*$-subalgebras. It is worth pointing out that this is not the same of the direct sum of $C^*$-algebras. The latter condition is stronger and requires that also $\beta'(\s{C})$ is an ideal in $\s{B}$.
When this extra condition holds true one has that  
  $\s{B}=\alpha(\s{A})\oplus\beta'(\s{C})$ is an \emph{orthogonal} direct sum of $C^*$-algebras.

 \begin{example}[Toeplitz extension for a localized magnetic field]
\label{ex:toep_ext_loc}
From 
Example
\ref{ex:int_alg_loc} one infers that  a localized magnetic field  provides the Toeplitz extension
\begin{equation}\label{eq:ex_seq_021_0_int_loc}
0\;\stackrel{}{\longrightarrow}\;\s{K}\big(\ell^2(\Z^2)\big)\;\stackrel{\imath}{\longrightarrow}\;\s{A}_{\Lambda}\;\stackrel{\rm ev}{\longrightarrow}\s{A}_{0}\;\stackrel{}{\longrightarrow}0\;.
\end{equation}
By adapting the argument in \cite[Proposition 2]{denittis-schulz-baldes-16}, and up to gauge transformation,	one can 
show that the difference $\rr{s}_{\Lambda,1}-\rr{s}_{1}=(\rr{y}_{A_\Lambda,1}-{\bf 1})\rr{s}_{1}$ is a compact operator.
Therefore, one has that  $\s{A}_{\Lambda}$ coincides with the $C^*$-algebra generated by $\s{K}(\ell^2(\Z^2))$ and $\s{A}_{0}$
\cite[Theorem 11]{denittis-schulz-baldes-16} and, as a consequence, 
the Toeplitz extension \eqref{eq:ex_seq_021_0_int_loc} is split exact in view of the (inclusion) homomorphism $\jmath:\s{A}_{0}\hookrightarrow \s{A}_{\Lambda}$. This fact implies that 
$\s{A}_{\Lambda}=\s{K}(\ell^2(\Z^2))+\s{A}_{0}$  as direct sum of Banach spaces.
\hfill $\blacktriangleleft$
\end{example}
 
\medskip

Example  \ref{ex:toep_ext_loc} is somehow special since 
  the Toeplitz extensions \eqref{eq:ex_seq_021_0_int} considered in this work will be  not split exact in general. Nevertheless,
in many cases it is possible to show that there exists a \emph{linear map} (not a $C^*$-homomorphism) $\jmath:\s{A}_{A_{B_2}}\to \s{A}_{A_{B_1}}$ such that ${\rm ev}\circ \jmath$ is the identity on ${\s{A}_{A_{B_2}}}$ and which provides a splitting of the linear space structure of ${\s{A}_{A_{B_1}}}$. This is   made possible by the fact that the magnetic algebras are generated by monomials in the magnetic translations.
An extended discussion of this aspect can be found in \cite[Section 3.2.2]{prodan-schulz-baldes-book}.

\subsection{Evaluation homomorphisms and dynamics}\label{sect:dim_ev}
In the previous section we described the consequences of having an evaluation homomorphism between two magnetic algebras.
In this section we will analyze the relation between the existence of evaluation homomorphisms and the dynamical properties  of the dynamical systems 
generated by the magnetic hulls. As a result we will provide a generalized definition of \emph{magnetic multi-interface} based on purely dynamical properties of the  magnetic hulls. 

\medskip

 Let $(\Omega_{B_1},{\tau}^*,\Z^2)$ and $(\Omega_{B_2},{\tau}^*,\Z^2)$ be the two 
topological dynamical systems 
associated to the magnetic fields $B_1$ and $B_2$, respectively. An \emph{equivariant} map from $\Omega_{B_2}$ to $\Omega_{B_1}$ is a continuous function ${\phi}^*:\Omega_{B_2}\to\Omega_{B_1}$ such that
$$
{\phi}^*\;\circ\;{{\tau}}^*_\gamma\;=\;{\tau}^*_\gamma\;\circ\; {\phi}^*\;,\qquad\quad \forall\;\gamma\in\Z^2\;.
$$
\begin{proposition}\label{prop:sub_dyn_sys}
Every evaluation homomorphism ${\rm ev}:\s{A}_{A_{B_1}}\to\s{A}_{A_{B_2}}$ defines an injective closed
{equivariant} map ${\phi}^*:\Omega_{B_2}\hookrightarrow\Omega_{B_1}$.
\end{proposition}
\proof
Let us consider the Gelfand trasforms 
$\bb{G}_j:\s{F}_{B_j}\to \s{C}(\Omega_{B_j})$, with $j=1,2$. The map
$\phi:\s{C}(\Omega_{B_1})\to\s{C}(\Omega_{B_2})$ defined by 
$$
\phi\;:=\;\bb{G}_2\;\circ\;{\rm ev}|_{\s{F}_{B_1}}\;\circ\;\bb{G}_1^{-1}
$$
is the composition of surjective $C^*$-homomorphisms, hence it is a surjective $C^*$-homomorphism. 
By duality, $\phi$ induces a continuous map ${\phi}^*:\Omega_{B_2}\to\Omega_{B_1}$ defined by
$$
{\phi}^*(\omega)\;:=\;\omega\;\circ\;\phi\;.
$$ 
Indeed, if  
$\omega\in \Omega_2$  is meant as a character of $\s{C}(\Omega_{B_2})$, then 
${\phi}^*(\omega)$ is a character of $\s{C}(\Omega_{B_1})$, hence a point of $\Omega_1$.
The surjectivity of $\phi$ implies the injectivity of ${\phi}^*$. Indeed, ${\phi}^*(\omega_1)={\phi}^*(\omega_2)$ implies that $\omega_1(\hat{g})=\omega_2(\hat{g})$ for all $\hat{g}\in \s{C}(\Omega_{B_2})$ which is exactly   $\omega_1=\omega_2$. Finally ${\phi}^*$ is closed in view of the \emph{Closed Map Lemma} \cite[Lemma 4.25]{lee-00} since $\Omega_{B_1}$ and $\Omega_{B_2}$ are both compact Hausdorff spaces. 
\qed

\medskip

Let us recall that a continuous closed injection between topological spaces is usually called a \emph{(topological) embedding}.
Let ${\phi}^*:\Omega_{B_2}\hookrightarrow\Omega_{B_1}$ be the 
 {equivariant} embedding of  Proposition \eqref{prop:sub_dyn_sys}. The subset ${\Omega}_{\ast}:={\phi}^*(\Omega_{B_2})$ is evidently a closed invariant subset of $\Omega_{B_1}$ and 
$({\Omega}_{\ast},{\tau}^*,\Z^2)$ becomes a dynamical subsystem of $(\Omega_{B_1},{\tau}^*,\Z^2)$. 
Moreover
$$
{\Omega}_{\ast}\;=\;{\phi}^*\left(\overline{{\rm Orb}(\omega_0)}\right)\;=\;\overline{{\rm Orb}(\omega_\ast)}
$$
where $\omega_\ast:={\phi}^*(\omega_0)$ and $\omega_0\in \Omega_{B_2}$ is the evaluation at $0$.
In conclusion, Proposition \eqref{prop:sub_dyn_sys} states that every 
evaluation homomorphism identifies (up to  isomorphisms) a 
dynamical subsystem of the initial magnetic hull. However, in view of Proposition 
\ref{prop:dens_orb}, the only possibilities 
for a closed and invariant subset $\Omega_\ast$ are $\Omega_\ast\subseteq \partial\Omega_{B_1}$ or $\Omega_\ast=\Omega_{B_1}$. The latter circumstance corresponds to the case of ${\phi}^*$ being an isomorphism and, as a consequence of Proposition \ref{prop:sub_dyn_sys} and the 
short exact sequence of Theorem \ref{theo:SES}, this is equivalent to the
 isomorphism $\s{A}_{A_{B_1}}\simeq\s{A}_{A_{B_2}}$. 
This case will be called trivial as opposite to the  non trivial case in which ${\phi}^*$ defines a  proper dynamical subsystem of the initial dynamical system. The next result provides a sort  of inverse implication of Proposition \eqref{prop:sub_dyn_sys}.
\begin{proposition}\label{prop:cond_ev_hom}
Let $\s{A}_{A_{B}}$ be a magnetic algebra and 
$(\Omega_{B},{\tau}^*,\Z^2)$
the topological dynamical system associated to its magnetic hull. Let $\Omega_\ast\subseteq \partial\Omega_{B}$ be a proper invariant closed subset. 
Assume that $\Omega_\ast=\overline{{\rm Orb}(\omega_\ast)}$ for some $\omega_\ast\in\partial\Omega_{B}$.
Then, there is a  magnetic algebra $\s{A}_{A_{B_\ast}}$ with magnetic hull $\Omega_\ast$ and an evaluation homomorphism ${\rm ev}:\s{A}_{A_{B}}\to\s{A}_{A_{B_\ast}}$.
\end{proposition}
\proof
Let $\phi:\s{C}(\Omega_{B})\to\s{C}(\Omega_{\ast})$ be the surjective \emph{restriction}
$C^*$-homomorphism defined by $\phi(\hat{g}):=\hat{g}|_{\Omega_{\ast}}$ for all ${g}\in \s{C}(\Omega_{B})$.
Let $\widehat{f}_B$ be the Gelfand transform of the generator  $f_B$ of $\s{C}(f_B,\Z^2)$
and define the function $f_{B_\ast}:\Z^2\to\C$ by
$$
f_{B_\ast}(\gamma)\;:=\;
\widehat{f}_B({\tau}^*_\gamma(\omega_\ast))\;,\qquad\quad \gamma\in\Z^2\;.
$$
By the Gelfand isomorphism one obtains that $\s{C}(f_{B_\ast},\Z^2)\simeq\s{C}(\Omega_{\ast})$.
The function $f_{B_\ast}$ provides  a magnetic flux with an associated (non unique) magnetic field $B_\ast:\Z^2\to\R$. Let $A_{B_\ast}$ be a suitable vector potential for  $B_\ast$ and  $\s{A}_{A_{B_\ast}}$ the associated magnetic algebra. Let $\s{F}_{B_\ast}\subset \s{A}_{A_{B_\ast}}$ be the abelian subalgebra generated by the magnetic flux $\rr{f}_{B_\ast}:=\rr{m}_{f_{B_\ast}}$. The surjective 
$C^*$-homomorphism $\phi$ and the Gelfand isomorphism provide a surjective 
$C^*$-homomorphism $\widetilde{{\rm ev}}:\s{F}_{B}\to\s{F}_{B_\ast}$ characterized by 
$\widetilde{{\rm ev}}(\rr{f}_{B})=\rr{f}_{B_\ast}$. It turns out that the map
${\rm ev}:\s{A}_{B}^\infty\to\s{A}_{B_\ast}^\infty$ defined by
$$
{\rm ev}\left(\sum_{(r,s)\in\Z^2}\rr{g}_{r,s}\;(\rr{s}_{A_B,1})^{r}\;(\rr{s}_{A_B,2})^{s}\right)\:=\;\sum_{(r,s)\in\Z^2}\widetilde{{\rm ev}}(\rr{g}_{r,s})\;(\rr{s}_{A_{B_\ast},1})^{r}\;(\rr{s}_{A_{B_\ast},2})^{s}
$$
is a $\ast$-homomorphism of pre-$C^*$-algebras (Proposition \ref{prop:pre-C-ast}). Therefore, the claim follows from \cite[Lemma 3.41]{gracia-varilly-figueroa-01}.\qed

\medskip

\begin{remark}[Non-uniqueness of the magnetic field]
The magnetic algebra $\s{A}_{A_{B_\ast}}$ which enters in Proposition \ref{prop:cond_ev_hom} is
not unique for two reasons. First of all $\s{A}_{A_{B_\ast}}$ depends on the election of a vector potential $A_{B_\ast}$ for the magnetic field $B_\ast$ and this involves the election of gauge. However, magnetic algebras related to different gauges are 
unitarily equivalent as discussed in Section \ref{sect:magnetic_C_Al}. The second source of ambiguity is  more subtle and is related with the determination of the magnetic field ${B_\ast}$ from the magnetic flux $f_{B_\ast}$. Indeed, 
the natural candidate would be $B_\ast=-\ii\log(f_{B_\ast})$
but the the logarithm is not univocally defined in the complex plane. In particular, given a magnetic field ${B_\ast}$ compatible with the magnetic flux $f_{B_\ast}$ and a (not necessarily bounded) function $\zeta:\Z^2\to\Z$ one gets that 
$B_\ast':=B_\ast+2\pi\zeta$ provides the same magnetic flux.
A way to solve this ambiguity is to fix the convention that
$B_\ast:={\rm Arg}(f_{B_\ast})\in[0,2\pi)$ is given by the \emph{principal argument} of the flux $f_{B_\ast}$.
This correspond to a sort of \emph{minimal growth} assumption for the magnetic field at infinity and we will use this convention in the rest of this work.
\hfill $\blacktriangleleft$
\end{remark}

\medskip

We are now in position to introduce a key definition for this work.
\begin{definition}[Magnetic multi-interface]
\label{def:bulk_alg}
A  system subjected to a magnetic field $B:\Z^2\to[0,2\pi)$ and with the boundary of the magnetic hull given by 
a finite collection of invariant points
$$
\partial\Omega_B\;=\;\{\omega_{\ast,1},\ldots,\omega_{\ast,N+1}\}
$$
will be called a \emph{$N$-interface} magnetic system. In this case the associated Toeplitz extension is given by
\begin{equation}\label{eq:ex_seq_021_0_int_mult}
0\;\stackrel{}{\longrightarrow}\;\s{I}\;\stackrel{\imath}{\longrightarrow}\;\s{A}_{A_{B}}\;\stackrel{\rm ev}{\longrightarrow}\s{A}_{\rm bulk}\;\stackrel{}{\longrightarrow}0\;
\end{equation}
where $\s{A}_{A_{B}}$ is any magnetic algebra associated to the magnetic field $B$. The \emph{bulk algebra}  
\begin{equation}\label{eq:ex_seq_021_0_int_mult-02}
\s{A}_{\rm bulk}\;:=\;\s{A}_{b_1}\;\oplus\;\ldots \;\oplus\;\s{A}_{b_{N+1}}
\end{equation}
 is given by the orthogonal direct sum of $N+1$ magnetic algebras of constant magnetic fields. In particular  $\s{A}_{b_j}$ is constructed as in Example \ref{ex_B=0} with respect to the constant magnetic field of strength $b_j:={\rm Arg}(\widehat{f}_B(\omega_{\ast,j}))$ for every $j=1,\ldots,N+1$ where $\widehat{f}_B$ is the Gelfand transform of the flux function  $f_B$ as described in the proof of Proposition \ref{prop:cond_ev_hom}. Finally the evaluation map and the interface algebra $\s{I}$ are completely 
specified by
$$
{\rm ev} (\rr{f}_B)\;:=\;\left(\expo{\ii b_1}{\bf 1}\;,\ldots,\expo{\ii b_{N+1}}{\bf 1}\right)
$$
as discussed in Section \ref{sect:asym_ses}.
\end{definition}

\medskip

As showed in Example \ref{ex:magn_hull_loc} and Example \ref{ex:int_alg_loc}, a localized magnetic field provides an example of a magnetic interface with order $N=0$. On the other hand Example \ref{ex:hull_Iwatsuka} shows that 
the Iwatsuka magnetic field provides an example of magnetic interface of order $N=1$. The case of the Iwatsuka magnetic field 
will be discussed extensively in Section \ref{sect:iwatsuka}.

\subsection{The \texorpdfstring{$K$-}-theory of magnetic interfaces}\label{sect:K-theory}
In this section we will discuss some aspects of the $K$-theory of magnetic interfaces. There is a large literature concerning the $K$-theory for $C^*$-algebras. We will refer to the classic monographs \cite{murphy-90,wegge-olsen-93,blackadar-98,gracia-varilly-figueroa-01} as well as \cite{prodan-schulz-baldes-book}
for a stronger connection with  condensed matter problems.

\medskip

Let us recall that for each Toeplitz extension of type \eqref{eq:ex_seq} there is an associated \emph{six-term sequence} in $K$-theory \cite[Theorem 9.3.2]{wegge-olsen-93}. Therefore, there is   a six-term sequence for every magnetic
Toeplitz extension of type \eqref{eq:ex_seq_021_0_int} or 
\eqref{eq:ex_seq_021_0_int_mult}. We will focus here on the latter case concerning a magnetic multi-interface.

\medskip

From the exact sequence  \eqref{eq:ex_seq_021_0_int_mult} one obtains the  {six-term}  sequence 
\begin{equation}\label{eq:6-ex-seq}
\begin{array}{ccccc}
K_0(\s{I})&\overset{\imath_*}{\longrightarrow}          
 &K_0(\s{A}_{A_{B}}) &      \overset{{\rm ev}_*}{\longrightarrow}    
  &K_0(\s{A}_{\rm bulk})\\
  &&&&\\
\mbox{\footnotesize ind}\Big\uparrow   &&&&\Big\downarrow \mbox{\footnotesize exp}\\
   &&&&\\
  K_1(\s{A}_{\rm bulk})& \underset{{\rm ev}_*}{\longleftarrow}       
 &K_1(\s{A}_{A_{B}}) &   \underset{\imath_*}{\longleftarrow}        
  &K_1(\s{I})\\
\end{array}
\end{equation}
where the canonical maps ${\rm ind}$ and ${\rm exp}$ are  called \emph{index map} and \emph{exponential map} respectively. 
The role of the {six-term}  sequence \eqref{eq:6-ex-seq} is twofold: first of all it allows  to reconstruct the $K$-theory
of $\s{A}_{A_{B}}$ from the knowledge of the $K$-theory of $\s{I}$ and $\s{A}_{\rm bulk}$; secondly it defines how the $K$-theory of $\s{A}_{A_{B}}$ intertwines the $K$-theories of $\s{I}$ and $\s{A}_{\rm bulk}$ through the maps ${\rm ind}$ and ${\rm exp}$. The latter aspect is the core of the topological interpretation of the \emph{bulk-boundary correspondence} in
condensed matter \cite{prodan-schulz-baldes-book}.

\medskip

The $K$-theory of the bulk algebra $\s{A}_{\rm bulk}$ can be easily computed since  $\s{A}_{\rm bulk}$ is an orthogonal direct sum of noncommutative tori (\cf Example \ref{ex_B=0}) and the $K$-theory of the noncommutative torus is well-known (\cf Appendix \ref{app:Ktheo NCT}).
\begin{proposition}\label{prop:K-bulk}
Let $\s{A}_{\rm bulk}$ be the bulk algebra \eqref{eq:ex_seq_021_0_int_mult-02} of an {$N$-interface} magnetic system. Then 
$$
\begin{aligned}
K_0(\s{A}_{\rm bulk})\;&=\;\bigoplus_{j=1}^{N+1}K_0(\s{A}_{b_j})\;\simeq\;\bigoplus_{j=1}^{N+1}\Z^2\;,\\
K_1(\s{A}_{\rm bulk})\;&=\;\bigoplus_{j=1}^{N+1}K_1(\s{A}_{b_j})\;\simeq\;\bigoplus_{j=1}^{N+1}\Z^2\;.\\
\end{aligned}
$$
The isomorphisms above are induced by the $K$-theory of the noncommutative tori
according to \eqref{eq:K-teo_NCT_ind}.
\end{proposition}
\proof
The first part of the claim follows from the additive property of the $K$-theory with respect to the orthogonal direct sum of $C^*$-algebras \cite[Exercise 6.E \& Example 7.1.11(4)]{wegge-olsen-93}. The second part is a consequence of the structure of the $K$-theory of the noncommutative torus (\cf  Appendix \ref{app:Ktheo NCT}).
\qed

\medskip

The $K$-theory of the interface algebra requires a preliminary observation. In fact, if one  assumes that ${\rm ev}$ is \emph{not} trivial one has that  $\s{I}$ is \emph{not} unital (Corollary \ref{cor:triv_inter}) and as a consequence 
$K_j(\s{I})$, $j=0,1$, must be understood as the $K$-groups of the unitalization\footnote{In our case $\s{I}\subset\s{B}(\ell^2(\Z^2))$ is a concrete $C^*$-algebra, therefore its unitalization is given by $
\s{I}^+:=\left\{\rr{a}+\alpha{\bf 1}\;|\;\rr{a}\in\s{I}\;,\;\;\alpha\in\C\right\}
.$} $\s{I}^+$ of $\s{I}$ \cite{wegge-olsen-93,rordam-larsen-laustsen-00}. The main case of interest for this work\footnote{The ansatz \eqref{eq:ansatz_inter}
imposes a quite strong condition on the geometry of the interface. To handle more general geometries like corners, the ansatz \eqref{eq:ansatz_inter} must be modified in the form discussed in \cite{thiang-20}.} is when there exists a unitary equivalence
\begin{equation}\label{eq:ansatz_inter}
\s{I}\;\stackrel{\varpi}{\simeq}\; \s{I}_0\;\otimes\;\s{K}(\s{H}_{\rm red})
\end{equation}
where 
$\s{I}_0$ 
is a unital and abelian $C^*$-algebra and $\s{K}(\s{H}_{\rm red})$ is the $C^*$-algebra of compact operators on the (reduced) separable Hilbert space $\s{H}_{\rm red}$. In such case one has
$$
K_j(\s{I})\;\simeq\; K_j(\s{I}_0)\;,\qquad j=0,1\;,
$$
because of the stability property of  $K$-theory \cite[Corollary 6.2.11 \& Corollary 7.1.9]{wegge-olsen-93}.

 \begin{example}[{six-term}  sequence  for a localized magnetic field]
The  {six-term}  sequence  associated to the Toeplitz extension
\eqref{eq:ex_seq_021_0_int_loc} for a localized magnetic field can be easily computed by observing that in this case the interface algebra has the form $\s{I}\simeq\C\otimes\s{K}$ (\cf Example \ref{ex:int_alg_loc}) and in turn its $K$-theory is given by
$$
K_0(\s{I})\;\simeq\; K_0(\C)\;\simeq\;\Z\;,\qquad
K_1(\s{I})\;\simeq\; K_1(\C)\;=\;0\;.
$$
Moreover, with the same argument used in the proof of \cite[Theorem 12]{denittis-schulz-baldes-16} one gets
$$
\begin{aligned}
K_0(\s{A}_\Lambda)\;&=\; K_0(\s{A}_0)\;\oplus\;\Z[\rr{p}_{\{0\}}]\;\simeq\;\Z^3,\\
K_1(\s{A}_\Lambda)\;&=\; K_1(\s{A}_0)\;\simeq\;\Z^2\;
\end{aligned}
$$  
 where $\s{A}_0$ is the magnetic algebra for a zero magnetic field (\cf Example \ref{ex_B=0}) and $\rr{p}_{\{0\}}$ is the projection on the  site $0\in\Z^2$
 (\cf Example \ref{ex:int_alg_loc}).\hfill $\blacktriangleleft$
 \end{example}

\medskip

The case of  straight-line interface (Definition \ref{def:typ_inter}) will be relevant in Section \ref{sect:iwatsuka}. Its $K$-theory is described below.
\begin{proposition}[$K$-theory for the straight-line interface]
\label{prop:kteor_intr1D}
In the case of a straight-line interface $\s{I}\simeq\s{C}(\n{S}^1)\otimes\s{K}(\ell^2(\Z))$ the $K$-theory is given by
$$
K_0(\s{I})\;\simeq\;\Z\;,\qquad K_1(\s{I})\;\simeq\;\Z\;.
$$
\end{proposition}
\proof
The result follows from the stability property of $K$-theory 
along with 
$K_0(\s{C}(\n{S}^1))\simeq\Z[{\bf 1}]$ and $K_1(\s{C}(\n{S}^1))\simeq\Z[u]$ where $u(k)=\expo{\ii k}$ \cite[Section 6.5]{wegge-olsen-93}. 
\qed

\subsection{Bulk and interface currents}
\label{sec:Bulk_interface_currents}
Let $\s{A}_{A_{B}}$ be a magnetic algebra endowed with the trace $\bb{T}_{\n{P}}$ associated to an ergodic measure $\n{P}\in{\rm Erg}( \Omega_B)$ as discussed in Section \ref{sect:integr}.
Given a  differentiable projection $\rr{p}\in \s{C}^1(\s{A}_{A_B})$, the (generalized) \emph{transverse Hall conductance} associated to $\rr{p}$ is defined by
\begin{equation}\label{eq:kub_01}
\sigma_{B,\n{P}}(\rr{p})\;:=\; \frac{e^2}{h}\;{\rm Ch}_{B,\n{P}}(\rr{p})
\end{equation}
where $e$ is the electron charge, $h=2\pi\hslash$ is the Planck's constant and the dimensionless part, known as  \emph{Chern number},  is given by
\begin{equation}\label{eq:chern_01}
{\rm Ch}_{B,\n{P}}(\rr{p})\;:=\;\ii 2\pi\;\bb{T}_{\n{P}}\big(\rr{p}\big[\nabla_1(\rr{p}),\nabla_2(\rr{p})\big]\big)\;.
\end{equation}
The projection $\rr{p}$ is usually obtained as the spectral projection into a gap of a self-adjoint element (Hamiltonian) of $\s{A}_{A_{B}}$ and represents the ground state of the system 
as described by the Fermi-Dirac distribution in the limit of the temperature $T=0$ and chemical potential (Fermi energy) sited into the gap. 
The quantity \eqref{eq:kub_01} enters in the (microscopic) Ohm's law
\begin{equation}\label{eq:kub_02}
J_\bot\;=\;\sigma_{B,\n{P}}(\rr{p})\;E
\end{equation}
which describes the transverse current density $J_\bot$ 
generated in the material as a response to the external electric perturbation $E$. The expression \eqref{eq:kub_02} is usually known as \emph{Kubo's formula} and is obtained in the linear response approximation. There are countless derivations of the Kubo's formula \eqref{eq:kub_02} in the literature. For our aims we will refer to \cite{bellissard-elst-schulz-baldes-94,schulz-baldes-bellissard-98} for the case of a constant magnetic field 
and to \cite{denittis-lein-book} for more general cases.

\medskip

In the case of a constant magnetic field $B$ of strength $b$ there is a unique ergodic measure (\cf Example \ref{ex:cost_hull})
and the associated trace, simply denoted with $\bb{T}$, is given by the trace per unit volume as proved in Proposition \ref{prop:tr_u_vol}.
Therefore, it is appropriate to rewrite equations \eqref{eq:kub_01} and 
\eqref{eq:chern_01} with the lighter notation 
\begin{equation}\label{cond_cost_eq}
\sigma_{b}(\rr{p})\;=\;\frac{e^2}{h}\;{\rm Ch}_{b}(\rr{p})
\end{equation}
In particular, the map ${\rm Ch}_{b}$ can be obtained from the trilinear map
$\xi_b:\s{C}^1(\s{A}_{b})^{\times 3} \to\C$, defined by
\begin{equation}\label{eq:chern_02}
\xi_b(\rr{a}_0,\rr{a}_1,\rr{a}_2)\;:=\;\ii 2\pi\;\bb{T}\big(\rr{a}_0\big(\nabla_1(\rr{a}_1)\nabla_2(\rr{a}_2)-\nabla_2(\rr{a}_1)\nabla_1(\rr{a}_2)\big)\big)\;,
\end{equation}
according to ${\rm Ch}_{b}(\rr{p})=\xi_b(\rr{p},\rr{p},\rr{p})$.
Formula \eqref{eq:chern_02} is crucial in the study of the topology of the algebra $\s{A}_{b}$ (which coincides with the noncommutative torus). In fact, as discussed in \cite[Chapter 3]{connes-94}, \cite[Chapter 12]{gracia-varilly-figueroa-01} or \cite[Chapter 5]{prodan-schulz-baldes-book} among others,
it turns out that the map $\xi_b$ is a cyclic 2-cocycle of the $C^*$-algebra $\s{A}_{b}$ and therefore defines a class $[\xi_b]\in HC^2(\s{A}_{b})$ in the cyclic cohomology of $\s{A}_{b}$. 
The class $[\xi_b]$ plays a special role in the canonical bilinear pairing 
$$
\prec\;,\;\succ\;:\;K_0(\s{A}_{b})\;\times\; HC^2(\s{A}_{b})\;\longrightarrow\;\C
$$
between (even) $K$-theory and (even) cyclic cohomology, defined by
$$
\big([\rr{p}],[\varphi]\big)\;\longmapsto\;\prec[\rr{p}],[\varphi]\succ\;:=\;({\rm tr}\sharp\varphi)(\rr{p},\rr{p},\rr{p})
$$
where the projection $\rr{p}\in \s{C}^1(\s{A}_{b})\otimes{\rm Mat}_N(\C)$ is a representative of the class $[\rr{p}]$, $N\in\N$ is a suitable integer\footnote{For a \emph{non-trivial} magnetic field $b(2\pi)^{-1}\in\R\setminus\Z$ it is always possible to fix $N=1$ since the $K$-theory is entirely realized inside the algebra $\s{A}_{b}$ (\cf Appendix \ref{app:Ktheo NCT}).} and ${\rm tr}$ denotes the trace on ${\rm Mat}_N(\C)$ \cite[Theorem 5.1.4]{prodan-schulz-baldes-book}. In the case  $N=1$,
a comparison with equations \eqref{eq:chern_01} and \eqref{eq:chern_02} shows that 
\begin{equation}\label{eq:chern_03}
{\rm Ch}_{b}(\rr{p})\;=\;\prec[\rr{p}],[\xi_b]\succ\;\in\Z\;\;\qquad 
\end{equation}
where the integrality of the pairing $[\rr{p}]\mapsto\prec[\rr{p}],[\xi_b]\succ$ 
is the celebrated \emph{index theorem} for the even $K$-theory \cite[Section 3.3, Corollary 16]{connes-94}.
 Equation \eqref{eq:chern_03} along with \eqref{eq:kub_01} provides the quantization (in units of $e^2h^{-1}$) of the transverse Hall conductance for a constant magnetic field \cite{thouless-kohmoto-nightingale-nijs-82,bellissard-elst-schulz-baldes-94}.

\medskip

The conductance for the bulk algebra  \ref{eq:ex_seq_021_0_int_mult-02} can be defined (by linearity)
from the case of a constant magnetic field.
\begin{definition}[Bulk transverse conductance]\label{def:cond_bulk}
Let $\s{A}_{\rm bulk}$ be the bulk algebra defined in equation \eqref{eq:ex_seq_021_0_int_mult-02} and $\rr{p}:=(\rr{p}_1,\ldots,\rr{p}_{N+1})$ a projection in $\s{C}^1(\s{A}_{\rm bulk})$. The \emph{bulk transverse conductance} for the projection $\rr{p}$ is given by the collection
$$
\sigma_{\rm bulk}(\rr{p})\;:=\; \big\{\sigma_{b_1}(\rr{p}_1),\ldots,\sigma_{b_{N+1}}(\rr{p}_{N+1})\big\}
$$
where every $\sigma_{b_j}(\rr{p}_j)$ is defined by \eqref{cond_cost_eq}.
\end{definition}

\medskip

Let us now consider the current   associated with the interface algebra $\s{I}$. We will focus on the case described by the ansatz \eqref{eq:ansatz_inter} and we will assume that the 
unital and abelian $C^*$-algebra $\s{I}_0$ is endowed with a faithful (normalized) trace $\tau_0$ and a suitable (unbounded) derivation
$\delta_0$
which meet the compatibility condition $\tau_0\circ \delta_0=0$. In this way one can define a faithful lower-semicontinuous trace $\bb{T}_{\s{I}}$ on $\s{I}$ through the prescription
$$
\bb{T}_{\s{I}}(\rr{a})\;:=\;\tau_0\otimes{\rm Tr}_{\s{H}_{\rm red}}\big(\varpi(\rr{a})\big)\;,\qquad \rr{a}\in\bb{D}_{\s{I}}
$$
where the ideal $\bb{D}_{\s{I}}\subset \s{I}$ is defined by
$\bb{D}_{\s{I}}:=\varpi^{-1}(\s{I}_0\otimes \s{L}^1(\s{H}_{\rm red}))$ and $\s{L}^1(\s{H}_{\rm red})$ is the ideal of trace class operators on  $\s{H}_{\rm red})$. 
Similarly, one can endow $\s{I}$ with the derivation $\nabla_{\s{I}}$ given by
$$
\nabla_{\s{I}}(\rr{a})\;:=\;\delta_0\otimes{\rm Id}_{\s{K}}
\big(\varpi(\rr{a})\big)\;,\qquad \rr{a}\in\s{C}^1_{\s{I}}
$$
where $\s{C}^k_{\s{I}}:=\varpi^{-1}(\s{C}^k(\s{I}_0)\otimes \s{K}(\s{H}_{\rm red}))$ for every $k\in\N$. Therefore, such a derivation  can be extended to the unitalization $\s{I}^+$ by the prescription $\nabla_{\s{I}}({\bf 1})=0$. With these structures one can define the map
\begin{equation}\label{eq:wind_numb}
W_{\s{I}}(\rr{u})\;:=\;\ii\bb{T}_{\s{I}}\big((\rr{u}^*-{\bf 1})\;\nabla_{\s{I}}(\rr{u}-{\bf 1})\big)\;=\;\ii\bb{T}_{\s{I}}\left(\rr{u}^*\nabla_{\s{I}}(\rr{u})\right)
\end{equation}
for every unitary operator $\rr{u}\in\s{I}^+$  such that $\rr{u}-{\bf 1}\in \s{C}^1_{\s{I}}\cap \bb{D}_{\s{I}}$. The map $W_{\s{I}}$ is known as the \emph{(non-commutative) winding number of} $\rr{u}$.
\begin{example}[Triviality of the winding number in the localized case]
According to Example \eqref{ex:int_alg_loc} the structure of the interface algebra in the case of a localized magnetic field is given by $\s{I}\simeq \C\otimes \s{K}(\ell^2(\Z^2))$ and therefore it satisfies the ansatz \eqref{eq:ansatz_inter}.
However, in view of the simple structure of $\s{I}_0=\C$ one has that the only faithful (normalized) trace $\tau_0$ is the identity $\tau_0(a)=a$ and the only derivation $\delta_0$ is the null-map   $\delta_0(a)=0$ for all $a\in\C$. As a consequence the 
associated trace on $\s{I}$ coincides with the canonical trace of the Hilbert space $\ell^2(\Z^2)$,  while there is no non-trivial derivation compatible with the ansatz \eqref{eq:ansatz_inter}. In view of that one has that $W_{\s{I}}=0$ identically in the case of a localized magnetic field.
\hfill $\blacktriangleleft$\end{example}
\begin{definition}[Interface conductance]\label{def:cond_inter}
Let $\s{I}$ be an interface algebra of type \ref{eq:ansatz_inter} endowed with the derivation $\nabla_{\s{I}}$ and the trace  $\bb{T}_{\s{I}}$. Let $\rr{u}\in\s{I}^+$ be a unitary operator such that
$\rr{u}-{\bf 1}\in \s{C}^1_{\s{I}}\cap \bb{D}_{\s{I}}$.
The \emph{interface conductance} associated to the \emph{configuration} $\rr{u}$ is defined by
\begin{equation}\label{eq:int_coduct}
\sigma_{\s{I}}(\rr{u})\;:=\; \frac{ e^2}{h}\;W_{\s{I}}(\rr{u})\;.
\end{equation}
\end{definition}

 \medskip

The definition  above is justified by the fact that $\sigma_{\s{I}}$ provides the proportionality coefficient for the current that flows along the interface (\cf \cite{schulz-baldes-kellendonk-richter-00} or \cite[Section 7.1]{prodan-schulz-baldes-book}). 
To clarify Definition \ref{eq:int_coduct} we need   some intermediate concept.

\medskip

Let us call \emph{magnetic Hamiltonians}  the
self-adjoint elements of $\s{A}_{A_B}$. Let $\hat{\rr{h}}\in \s{A}_{A_B}$ be a magnetic Hamiltonian and ${\rr{h}}:={\rm ev}(\hat{\rr{h}})\in \s{A}_{\rm bulk}$ its image in the bulk algebra.
By construction the bulk Hamiltonian $\rr{h}=(\rr{h}_1,\ldots,\rr{h}_{N+1})$ is made by a $N+1$-upla of suitable self-adjoint elements of the constant magnetic field algebras $\s{A}_{b_j}$  and its spectrum is given by $\sigma(\rr{h})=\bigcup_{j=1}^{N+1}\sigma(\rr{h}_j)$.
\begin{definition}[Non-trivial bulk gap]\label{def_bulk_gap}
The magnetic Hamiltonian $\hat{\rr{h}}\in\s{A}_{A_B}$ has a  \emph{non-trivial bulk gap} if there is a compact set $\Delta\in\R$ such that 
 $$
 \min \sigma({\rr{h}})\; <\;\min \Delta\;<\;\max \Delta
 \;<\;\max\sigma({\rr{h}})\;
 $$
 and $\Delta\cap\sigma({\rr{h}})=\emptyset$.
\end{definition}

\medskip

According to the above definition  $\Delta$ lies inside a non-trivial spectral gap of the bulk Hamiltonian ${\rr{h}}$ and for every chemical potential $\mu\in\Delta$ the \emph{Fermi projection}
$$
\rr{p}_\mu\;=\;(\rr{p}_{\mu,1},\ldots,\rr{p}_{\mu,N+1})\;\in\;\s{A}_{\rm bulk}\;,\qquad\rr{p}_{\mu,j}\;:=\;\chi_{(-\infty,\mu)}(\rr{h}_j)\;\in\;  \s{A}_{b_j}
$$
is an element of the bulk algebra. If the {magnetic Hamiltonian} is smooth, \ie $\hat{\rr{h}}\in \s{A}_{\rm I}^\infty$, then  also ${\rr{h}}\in\s{A}_{\rm bulk}^\infty$ (the evaluation map preserves the regularity), and in turn 
$\rr{p}_\mu\in\s{A}_{\rm bulk}^\infty$ since $\s{A}_{\rm bulk}^\infty$ is closed under holomorphic functional calculus. 
Let $[\rr{p}_\mu]=[(\rr{p}_{\mu,1},\ldots,\rr{p}_{\mu,N+1})]\in K_0(\s{A}_{\rm bulk})$ be the class of the 
Fermi projection in the $K_0$-group of $\s{A}_{\rm bulk}$. As a first step let us compute the image of $[\rr{p}_\mu]$ inside $K_1(\s{I})$ under the exponential map.
\begin{proposition}\label{prop:exp:map_interf}
Assume that the magnetic Hamiltonian $\hat{\rr{h}}\in\s{A}_{A_B}$ has a  \emph{non-trivial bulk gap} detected by $\Delta$.
Let $g:\R\to[0,1]$ be a non-decreasing (smooth)  function such that $g=0$ below $\Delta$ and $g=1$ above $\Delta$ and consider the unitary operator 
\begin{equation}\label{eq:pseudo_unit_XX}
\rr{u}_\Delta\;:=\;\expo{\ii2\pi g(\hat{\rr{h}})}\;.
\end{equation}
 Then $\rr{u}_\Delta\in\s{I}^+$ and
$$
{\rm exp}([\rr{p}_\mu])\;=\;-\left[\rr{u}_\Delta\right]\;\in\;K_1(\s{I}) \;.
$$
\end{proposition}
\proof
The proof is similar to that of \cite[Proposition 4.3.1]{prodan-schulz-baldes-book}.
Since the evaluation map is a homomorphism of $C^*$-algebras it commutes with the functional calculus and consequently 
$$
{\rm ev}(g(\hat{\rr{h}}))\;=\;g(\rr{h})\;=\;(g(\rr{h}_1),\ldots,g(\rr{h}_{N+1}))\;=\;{\bf 1}\;-\;\rr{p}_\mu\;
$$
due to the fact that $g$ is equal to 0 below the bulk gap and to 1 above the bulk gap and therefore $g(\rr{h}_j)={\bf 1}-\rr{p}_{\mu,j}$ .
As a consequence 
$$
{\rm ev}\left({\bf 1}- \expo{\ii 2\pi g(\hat{\rr{h}})}\right)\;=\;{\bf 1}-\expo{\ii 2\pi g({\rr{h}})}\;=\;0
$$
showing that $\rr{u}_\Delta$ is a unitary element in $\s{I}^+$.
Since ${\bf 1}- g(\hat{\rr{h}})$ is a self-adjoint lift of $\rr{p}_\mu$ one can compute the exponential map as in 
\cite[Definition 9.3.1 \& Exercise 9.E]{wegge-olsen-93} obtaining 
in this way
$$
{\rm exp}([\rr{p}_\mu])\;=\;\left[\expo{-\ii 2\pi ({\bf 1}-g(\hat{\rr{h}}))}\right]\;=\; \left[\expo{-\ii 2\pi g(\hat{\rr{h}})}\right]\;=\;-\left[\rr{u}_\Delta\right]\;
$$
where the additive notation\footnote{\label{note_ad_not}In terms of the additive notation of $K_1(\s{I})$, the trivial element is $[{\bf 1}]=0$   and $-[\rr{u}]=[\rr{u}^*]$ denotes the inverse of $[\rr{u}]$.} for $K_1(\s{I})$ has been used.
\qed

\medskip

In the case $\hat{\rr{h}}\in\s{A}_{A_B}^\infty$ it follows from the construction that $\rr{u}_\Delta\in\s{I}^+\cap\s{A}_{A_B}^\infty$ acquires the same regularity.
It is worth noting that the element ${\bf 1}- \rr{u}_\Delta$ can be constructed entirely from the spectral subspace of $\hat{\rr{h}}$ corresponding to the bulk insulating gap $\Delta$. Indeed, the support of the function $\expo{\ii 2\pi g}-1$ is 
contained inside  the region $\Delta$ which lies in  the insulating gap.

\begin{remark}[Gap closing as a topological obstraction]
The condition $[\rr{u}_\Delta]\neq0 $ (\cf Note \ref{note_ad_not})  measures the obstruction to lift the  Fermi projection $\rr{p}_\mu\in\s{A}_{\rm bulk}$ to a true projection in $\s{A}_{A_B}\otimes \rm{Mat}_N(\C)$ (for some $N$ large enough). From the construction emerges that this  obstruction  detects the presence of spectrum of $\hat{\rr{h}}$ inside $\Delta$ which is generated by the existence of the magnetic interface. Since the election of $\Delta$ inside the bulk gap is totally arbitrary, and the Fermi projection does not depend on the specific $\mu$ inside the bulk gap, one gets that 
for any given $\Delta$ the related element ${\bf 1}-g(\hat{\rr{h}})$ is a self-adjoint lift of the Fermi projection. This implies immediately that the condition $[\rr{u}_\Delta]\neq 0$ guarantees the complete closure of the bulk gap due to the presence of the magnetic interface. \hfill $\blacktriangleleft$
\end{remark}

\medskip

Let $g$ as in the claim of Proposition \ref{prop:exp:map_interf}.
The derivative $g'$ is non-negative, supported in $\Delta$ and normalized in the sense that $\|g'\|_{L^1}=1$. By construction the element $g'(\hat{\rr{h}})$  satisfies the condition ${\rm ev}(g'(\hat{\rr{h}}))=g'({\rm ev (\hat{\rr{h}})})=0$ and so 
$g'(\hat{\rr{h}})\in\s{I}$ is an element of the interface algebra. Moreover $g'(\hat{\rr{h}})$ can be regarded as a density matrix which describes a state of the system with energy distributed in the region $\Delta$. If one interprets the operator $\hslash^{-1}\nabla_{\s{I}}(\hat{\rr{h}})$ as the \emph{velocity} along the interface one deduces that
\begin{equation}\label{ewq:int_condI}
J_{\s{I}}(\Delta)\;:=\;-\frac{e}{\hslash}\bb{T}_{\s{I}}\left(g'(\hat{\rr{h}})\;\nabla_{\s{I}}(\hat{\rr{h}})\right)
\end{equation}
is the \emph{current density} along the interface carried by the \virg{extended states} in $\Delta$ and, as a consequence, $\sigma_{\s{I}}=e J_{\s{I}}$ provides the associated  conductance
(we are assuming that $e>0$ is the magnitude of the electron charge). 
The connection between the latter formula and Definition \ref{def:cond_inter}
is provided by the following result originally proved in   \cite{schulz-baldes-kellendonk-richter-00}.
\begin{proposition}\label{prop:int_condI}
It holds true that
$$
\bb{T}_{\s{I}}\left(g'(\hat{\rr{h}})\;\nabla_{\s{I}}(\hat{\rr{h}})\right)\;=\;-\frac{1}{2\pi}\; W_{\s{I}}\big(\rr{u}_\Delta\big)\;.
$$
\end{proposition}
\proof
The result can be obtained by adapting step by step the proof of \cite[Proposition 7.1.2]{prodan-schulz-baldes-book}. Indeed the proof is purely algebraic and only uses the properties of the trace $\bb{T}_{\s{I}}$ and the derivation $\nabla_{\s{I}}$
assumed by hypothesis at the beginning of this section.
\qed

\medskip

By combining definition \ref{ewq:int_condI} (which is motivated by physics) with Proposition \ref{prop:int_condI} one gets that the 
interface conductance generated by the \virg{extended states} in $\Delta$ is given by
\begin{equation}\label{ewq:int_condII}
\sigma_{\s{I}}(\Delta)\;:=\;\frac{ e^2}{h}\;W_{\s{I}}(\rr{u}_\Delta)\;.
\end{equation}
This equation justifies the \virg{abstract} Definition \ref{def:cond_inter}.

\medskip

The relevance of Definition \ref{def:cond_inter} lies in its topological interpretation.
Consider the map $\eta_{\s{I}}:(\s{C}^1({\s{I}})\cap \bb{D}_{\s{I}})^{\times 2} \to\C$, defined by
\begin{equation}\label{eq:chern_0dd_03}
\eta_{\s{I}}(\rr{b}_0,\rr{b}_1)\;:=\;\ii \;\bb{T}_{\s{I}}\big(\rr{b}_0\; \nabla_{\s{I}}(\rr{b}_1)\big)\;.
\end{equation}
In view of the properties of $\bb{T}_{\s{I}}$ and  $\nabla_{\s{I}}$
assumed by hypothesis,  $\eta_{\s{I}}$ turns out to be a cyclic 1-cocycle  and therefore defines a class $[\eta_{\s{I}}]\in HC^1(\s{I})$ in the cyclic cohomology of  the $C^*$-algebra $\s{I}$ \cite[Chapter 3]{connes-94}. 
Let 
$$
\prec\;,\;\succ\;:\;K_1(\s{I})\;\times\; HC^1(\s{I})\;\longrightarrow\;\C
$$
be the  canonical bilinear pairing
between (odd) $K$-theory and (odd) cyclic cohomology, defined by
$$
\big([\rr{u}],[\phi]\big)\;\longmapsto\;\prec[\rr{u}],[\phi]\succ\;:=\;({\rm tr}\sharp\phi)\big((\rr{u}-{\bf 1})^*,\rr{u}-{\bf 1}\big)
$$
where the unitary $\rr{u}\in \s{C}^1(\s{I}^+)\otimes{\rm Mat}_N(\C)$ is any representative of the class $[\rr{u}]$. \cite[Section 3.3, Proposition 3]{connes-94}.
Since every unitary $\rr{u}\in\s{I}^+$ (like $\rr{u}_\Delta$) defines an element $[\rr{u}]\in K_1(\s{I})$
in the $K_1$-group of the interface algebra, one gets that 
\begin{equation}\label{eq:chern_03_BIS}
W_{\s{I}}(\rr{u})\;=\;\prec[\rr{u}],[\eta_{\s{I}}]\succ
\end{equation}
only depends on the class  $[\rr{u}]$. In particular, by combining together 
Proposition \ref{prop:exp:map_interf} and equation \eqref{eq:chern_03_BIS}
one gets
\begin{equation}\label{ewq:int_condIIII}
\sigma_{\s{I}}(\Delta)\;:=\;\frac{ e^2}{h}\;\prec[\rr{u}_\Delta],[\eta_{\s{I}}]\succ\;=\;-\frac{ e^2}{h}\;\prec {\rm exp}([\rr{p}_\mu]),[\eta_{\s{I}}]\succ
\;.
\end{equation}

\medskip

Equation \eqref{ewq:int_condIIII} is the topological essence of the 
\emph{bulk-interface duality} and will be used in Sections \ref{sect:bulk-int-Iwatsuka}  to prove equation \eqref{eq:int_01}  in the case of the   
Iwatsuka magnetic field (\cf Theorem \ref{theo_main_Iw}).

\section{The Iwatsuka \texorpdfstring{$C^*$-}-algebra}\label{sect:iwatsuka}
This section is devoted to the detailed study of the 
Toeplitz extension and the $K$-theory of the Iwatsuka $C^*$-algbera. The magnetic translations associated to the Iwatsuka magnetic field has been described in 
Example \ref{Ex2:MT_Iwatsuka} and the Iwatsuka $C^*$-algbera has been defined in Example \ref{ex_B=0}.

\subsection{Toeplitz extension for the Iwatsuka magnetic field}\label{sect:toep_iwatsuka_MT}
The simplest examples of a 
magnetic multi-interface system as described in Definition \ref{def:bulk_alg} is provided by the
 Iwatsuka magnetic $B_{\rm I}$ defined by
\eqref{eq:watsuka_B}  with the conditions
\begin{equation}\label{eq:nontriv_iwat}
b_-\;\neq\; b_+\;,\qquad \; b_\pm\notin\R\setminus2\pi\Z\;.
\end{equation}
In fact, according to the content of Example 
\ref{ex:hull_Iwatsuka} one has that the boundary of the Iwatsuka magnetic hull $\Omega_{\rm I}$ can be represented as 
$\partial\Omega_{\rm I}=\{\omega_{-\infty},\omega_{+\infty}\}$
with $\omega_{\pm\infty}$ two distinct invariant points. As a consequence the associated  Toeplitz extension is given by 
\begin{equation}\label{eq:ex_seq_021_0_int_mult_iwatz}
0\;\stackrel{}{\longrightarrow}\;\s{I}\;\stackrel{\imath}{\longrightarrow}\;\s{A}_{\rm I}\;\stackrel{\rm ev}{\longrightarrow}\s{A}_{\rm bulk}\;\stackrel{}{\longrightarrow}0\;
\end{equation}
with  \emph{bulk algebra}  given by
\begin{equation}\label{eq:ex_seq_021_0_int_mult-02_iwatz}
\s{A}_{\rm bulk}\;:=\;\s{A}_{b_-}\;\oplus\;\s{A}_{b_{+}}
\end{equation}
and evaluation map defined by
\begin{equation}\label{eq:evalumap_iwatz}
\begin{aligned}
{\rm ev}(\rr{s}_{{\rm I},1})\;:&=\;(\rr{s}_{b_-,1},\rr{s}_{b_+,1})\\
{\rm ev}(\rr{s}_{{\rm I},2})\;:&=\;(\rr{s}_{b_-,2},\rr{s}_{b_+,2})\\
{\rm ev} (\rr{f}_{\rm I})\;:&=\;\left(\expo{\ii b_-}{\bf 1},\expo{\ii b_{+}}{\bf 1}\right)
\end{aligned}
\end{equation}
where $\rr{s}_{{\rm I},1}$ and $\rr{s}_{{\rm I},2}$ are the 
Iwatsuka magnetic translations and 
$\rr{f}_{\rm I}$ is the associated flux operator as defined in Example \ref{Ex2:MT_Iwatsuka}.
\begin{remark}[Interpretation of the evaluation map]\label{rk:gen_lim}
It is worth  interpreting the action of the evaluation map on the 
commutative subalgebra $\s{F}_{\rm I}$ generated by $\rr{f}_{\rm I}$ (\cf eq. \eqref{eq:def_F_B}) as a generalized limit. Let $\rr{g}\in \s{F}_{\rm I}$ and ${g}\in \s{C}(\Omega_{\rm I})$ its Gelfand transform as a continuous function on the hull $\Omega_{\rm I}$. According to the discussion in Example \ref{ex:hull_Iwatsuka},   $\s{C}(\Omega_{\rm I})$ coincides with the the $C^*$-subalgebra of $\s{C}_{\rm b}(\Z^2)$ of sequences that admit left and right limits.  Then, it follows that 
\begin{equation}\label{eq:gen_lim}
{\rm ev} (\rr{g})\;=\;\left(\lim\limits_{s\rightarrow -\infty}g(a,s),\lim\limits_{s\rightarrow +\infty}g(b,s)\right)
\end{equation}
for every $a,b\in\Z$. It is immediate to check that equation \eqref{eq:gen_lim} holds true on the dense subalgebra of  $\s{F}_{\rm I}$ generated by finite linear combinations of elements of the type
$
(\rr{s}_{{\rm I},1})^{r}\;(\rr{s}_{{\rm I},2})^{s}\;\rr{f}_{\rm I}\;(\rr{s}_{{\rm I},2})^{-s}\;(\rr{s}_{{\rm I},1})^{-r}=
(\rr{s}_{{\rm I},1})^{r}\;\rr{f}_{\rm I}\;(\rr{s}_{{\rm I},1})^{-r}\;
$
and the result follows from a standard density argument.\hfill $\blacktriangleleft$
\end{remark}

\subsection{Interface algebra for the  Iwatsuka magnetic field}\label{sect:interf_iwatsuka_MT}
The Iwatsuka $C^*$-algbera
 $\s{A}_{\rm I}$ contains several interesting projections. Let us introduce the projections 
$$
\begin{aligned}
(\rr{p}_\pm\psi)(n)\;:&=\;\delta_\pm(n)\;\psi(n)\\
(\rr{p}_0\psi)(n)\;:&=\;\delta_0(n)\;\psi(n)\\
\end{aligned}\;,
\qquad \psi\in \ell^2(\Z^2)\;
$$
where the functions $\delta_\pm$ and $\delta_0$ are defined in Example \ref{Ex2:Iwatsuka}.
\begin{lemma}\label{lemm:proj_1}
Under the assumption \eqref{eq:nontriv_iwat}
  the projections $\rr{p}_\pm$ and $\rr{p}_0$ are elements of  $\s{A}_{\rm I}$.
\end{lemma}
\proof
The identity $\bf{1}$ and the flux operator $\rr{f}_{\rm I}$ defined by \eqref{eq:flux_op} are elements of $\s{A}_{\rm I}$. Let us start  with the case  $b_0\neq b_+$ and $b_0\neq b_-$.
A straightforward computation shows that
$$
\rr{p}_0\;=\;\left(\expo{\ii b_-}-\expo{\ii b_0}\right)^{-1}\left(\expo{\ii b_+}-\expo{\ii b_0}\right)^{-1}\left(\expo{\ii b_-}{\bf 1}-\rr{f}_{\rm I}\right)\left(\expo{\ii b_+}{\bf 1}-\rr{f}_{\rm I}\right)\;,
$$
hence $\rr{p}_0\in \s{A}_{\rm I}$. Similarly, one can check that 
$$
\rr{p}_\pm\;=\;\left(\expo{\ii b_\mp}-\expo{\ii b_\pm}\right)^{-1}\left(\expo{\ii b_\mp}{\bf 1}-\rr{f}_{\rm I}\right)\left({\bf 1}-\rr{p}_0\right)\;.
$$
Let us assume now  $b_0=b_+$ and consider the projection $\rr{p}_\geqslant:=\rr{p}_0+\rr{p}_+$.  One can check as above that
\begin{equation}\label{eq:aux_001}
\begin{aligned}
\rr{p}_\geqslant\;&=\;\left(\expo{\ii b_-}-\expo{\ii b_+}\right)^{-1}\left(\expo{\ii b_-}{\bf 1}-\rr{f}_{\rm I}\right)\\
\rr{p}_-\;&=\;\left(\expo{\ii b_+}-\expo{\ii b_-}\right)^{-1}\left(\expo{\ii b_+}{\bf 1}-\rr{f}_{\rm I}\right)
\end{aligned}
\end{equation}
are both elements of $\s{A}_{\rm I}$. Moreover, the equality
\begin{equation}\label{eq:aux_002}
\rr{p}_0\;=\; \rr{p}_\geqslant\;-\; \rr{s}_{{\rm I},1}\;\rr{p}_\geqslant\;\rr{s}_{{\rm I},1}^*
\end{equation}
shows that also $\rr{p}_0\in \s{A}_{\rm I}$. Finally $\rr{p}_+=\rr{p}_\geqslant-\rr{p}_0$. The case $b_0=b_-$ is similar.
\qed

\medskip

For every $j\in\Z$ let us introduce  the projection 
$$
(\rr{p}_j\psi)(n)\;:=\;\delta_0(n-je_1)\;\psi(n)\;
\qquad \psi\in \ell^2(\Z^2)\;.
$$
From the definition it follows that $\rr{p}_j$ is the translation of $\rr{p}_0$ along the vertical line located at $n_1=j$. The projections $\rr{p}_j$ are mutually orthogonal.
\begin{corollary}\label{corol:proj_1}
Under the assumption \eqref{eq:nontriv_iwat} it holds true that $\rr{p}_j\in\s{A}_{\rm I}$ for all $j\in\Z$.
\end{corollary}
\proof
From Lemma \ref{lemm:proj_1} we know that $\rr{p}_0\in\s{A}_{\rm I}$. Moreover, a direct computation shows that
 \begin{equation}\label{eq_p_j}
\rr{p}_{j}\;=\; \left\{
\begin{aligned}
&(\rr{s}_{{\rm I},1})^{j}\;\rr{p}_0\;(\rr{s}_{{\rm I},1}^*)^{j}&\quad&\text{if}\;\; j>0\\
&(\rr{s}_{{\rm I},1}^*)^{|j|}\;\rr{p}_0\;(\rr{s}_{{\rm I},1})^{|j|}&\quad&\text{if}\;\; j<0\;.
\end{aligned}
\right.
\end{equation}
This completes the proof.
\qed

\medskip

From \eqref{eq_p_j} one gets the useful formula
\begin{equation}\label{eq_p_j_bis}
\rr{p}_{j}\; \rr{s}_{{\rm I},1}\;=\;\rr{s}_{{\rm I},1}\;\rr{p}_{j-1}\;,\qquad\quad j\in\Z\;.
\end{equation}
The next result provides a first step for the description
 of the  evaluation map.
\begin{lemma}\label{lemma:ev_01}
Under the assumption \eqref{eq:nontriv_iwat} it holds true that
\begin{equation}\label{eq:ev_01}
\begin{aligned}
&{\rm ev}(\rr{p}_{+})&=&\;\;(0,{\bf 1})\\
&{\rm ev}(\rr{p}_{-})&=&\;\;({\bf 1},0)\;
\end{aligned}
\end{equation}
and
\begin{equation}\label{eq:ev_0101}
{\rm ev}(\rr{p}_{j})\;=\;\;\;(0,0)\;,\qquad \forall\; j\in\Z\;.
\end{equation}
\end{lemma}
\proof
Let us start with the case   $b_0\neq b_+$ and $b_0\neq b_-$.
Then the result follows from the last equation in \eqref{eq:evalumap_iwatz}, the formulas for $\rr{p}_\pm$  and $\rr{p}_j$  in Lemma \ref{lemm:proj_1} and Corollary \ref{corol:proj_1} along with the fact that ${\rm ev}$
is a $C^*$-homomorphism. In the case $b_0=b_+$ one obtains from \eqref{eq:aux_001} that 
${\rm ev}(\rr{p}_{\geqslant})=(0,{\bf 1})$ and ${\rm ev}(\rr{p}_{-})=({\bf 1},0)$
. Moreover, from \eqref{eq:aux_002} one gets that 
$$
{\rm ev}(\rr{p}_{+})\;=\;(0,{\bf 1})\;-\;(0,\rr{s}_{b_+,1}{\bf 1}\rr{s}_{b_+,1}^*)\;=\;0\;.
$$
The case $b_0=b_-$ is similar.
\qed

\medskip

Let $\Sigma\subset\Z$ be a finite subset and
\begin{equation}\label{eq_p_j_02}
\rr{p}_\Sigma\;:=\;\bigoplus_{j\in\Sigma}\rr{p}_j\;.
\end{equation}
the next result is a direct consequence of Lemma \ref{lemma:ev_01}.
\begin{corollary}\label{corol:ev_01}
Under the assumption \eqref{eq:nontriv_iwat} it holds true that 
$$
{\rm ev}(\rr{a}\rr{p}_\Sigma\rr{b})\;=\;0
$$
for all $\rr{a},\rr{b}\in\s{A}_{\rm I}$ and for all finite subset $\Sigma\subset\Z$. 
\end{corollary}

\medskip

Elements of the type $\rr{a}\rr{p}_\Lambda\rr{b}$ can be considered as \virg{localized} operators (in the direction $e_1$) and
Corollary \ref{corol:ev_01} establishes  that localized elements are in the kernel of the evaluation map, namely they are elements of the interface algebra $\s{I}$ in view of Theorem \ref{theo:SES}. We are now in position to provide a useful characterization of the {interface algebra}.

\begin{proposition}\label{pro_intAlg_iwa}
The \emph{interface algebra} $\s{I}$ is  the closed  two-sided ideal  of $\s{A}_{\rm I}$ generated by $\rr{p}_{0}$, \ie
$$
\s{I}\;=\;\s{A}_{\rm I}\;\rr{p}_{0}\;\s{A}_{\rm I}\;:=\;{\rm span}\left\{\rr{a}\rr{p}_{0}\rr{b}\;|\; \rr{a},\rr{b}\in \s{A}_{\rm I}\right\}\;.
$$
\end{proposition}
\proof
A comparison with Definition \ref{def:interf} shows  the  claim is equivalent to state that $\rr{p}_{0}$ generates ${\rm Ker}({\rm ev}|_{\s{F}_{\rm I}})$. From Corollary \ref{corol:ev_01} one gets that $\rr{p}_{j}\in {\rm Ker}({\rm ev}|_{\s{F}_{\rm I}})\subset \s{F}_{\rm I}$ for every $j\in\Z$. A close look at the construction of $\s{F}_{\rm I}$ shows
 that every $\rr{g}\in \s{F}_{\rm I}$ admits the (unique) representation 
$$
\rr{g}\;=\;\sum_{j\in\Z}g_j\; \rr{p}_{j}\;.
$$
where the sequence $\{g_j\}\in \s{C}_{\rm b}(\Z)$ 
admits left and right limits (\cf Remark \ref{rk:gen_lim}). Therefore, one gets that $\rr{g}\in {\rm Ker}({\rm ev}|_{\s{F}_{\rm I}})$ if and only if the associated sequence $\{g_j\}$ vanishes at infinity, \ie if and only if   $\{g_j\}\in\s{C}_0(\Z)$. The proof is completed by observing that $\s{C}_0(\Z)$ is the uniform closure of the sequences with compact support on $\Z$.
\qed

\medskip

The Iwatsuka magnetic field is constant in one direction and therefore one can use the  magnetic Bloch-Floquet transform described in Appendix \ref{app:bloch-floquet} to study the interface algebra. Indeed, the Iwatsuka magnetic translations commute with the operator $V_f:=\expo{\ii f(\rr{n}_1)}\rr{s}_2$ defined through the function
$$
f(m)\;:=\;
\left\{
\begin{aligned}
&m b_+\quad        &\text{if}&\;m\geqslant0\\
&(m+1)b_- -b_0     &\text{if}&\;m<0\;.
\end{aligned}
\right.
$$
Let $\s{U}_B:\ell^2(\Z^2)\to L^2(\n{S}^1)\otimes\ell^2(\Z)$ be the associated magnetic Bloch-Floquet transform as defined in Appendix \ref{app:bloch-floquet}. The next result contains the main feature of the Iwatsuka interface algebra.
\begin{proposition}\label{prop:iwats_in_stri}
It holds true that $\s{U}_B\s{I}\s{U}_B^{-1}=\s{C}(\n{S}^1)\otimes\s{K}(\ell^2(\Z))$. In particular the Iwatsuka interface algebra is a \emph{straight-line} according to Definition \ref{def:typ_inter}.
\end{proposition}
\proof
A direct computation shows that $\s{U}_B\rr{p}_j\s{U}_B^{-1}={\bf 1}\otimes \pi_j$ where $\pi_j$ is the rank-one projection on $\ell^2(\Z)$ defined by $(\pi_j\phi)(m):=\delta_{m,j}\phi(m)$.
Since $\s{U}_B(\rr{s}_{{\rm I},2}\rr{p}_j)^n\s{U}_B^{-1}=\s{U}_B(\rr{s}_{{\rm I},2})^n\rr{p}_j\s{U}_B^{-1}$ is proportional to $\expo{\ii n k}\otimes \pi_j$ up to a phase factor one gets that $g\otimes \pi_j\in \s{U}_B\s{I}\s{U}_B^{-1}$ for every $g\in \s{C}(\n{S}^1)$ and $j\in\Z$. Acting with powers of $\s{U}_B\rr{s}_{{\rm I},1}\s{U}_B^{-1}$ on the latter elements one gets that also $g\otimes \pi_{i,j}\in \s{U}_B\s{I}\s{U}_B^{-1}$ where $\pi_{i,j}$ is the rank-one operator defined by
$(\pi_{i,j}\phi)(m):=\delta_{m,j}\phi(i)$. The result follows by observing that  the rank-one operators generates the compact operators.
\qed

\medskip

Following the procedure described in Section \ref{sec:Bulk_interface_currents} we can use Proposition \ref{prop:iwats_in_stri} to equip $\s{I}$ with a derivation and a trace. The natural derivation on $\s{C}(\n{S}^1)$ is $\delta_0:=-\frac{\dd}{\dd k}$. With this sign convention 
a comparison with \eqref{eq:deriv_02}
provides
$$
\delta_0\otimes{\rm Id}_{\s{K}}\left(\s{U}_B\rr{a}\s{U}_B^{-1}\right)\;=\;\nabla_2(\rr{a})\;=\;\ii[\rr{a},\rr{n}_2]
$$
for differentiable elements $\rr{a}\in \s{I}$. Therefore we
obtain that the interface derivation is given by $\nabla_{\s{I}}:=\ii[\;\cdot\;,\rr{n}_2]$. Similarly the natural trace on $\s{C}(\n{S}^1)$ is given by $\tau_0:=\int_{\n{S}^1}\dd k$ where 
$\dd k$ is the normalized Haar measure. Since $\tau_0(\expo{\ii n k})=\delta_{n,0}$ one gets that 
$$
\tau_0\otimes{\rm Tr}_{\ell^2(\Z)}\left(\s{U}_B\rr{a}\s{U}_B^{-1}\right)\;=\;{\rm Tr}_{\ell^2(\Z^2)}(\rr{q}_0\rr{a}\rr{q}_0)
$$
where the projection $\rr{q}_0$ is given by $(\rr{q}_0\psi)(n,m)=\delta_{m,0}\psi(n,m)$
and $\rr{a}\in\s{I}$ is any suitable integrable elements. In this way one can define the interface trace as
\begin{equation}\label{eq:tra_int}
\bb{T}_{\s{I}}(\rr{a})\;:=\;{\rm Tr}_{\ell^2(\Z^2)}(\rr{q}_0\rr{a}\rr{q}_0)\;=\;\sum_{n\in\Z}\bra{n,0}\,\rr{a}\,\ket{n,0}
\end{equation}
where the Dirac notation in the right-hand side turns out to be particularly useful.

\subsection{Linear space splitting of the  Toeplitz extension}\label{sect:split_Toeplitz}
The 
Toeplitz extension for the Iwatsuka magnetic field admits a natural splitting of the  linear space structure which turns out to be useful in applications. Such a fact has already been anticipated at the end of Section \ref{sect:toep_inter}.
 
 \medskip

 Let us start by recalling that 
$\s{A}_{\rm bulk}$ is generated, as $\ast$-linear space, by the linear combinations of monomials  of the type $(\rr{s}_{b_-,1}^r\rr{s}_{b_-,2}^s,\rr{s}_{b_+,1}^p\rr{s}_{b_+,2}^q)$ with $r,s,p,q\in\Z$. 
Consider the \emph{linear} map $\jmath:\s{A}_{\rm bulk}\to \s{A}_{\rm I}$ initially defined on the monomials by
\begin{equation}\label{eq:j_split}
\jmath\big(\rr{s}_{b_-,1}^r\rr{s}_{b_-,2}^s,\rr{s}_{b_+,1}^p\rr{s}_{b_+,2}^q\big)\;:=\;
\rr{p}_{-}\;\rr{s}_{{\rm I},1}^r\; \rr{s}_{{\rm I},2}^s\; \rr{p}_{-}\;+\;
\rr{p}_{+}\;\rr{s}_{{\rm I},1}^p\; \rr{s}_{{\rm I},2}^q\; \rr{p}_{+}
\end{equation}
and then extended linearly to $\s{A}_{\rm bulk}$. 
Such a  map is well defined because both $\s{A}_{\rm bulk}$ and 
$\s{A}_{\rm I}$ are spanned as Banach spaces by the families of respective  monomials.
From it very definition it follows that ${\rm ev}\circ \jmath={\rm Id}_{\s{A}_{\rm bulk}}$, namely $\jmath$ provides a 
splitting of the linear structures. It follows that 
$$
\s{A}_{\rm I}\;=\;\s{I}\;+\; \jmath\big(\s{A}_{\rm bulk}\big)
$$
as direct sum of linear spaces \cite[Proposition 3.1.3]{wegge-olsen-93}.

\medskip

It is worth noting that the \emph{linear} map $\jmath$ defined by \eqref{eq:j_split} cannot be extended to a $C^*$-homomorphism. For instance, a direct computation shows that 
$$
\begin{aligned}
\jmath\big({\bf 1},\rr{s}_{b_+,1}\big)\jmath\big({\bf 1},\rr{s}_{b_+,1}^*\big)\;-\;\jmath\big({\bf 1},{\bf 1}\big)\;&=\;\rr{p}_{+}\big(\rr{s}_{b_+,1}\rr{p}_{+}\rr{s}_{b_+,1}^*-{\bf 1}\big)\rr{p}_{+}\;=\;-\rr{p}_{1}\\
\end{aligned}
$$
since $\rr{s}_{b_+,1}\rr{p}_{+}\rr{s}_{b_+,1}^*=\rr{p}_{+}-\rr{p}_{1}$.
On the other hand, 
$$
\begin{aligned}
\jmath\big({\bf 1},\rr{s}_{b_+,1}^*\big)\jmath\big({\bf 1},\rr{s}_{b_+,1}\big)\;-\;\jmath\big({\bf 1},{\bf 1}\big)\;&=\;\rr{p}_{+}\big(\rr{s}_{b_+,1}^*\rr{p}_{+}\rr{s}_{b_+,1}-{\bf 1}\big)\rr{p}_{+}\;=\;0\\
\end{aligned}
$$
due to $\rr{s}_{b_+,1}^*\rr{p}_{+}\rr{s}_{b_+,1}=\rr{p}_{+}+\rr{p}_{0}$.

\begin{remark}[Failure of the $C^*$-splitting]\label{rk:split}
A linear splitting is the best that we can do since the existence of a $C^*$-lifting  would imply
 the short exact sequence (see \cite[Proposition 8.2.2]{wegge-olsen-93} or \cite[Proposition 3.29]{gracia-varilly-figueroa-01})
\begin{equation}\label{fake-short-exact-sequence}
0\longrightarrow K_0(\s{I})\overset{\imath_*}{\longrightarrow}K_0(\s{A}_{\rm I})\overset{{\rm ev}_*}{\longrightarrow}K_0(\s{A}_{\rm bulk})\longrightarrow 0,
\end{equation}
 at the level of the $K_0$-groups. However, it will be proved that $K_0(\s{I})\simeq \Z$ (Proposition \ref{prop:k-group-Interface}) and  $\imath_*=0$ (Remark \ref{rk:i=0}) making the 
short exact sequencee \eqref{fake-short-exact-sequence}
just impossible.\hfill $\blacktriangleleft$
\end{remark}

\subsection{\texorpdfstring{$K$-}-theory for the Iwatsuka \texorpdfstring{$C^*$-}-algebra}\label{sect:K-Iwatsuka}
In this section we will provide a preliminary study of the $K$-theory of the   Iwatsuka $C^*$-algebra which will be complemented in next Section \ref{sect:6t-K-Iwatsuka}. We will make use of the fact that the  $C^*$-algebra $\s{A}_{\rm I}$ can be represented as an iterated crossed product with $\Z$ (see Appendix \ref{sect:crossedproduct}), and in turn we will exploit the Pimsner-Voiculescu exact sequence (Appendix \ref{app:PV})  described in \cite{pimsner-voiculescu-exact-seq} or in \cite[Chapter V]{blackadar-98}.

\medskip

By adapting the notation of Appendix \ref{sect:crossedproduct} we have the isomorphisms
$$
\s{A}_{\rm I} \;\simeq\; \s{Y}_{{\rm I},1}\rtimes_{\alpha_2}\Z\;,\qquad \s{Y}_{{\rm I},1}\;:=\; \s{F}_{\rm I}\rtimes_{\alpha_1}\Z$$
where $\s{F}_{\rm I}$ is the $C^*$-algebra generated by the the flux operator $\rr{f}_{\rm I}$ according to \eqref{eq:def_F_B}, the automorphism $\alpha_1$ is defined by 
$\alpha_1(\rr{g}):=\rr{s}_{\rm I,1}\rr{g}\rr{s}_{\rm I,1}^*$ for every $\rr{g}\in \s{F}_{\rm I}$ and the automorphism $\alpha_2$ is defined by $\alpha_2(\rr{g}\rr{s}_{\rm I,1}^r):=\rr{s}_{\rm I,2}\rr{g}\rr{s}_{\rm I,1}^r\rr{s}_{\rm I,2}^*$
for every $\rr{g}\in \s{F}_{\rm I}$ and $r\in\N_0$. 

\medskip

The $K$-theory of the $C^*$-algebra $\s{F}_{\rm I}$ is calculated in Appendix \ref{app:K-theoIMH} and is given by
\begin{equation}
\begin{aligned}
K_0(\s{F}_{\rm I})\;&=\;\left(\bigoplus_{i\in\Z}\Z[\rr{p}_j]\right)\oplus\Z[\rr{p}_-]\oplus\Z[\rr{p}_+]\;,\\
K_1(\s{F}_{\rm I})\;&=\;0\;.
\end{aligned}
\end{equation}
The $K$-theory of the first crossed product $\s{Y}_{{\rm I},1}$ can be  computed from the Pimsner-Voiculescu exact sequence

\begin{equation}\label{eq:6termexactseq_PV_Iwa1}
\begin{array}{ccccc}
K_0(\s{F}_{\rm I})&\overset{\beta_{1,\ast}}{\longrightarrow}          
 &K_0(\s{F}_{\rm I}) &      \overset{{\imath}_*}{\longrightarrow}    
  &K_0(\s{Y}_{{\rm I},1})\\
  &&&&\\
{\partial_1}\Big\uparrow   &&&&\Big\downarrow{\partial_0}\\
   &&&&\\
  K_1(\s{Y}_{{\rm I},1})& \underset{{\imath}_*}{\longleftarrow}       
 &K_1(\s{F}_{\rm I}) &   \underset{\beta_{1,\ast}}{\longleftarrow}        
  &K_1(\s{F}_{\rm I})\\
\end{array}
\end{equation}
where  the connecting maps $\beta_{1,\ast}$ and $\imath_*$ and the
 boundary maps $\partial_0$ and
$\partial_1$ are described in Appendix \ref{app:PV}.
\begin{proposition}\label{prop:PV_1}
Consider the six-term exact sequence \eqref{eq:6termexactseq_PV_Iwa1}. Then, it holds true that:
	\begin{itemize}
		\item[(i)] The image and kernel of the map $\beta_{1,*}:K_0(\s{F}_{\rm I})\rightarrow K_0(\s{F}_{\rm I})$ are given by
	$$
	{\rm Im}(\beta_{1,*})\;=\;\bigoplus_{j\in\Z}\Z[\rr{p}_j]\;,\qquad {\rm Ker}(\beta_{1,*})\;=\;\Z[{\bf 1}]\;.
	$$	
\vspace{1mm}		
		\item[(ii)]  $\partial_1[\rr{s}_{{\rm I},1}]=-[\bf{1}]$.
	\end{itemize}
	Consequently,
	$$K_0(\s{Y}_{{\rm I},1})\;=\;\Z[\rr{p}_-]\oplus \Z[\rr{p}_+]\;,\qquad  K_1(\s{Y}_{{\rm I},1})\;=\;\Z[\rr{s}_{{\rm I},1}]\;.$$
\end{proposition}
\proof
For (i) let us recall that $\beta_{1,*}={\rm Id}_*-\alpha_{1,*}^{-1}$, as described in Appendix \ref{app:PV}. Therefore, one gets
	\begin{equation*}
	\begin{aligned}
	\beta_{1,*}\big([\rr{p}_j]\big)&\;=\; [\rr{p}_j-\rr{s}_{\rm I,1}^*\rr{p}_j\rr{s}_{\rm I,1}]\;=\;[\rr{p}_j]-[\rr{p}_{j-1}]\;,\\
	\beta_{1,*}\big([\rr{p}_-]\big)&\;=\; [\rr{p}_--\rr{s}_{\rm I,1}^*\rr{p}_-\rr{s}_{\rm I,1}]\;=\; [\rr{p}_{-1}]\;,\\
	\beta_{1,*}\big([\rr{p}_+]\big)&\;=\; [\rr{p}_+-\rr{s}_{\rm I,1}^*\rr{p}_+\rr{s}_{\rm I,1}]\;=\;-[\rr{p}_{0}]\;.\\
	\end{aligned}
	\end{equation*}
	It follows that the image of $\beta_{1,*}$ is $\bigoplus_{j\in\Z}\Z[\rr{p}_j]$ and
	$$
	\beta_{1,*}\left(n_-[\rr{p}_-]+n_+[\rr{p}_+]+\sum_{j=-M}^{+M}n_j[\rr{p}_j]\right)\;=\;0
	$$
	has a non-trivial solution if and only if $n_-=n_0=n_+$, and $n_j=0$ in all other cases.
	As a consequence one has that the kernel of $\beta_{1,*}$ is generated by $[\rr{p}_-]+[\rr{p}_0]+[\rr{p}_+]=[{\bf 1}]$.
	For (ii) let us recall that the boundary map $\partial_1:=k_0^{-1}\circ{\rm ind}$ is the composition of the index map ${\rm ind}:K_1(\s{Y}_{{\rm I},1})\to K_0(\s{F}_{\rm I}\otimes\s{K})$ associated to the Toeplitz extension \eqref{eq:actual_toeplitz_ext_app} and the inverse of the stabilization isomorphism $\kappa_0:K_0(\s{F}_{\rm I})\to K_0(\s{F}_{\rm I}\otimes\s{K})$
	induced by the identification $\rr{g}\mapsto \rr{g}\otimes \pi_0$ (here $\pi_0\in \s{K}$ is any fixed rank-one projection). The isometry $V:=\rr{s}_{\rm I,1}\otimes v$ which generates the 	Toeplitz algebra $\s{T}_{\alpha_1}$ together with $\s{F}_{\rm I}\otimes{\bf 1}$
verifies the condition $\psi(V)=\rr{s}_{\rm I,1}$. Therefore,
$V$ provides a lift of $\rr{s}_{\rm I,1}$ by an isometry. Consider the unitary matrix
$$
w(\rr{s}_{\rm I,1})\;:=\;\begin{pmatrix}
	V        & P\\
	0 & V^*
	\end{pmatrix}\;\in\;{\rm Mat}_2(\s{T}_{\alpha_1})\;
$$	
where $P:={\bf 1}-VV^*$. By construction $w(\rr{s}_{\rm I,1})$ is a lift of ${\rm diag}(\rr{s}_{\rm I,1},\rr{s}_{\rm I,1}^*)$ and $[{\rm diag}(\rr{s}_{\rm I,1},\rr{s}_{\rm I,1}^*)]=[{\bf 1}]$	as a class in $K_1(\s{Y}_{{\rm I},1})$. As a consequence we can construct the index map according to \cite[Definition 8.1.1]{wegge-olsen-93} and after an explicit computation one gets
$$
\begin{aligned}
{\rm ind}\big([\rr{s}_{\rm I,1}]\big)\;&=\;\varphi_*^{-1}([{\bf 1}-V^*V]-[{\bf 1} -VV^*])\\
&=\;\varphi_*^{-1}([0]-[P])\;=\;-[{\bf 1}\otimes \pi_0]\;
\end{aligned}
$$
where in the last equality  we used the property $\varphi({\bf 1}\otimes \pi_0)=P$. By using the isomorphism $\kappa_0$, one finally gets 
$\partial_1[\rr{s}_{{\rm I},1}]=-[\bf{1}]$. Since $K_1(\s{F}_{\rm I})=0$, it follows that $\partial_1$ provides an isomorphism between $K_1(\s{Y}_{{\rm I},1})$ and ${\rm Ker}(\beta_{1,*})$. In view of (i) and (ii) one infers that $K_1(\s{Y}_{{\rm I},1})=\Z[\rr{s}_{\rm I,1}]$. Again, $K_1(\s{F}_{\rm I})=0$ implies the surjectivity of $\imath_*:K_0(\s{F}_{\rm I})\to K_0(\s{Y}_{{\rm I},1})$ and so $K_0(\s{Y}_{{\rm I},1})\simeq K_0(\s{F}_{\rm I})/{\rm Im}(\beta_{1,*})=\Z[\rr{p}_-]\oplus \Z[\rr{p}_+]$.
\qed

\medskip

For the $K$-theory of the second crossed product $\s{A}_{\rm I} \simeq \s{Y}_{{\rm I},1}\rtimes_{\alpha_2}\Z$ we need the 
Pimsner-Voiculescu exact sequence
\begin{equation}\label{eq:6termexactseq_PV_Iwa2}
\begin{array}{ccccc}
K_0(\s{Y}_{{\rm I},1})&\overset{\beta_{2,\ast}}{\longrightarrow}          
 &K_0(\s{Y}_{{\rm I},1}) &      \overset{{\imath}_*}{\longrightarrow}    
  &K_0(\s{A}_{\rm I})\\
  &&&&\\
{\partial_1}\Big\uparrow   &&&&\Big\downarrow{\partial_0}\\
   &&&&\\
  K_1(\s{A}_{\rm I})& \underset{{\imath}_*}{\longleftarrow}       
 &K_1(\s{Y}_{{\rm I},1}) &   \underset{\beta_{2,\ast}}{\longleftarrow}        
  &K_1(\s{Y}_{{\rm I},1})\\
\end{array}
\end{equation}
\begin{theorem}[$K$-theory of the Iwatsuka $C^*$-algebra I]\label{prop:PV_2} 
Consider the six-term exact sequence \eqref{eq:6termexactseq_PV_Iwa2}. Then, it holds true that
	\begin{itemize}
		\item[(i)] Both maps $\beta_{2,\ast}:K_j(\s{Y}_{{\rm I},1})\rightarrow K_j(\s{Y}_{{\rm I},1})$, with $j=1,2$, vanish;
		\vspace{1mm}
		\item[(ii)] The map $\partial_1$ verifies
		$$
		\begin{aligned}
		\partial_1\big([\rr{p_-}\rr{s}_{{\rm I},2}+\rr{p}_0+\rr{p}_+]\big)\;=\;-[\rr{p}_-]\;,\\
		\partial_1\big([\rr{p_-}+\rr{p}_0+\rr{p}_+\rr{s}_{{\rm I},2}]\big)\;=\;-[\rr{p}_+]\; ;
		\end{aligned}
		$$
\vspace{1mm}
		\item[(iii)] There exists $N\in\N$ and a projection $\rr{p}_{\rm I}\in \s{A}_{\rm I}\otimes{\rm Mat}_N(\C)$ such that 
		$$ 
		\partial_0[\rr{p}_{\rm I}]=[\rr{s}_{{\rm I},1}]\;.
		$$
		\end{itemize}
	Consequently,
	$$
	\begin{aligned}
	K_0(\s{A}_{\rm I})\;&=\;\Z[\rr{p}_-]\oplus\Z[\rr{p}_-]\oplus\Z[\rr{p}_{\rm I}]\\
	K_1(\s{A}_{\rm I})\;&=\;\Z[\rr{w}_{{\rm I},-}]\oplus\Z[\rr{w}_{{\rm I},+}]\oplus\Z[\rr{s}_{{\rm I},1}]\;
		\end{aligned}
	$$
	where $\rr{w}_{{\rm I},\pm}:={\bf 1}+\rr{p}_\pm(\rr{s}_{{\rm I},2}-{\bf 1})$.
	\end{theorem}
\proof For (i) it is enough to note that $\alpha_2(\rr{p}_l)=\rr{p}_l$ for $l\in\Z\cup\{\pm\}$ and
	$$\alpha_{2,*}^{-1}[\rr{s}_{\rm I,1}]\;=\;[\rr{s}_{\rm I,2}^*\rr{s}_{\rm I,1}\rr{s}_{\rm I,2}]\;=\;[(\rr{s}_{\rm I,2}^*\rr{f}_{\rm I}\rr{s}_{\rm I,2})\rr{s}_{\rm I,1}]\;=\;[\rr{f}_{\rm I}][\rr{s}_{\rm I,1}]\;=\;[\rr{s}_{\rm I,1}]\;,
	$$
	since $[\rr{f}_{\rm I}]=[{\bf 1}]$ in $K_1(\s{Y}_{\rm I,1})$. As a consequence $\beta_{2,*}={\rm Id}_*-\alpha_{2,*}^{-1}=0$.	
For (ii) let us observe  that the isometry $V=\rr{s}_{\rm I,2}\otimes v\in \s{T}_{\alpha_2}$ satisfies
$\psi((\rr{p}_+\otimes {\bf 1})V)=\rr{p}_+\rr{s}_{\rm I,2}$.	
It follows that $W:=(\rr{p}_-\otimes {\bf 1})V+(\rr{p}_0+\rr{p}_+)\otimes {\bf 1}$ is an isometry which provides a lift of 	
	$\rr{p}_-\rr{s}_{\rm I,2}+\rr{p}_0+\rr{p}_+$ in $\s{T}_{\alpha_2}$. The index map of the latter element can be computed as in the proof of Proposition \ref{prop:PV_1} and  after some computation one gets
$$
\begin{aligned}
{\rm ind}\big([\rr{p}_-\rr{s}_{\rm I,2}+\rr{p}_0+\rr{p}_+]\big)\;&=\;\varphi_*^{-1}([{\bf 1}-W^*W]-[{\bf 1} -WW^*])\\
&=\;\varphi_*^{-1}([0]-[\rr{p}_-\otimes P])\;=\;-[\rr{p}_-\otimes \pi_0]\;.
\end{aligned}
$$	
where we used $P:={\bf 1}-VV^*$ and $\varphi(\rr{p}_-\otimes \pi_0)=\rr{p}_-\otimes P$. 
After recalling that $\partial_1:=k_0^{-1}\circ{\rm ind}$, with $\kappa_0$ stabilization isomorphism, one gets  
the first equation in (ii). The derivation of the second  equation is identical.
Item (iii) follows from the fact that  $\beta_{2,*}=0$ implies the surjectivity of the map $\partial_0$ and so there must be a projection $\rr{p}_{\rm I}\in \s{A}_{\rm I}\otimes {\rm Mat}_N(\C)$ and $M\leq N$ such that
$$\partial_0([\rr{p}_{\rm I}]-M[{\bf 1}])=[\rr{s}_{{\rm I},1}]$$
(see \cite[Proposition 6.2.7]{wegge-olsen-93}). Now, since $[{\bf 1}]=[\rr{p}_{-}]+[\rr{p}_+]\in {\rm Im}(\imath_*)=\ker \partial_0$ it follows that $\partial_0[\rr{p}_{\rm I}]=[\rr{s}_{{\rm I},1}]$.
The exactness of the sequence \eqref{eq:6termexactseq_PV_Iwa2} along  with  $\beta_{2,*}=0$ implies
$K_j(\s{A}_{\rm I})=\imath_*(K_j(\s{Y}_{{\rm I},1}))\oplus\partial_j^{-1}(K_{j+1}(\s{Y}_{{\rm I},1}))$
with $j=0,1$ (mod. 2). This concludes the proof.
\qed

\medskip

Observe that the proof for the existence of the element $\rr{p}_{\rm I}$ works as well for the case $\partial_0(\rr{p}_{\rm I}')=[\rr{s}_{\rm I,1}^*]=-[\rr{s}_{\rm I,1}]$. Any projection $\rr{p}_{\rm I}$ with the property 
$$\partial_0(\rr{p}_{\rm I})\in \{ [\rr{s}_{\rm I,1}],[\rr{s}_{\rm I,1}^*] \}$$
will be called a \emph{Power-Rieffel-Iwatsuka projection} or simply a \emph{PRI-projection}. We can say a little more about $\rr{p}_{\rm I}$. From its very definition one has that 
$\exp[\rr{p}_{\rm I}]=[(\rr{s}_{{\rm I},1}-{\bf 1})\otimes \pi_{0}+\bf 1]$
where $\exp$ is the actual exponential map associated with the Toeplitz exact sequence \eqref{eq:actual_toeplitz_ext_app}.

\subsection{Six-term exact sequence for the Iwatsuka \texorpdfstring{$C^*$-}-algebra}\label{sect:6t-K-Iwatsuka}
We are now in position to study the {six-term} exact sequence
 associated with the Toeplitz extension for the Iwatsuka magnetic field \eqref{eq:ex_seq_021_0_int_mult_iwatz}. This is given by%
\begin{equation}\label{eq:6-ex-seq_Iwa}
\begin{array}{ccccc}
K_0(\s{I})&\overset{\imath_*}{\longrightarrow}          
 &K_0(\s{A}_{\rm I}) &      \overset{{\rm ev}_*}{\longrightarrow}    
  &K_0(\s{A}_{\rm bulk})\\
  &&&&\\
\mbox{\footnotesize ind}\Big\uparrow   &&&&\Big\downarrow \mbox{\footnotesize exp}\\
   &&&&\\
  K_1(\s{A}_{\rm bulk})& \underset{{\rm ev}_*}{\longleftarrow}       
 &K_1(\s{A}_{\rm I}) &   \underset{\imath_*}{\longleftarrow}        
  &K_1(\s{I})\\
\end{array}
\end{equation}
The $K$-theory of the bulk algebra is described in Proposition \ref{prop:K-bulk} and is explicitly given by
\begin{equation*}
\begin{aligned}
K_0(\s{A}_{\text{bulk}})&\;=\;\Z[({\bf 1},0)]\oplus\Z[(\rr{p}_{\theta_-},0)]\oplus\Z[(0,{\bf 1})]\oplus\Z[(0,\rr{p}_{\theta_+})]\;,\\
K_1(\s{A}_{\text{bulk}})&\;=\;\Z[(\rr{s}_{b_-,1},{\bf 1})]\oplus\Z[(\rr{s}_{b_-,2},{\bf 1})]\oplus\Z[({\bf 1},\rr{s}_{b_+,1})]\oplus\Z[({\bf 1},\rr{s}_{b_+,2})]\;,
\end{aligned}
\end{equation*}
where $\rr{p}_{\theta_\pm}$ are the the Power-Rieffel projections of the $C^*$-algebras $\s{A}_{b_\pm}$, respectively (\cf Appendix \ref{app:Ktheo NCT}).

\medskip

The description of the $K$-theory of the interface algebra follows from Proposition \ref{prop:kteor_intr1D} and Proposition \ref{prop:iwats_in_stri}. 
\begin{proposition}\label{prop:k-group-Interface}
It holds true that
$$
\begin{aligned}
K_0(\s{I})&\;=\;\Z[\rr{p}_0]\;,\qquad\;
K_1(\s{I})&\;=\;\Z[ \rr{w}_{\s{I}}]\;,
\end{aligned}
$$
where $\rr{w}_{\s{I}}:=\rr{p}_-+\rr{p}_0\rr{s}_{{\rm I},2}+\rr{p}_+={\bf 1}+\rr{p}_0(\rr{s}_{{\rm I},2}-{\bf 1})\in\s{I}^+$.
\end{proposition}
\proof
Let us start with the $K_0$-group. As showed in the proof of Proposition \ref{prop:kteor_intr1D}
the generator of $K_0(\s{C}(\n{S}^1))$ is the constant function 1. The group isomorphism $K_0(\s{C}(\n{S}^1))\simeq K_0(\s{C}(\n{S}^1)\otimes\s{K}(\ell^2(\Z)))$ is induced by the $C^*$-homomorphism
$\mu:\s{C}(\n{S}^1)\to \s{C}(\n{S}^1)\otimes\s{K}(\ell^2(\Z))$ defined by $\mu:g\mapsto g\otimes \pi_0$ where $\pi_0$ is the projection on $\ell^2(\Z)$ defined by $(\pi_0\phi)(m):=\delta_{m,0}\phi(m)$ \cite[Corollary 6.2.11]{wegge-olsen-93}. The result follows by observing that $\s{U}_B^{-1}(1\otimes \pi_0)\s{U}_B=\rr{p}_0$ where $\s{U}_B$ is the magnetic Bloch-Floquet transform used in Proposition \ref{prop:iwats_in_stri}.
The argument for the $K_1$-group follows a similar structure.
We already know that the generator of $K_1(\s{C}(\n{S}^1))$ is the exponential function $\expo{\ii k}$ and the isomorphism 
$K_1(\s{C}(\n{S}^1))\simeq K_1(\s{C}(\n{S}^1)\otimes\s{K}(\ell^2(\Z)))$ is induced by the same isomorphism
$\mu$ defined above \cite[Proposition 8.2.8]{rordam-larsen-laustsen-00}. However, since the $K_1$ is computed from the unitalization of the related $C^*$-algebra one needs to promote $\expo{\ii k}\otimes\pi_0$ to a unitary in $(\s{C}(\n{S}^1)\otimes\s{K}(\ell^2(\Z)))^+$. This can be done  through the map
$$
\expo{\ii k}\otimes\pi_0\;\longmapsto\;\expo{\ii k}\otimes\pi_0\;-\;1\otimes\pi_0\;+\; 1\otimes{\bf 1}
$$ 
as described in \cite[Proposition 8.1.6]{rordam-larsen-laustsen-00}. As a result one has the generator of the $K_1$-group can be identified with the class of $(\expo{\ii k}-1)\otimes\pi_0+1\otimes{\bf 1}$. Finally,  the  magnetic Bloch-Floquet transform
$$
\s{U}_B^{-1}\big((\expo{\ii k}-1)\otimes\pi_0+1\otimes{\bf 1}\big)\s{U}_B\;=\;
V_f\rr{p}_0\;-\;\rr{p}_0\;+\;{\bf 1}
$$
along with the identities $V_f\rr{p}_0=\rr{s}_{2}\rr{p}_0=\rr{s}_{{\rm I},2}\rr{p}_0$ and ${\bf 1}-\rr{p}_0=\rr{p}_-+\rr{p}_+$
provides  the desired result
\qed

\medskip

We now have all the ingredients to study the vertical homomorphisms of the diagram \eqref{eq:6-ex-seq_Iwa}. Let us start with the index map..

\begin{proposition}\label{prop:indeximage}
	The image of the generators of $K_0(\s{A}_{\text{bulk}})$
	under the map ${\rm ind}$ in diagram \eqref{eq:6-ex-seq_Iwa} are given by
	\begin{equation}\label{ind_map}
	\begin{aligned}
	\rm{ind}([(\rr{s}_{b_-,2},{\bf 1})])&\;=\;\rm{ind}([({\bf 1},\rr{s}_{b_+,2})])\;=\; 0\;,\\
	\rm{ind}([(\rr{s}_{b_-,1},{\bf 1})])&\;=\;-\rm{ind}([({\bf 1},\rr{s}_{b_+,1})])\;=\;
	[\rr{p}_{0}]\;.\\
	\end{aligned}
	\end{equation}	
	Consequently the index map  is surjective.
\end{proposition}
\begin{proof}
Let us construct the index map according to \cite[Definition 8.1.1]{wegge-olsen-93} for the set of generators 
$A\in\{(\rr{s}_{b_-,1},{\bf 1}),(\rr{s}_{b_-,2},{\bf 1}),({\bf 1},\rr{s}_{b_+,1}),({\bf 1},\rr{s}_{b_+,2})\}\subset\s{A}_{\rm bulk}$ of the $K_1$-group of $\s{A}_{\rm bulk}$.
Let $\jmath$ as in $\eqref{eq:j_split}$ and define the map
	$$
	w(A)\;:=\;\begin{pmatrix}
	\jmath(A)         & {\bf 1}-\jmath(A)\jmath(A)^*\\
	{\bf 1}-\jmath(A)^*\jmath(A) & \jmath(A)^*
	\end{pmatrix}\;\in\;{\rm Mat}_2(\s{A}_{\rm I})\;.
	$$
	A direct check shows that $\jmath(A)\in\s{A}_{\rm I}$ is a partial isometry for every $A$ in the generator set, indeed

	\begin{equation}\label{eq:preindex}
	\begin{aligned}
	\jmath(\rr{s}_{b_-,1},{\bf 1})\jmath(\rr{s}_{b_-,1},{\bf 1})^*&\;=\;{\bf 1}-\rr{p}_0\;,\\
	\jmath(\rr{s}_{b_-,2},{\bf 1})\jmath(\rr{s}_{b_-,2},{\bf 1})^*&\;=\;{\bf 1}-\rr{p}_0\;,\\
	\jmath({\bf 1},\rr{s}_{b_+,1})\jmath({\bf 1},\rr{s}_{b_+,1})^*&\;=\;{\bf 1}-(\rr{p}_0+\rr{p}_{1})\;,\\
	\jmath({\bf 1},\rr{s}_{b_+,2})\jmath({\bf 1},\rr{s}_{b_+,2})^*&\;=\;{\bf 1}-\rr{p}_0\;,\\
\end{aligned}
	\end{equation}
and, on the other hand,	
	\begin{equation}\label{eq:preindex_2}
	\begin{aligned}
	\jmath(\rr{s}_{b_-,1},{\bf 1})^*\jmath(\rr{s}_{b_-,1},{\bf 1})&\;=\;{\bf 1}-(\rr{p}_0+\rr{p}_{-1})\;,\\
	\jmath(\rr{s}_{b_-,2},{\bf 1})^*\jmath(\rr{s}_{b_-,2},{\bf 1})&\;=\;{\bf 1}-\rr{p}_0\;,\\
	\jmath({\bf 1},\rr{s}_{b_+,1})^*\jmath({\bf 1},\rr{s}_{b_+,1})&\;=\;{\bf 1}-\rr{p}_0\;,\\
	\jmath({\bf 1},\rr{s}_{b_+,2})^*\jmath({\bf 1},\rr{s}_{b_+,2})&\;=\;{\bf 1}-\rr{p}_0\;.
	\end{aligned}
	\end{equation}
As a consequence one can check that  $w(A)$ is a unitary operator for every generator $A$. Moreover ${\rm ev}(w(A))={\rm diag}(A,A^*)$ showing that $w(A)$ is a \emph{unitary lift} of ${\rm diag}(A,A^*)$. Finally $[{\rm diag}(A,A^*)]\simeq[{\bf 1}]$ as a class in the $K_1$-group. With all these data we can compute the index map of each generators according to 
$\text{ind}([A]):=[w(A)P_1w(A)^*]-[P_1]$ where $P_1:=\text{diag}({\bf 1},0)$. An explicit computation provides
$$
\begin{aligned}
\text{ind}([A])\;&=\;\left[\begin{pmatrix}
	\jmath(A)\jmath(A)^*   & 0\\
	0            & {\bf 1}-\jmath(A)^*\jmath(A)
	\end{pmatrix}\right]\;-\;\left[\begin{pmatrix}
	{\bf 1}  & 0\\
	0            & 0
	\end{pmatrix}\right]\\
	&=\;\left[\begin{pmatrix}
	0   & 0\\
	0            & {\bf 1}-\jmath(A)^*\jmath(A)
	\end{pmatrix}\right]\;-\;\left[\begin{pmatrix}
	{\bf 1} -\jmath(A)\jmath(A)^* & 0\\
	0            & 0
	\end{pmatrix}\right]\\
	&=\;[{\bf 1}-\jmath(A)^*\jmath(A)]-[{\bf 1} -\jmath(A)\jmath(A)^*]
\end{aligned}
$$
where the second and third equality are understood in the sense
of the $K_0$-group. The equations \eqref{ind_map} follow from the latter formula along with the computations \eqref{eq:preindex} and \eqref{eq:preindex_2} and the observation that, in view of \eqref{eq_p_j},  $\rr{p}_j$ is unitarily equivalent to $\rr{p}_0$ for every $j\in\Z$. The latter fact implies  $[\rr{p}_j]=[\rr{p}_k]$ for every pair $j,k\in\Z$ as elements of $K_0(\s{A}_{\rm I})$.
Consequently,  $\text{ind}([\rr{s}_{b_-,1},{\bf 1}])=[\rr{p}_{0}]$, and  $[\rr{p}_{0}]$ is the generators of $K_0(\s{I})$. This shows that the map ${\rm ind}$ is surjective.
\end{proof}

\begin{remark}\label{rk:i=0}
 As a consequence of Proposition \ref{prop:indeximage} and the exactness of diagram \ref{eq:6-ex-seq_Iwa} one infers that the map $\imath_*:K_0(\s{I})\to          
 K_0(\s{A}_{\rm I})$
 is just the zero map.  This implies that $[\imath_*([\rr{p}_j])]=[\rr{p}_j]=0$ as  element of $K_0(\s{A}_{\rm I})$. This fact is in agreement with the description of $K_0(\s{A}_{\rm I})$ in Theorem \ref{prop:PV_2} and can be justified by the following direct argument: From $\rr{p}_0=\rr{s}^*_{{\rm I},1}\rr{p}_+\rr{s}_{{\rm I},1}-\rr{p}_+$ one gets
 $[\rr{p}_0]=[\rr{s}^*_{{\rm I},1}\rr{p}_+\rr{s}_{{\rm I},1}]-[\rr{p}_+]=[\rr{p}_+]-[\rr{p}_+]=0$ and $[\rr{p}_0]=[\rr{p}_j]$
for every $j\in\Z$ as justified at the end of the proof of Proposition \ref{prop:indeximage}. 
\hfill $\blacktriangleleft$
\end{remark}

Now we are in position to study the exponential map of diagram \eqref{eq:6-ex-seq_Iwa}. 

\begin{proposition}\label{prop:expimage}
The map ${\rm exp}$ in diagram \eqref{eq:6-ex-seq_Iwa} is surjective. Moreover, it holds true that
\begin{equation}\label{eq:exp_form}
\begin{aligned}
{\rm exp}([({\bf 1},0)])\;&=\;{\rm exp}([(0,{\bf 1})])\;=\;[{\bf 1}]\;=\;0\;,\\
{\rm exp}([(0,\rr{p}_{\theta_+})])\;&=\;-{\rm exp}([(\rr{p}_{\theta_-},0)])\;=\;-[ \rr{w}_{\s{I}}]\;,
\end{aligned}
\end{equation}
where the additive notation for the group  $K_1(\s{I})$ is used (\cf Note \ref{note_ad_not}) and  $[\rr{w}_{\s{I}}]$ denotes the generator of $K_1(\s{I})$ defined in Proposition \ref{prop:k-group-Interface}.
\end{proposition}
\proof
The surjectivity of the exponential map can be deduced directly by the exactness of the diagram \eqref{eq:6-ex-seq_Iwa}. Since ${\rm Ker}(\exp)\simeq K_0(\s{A}_{\rm I})\simeq\Z^3$ and $K_0(\s{A}_{\rm bulk})\simeq\Z^4$ it follows that there is an element $[q]\in K_0(\s{A}_{\rm bulk})$ such that ${\rm exp}([q])=m[ \rr{p}_-+\rr{s}_{{\rm I},2}\rr{p}_0+\rr{p}_+]\in K_1(\s{I})$ for some $m\in\Z\setminus 0$.
However, by observing that $K_1(\s{A}_{\rm I})\simeq\Z^3$ is torsion-free one infers that $m=\pm 1$ are the only admissible values. In both cases the exponential map turns out to be surjective. Now we can prove formulas \eqref{eq:exp_form}. The construction of the exponential map is described in 
 \cite[Definition 9.3.1 \& Exercise 9.E]{wegge-olsen-93}.
The first step is to construct appropriate lifts of the representatives of the elements of the group $K_0(\s{A}_{\rm bulk})$. Let us start with the two generators $({\bf 1},0)$ and $(0,{\bf 1})$. From Lemma \ref{lemma:ev_01} we get that suitable self-adjoint lifts are given by
${\rm lift}({\bf 1},0):=\rr{p}_-$ and ${\rm lift}(0,{\bf 1}):=\rr{p}_+$. Moreover, since $\rr{p}_\pm$ are genuine projections one gets $\expo{-\ii 2\pi \rr{p}_\pm}={\bf 1}\in\s{I}^+$. As a consequence, one gets the first equation in \eqref{eq:exp_form}. For the second set of equations we need to construct explicitly the element $[q]$ introduced abstractly above. We will follows quite closely the strategy in 
\cite[pp. 114-116]{pimsner-voiculescu-exact-seq}.
Let us start with the Power-Rieffel projection (\cf Appendix \ref{app:Ktheo NCT})
$$
\rr{p}_{\theta_+}\;=\;\rr{s}_{b_+,1}^*\;\rr{d}_1\;+\;\rr{d}_0\;+\;\rr{d}_1\;\rr{s}_{b_+,1}\;\in\;\s{A}_{b_+}
$$
where $\rr{d}_1:=g(\rr{s}_{2})$ and $\rr{d}_0:=f(\rr{s}_{2})$
are self-adjoint elements of $\s{A}_{b_+}\cap \s{A}_{\rm I}$
in view of $\rr{s}_{2}=\rr{s}_{b_+,2}=\rr{s}_{{\rm I},2}$.
Consider the self-adjoint lift of $(0,\rr{p}_{\theta_+})$ given by
$$
\rr{q}_{+}\;=\;\rr{v}_{+}^*\;\rr{d}_1\;+\;\rr{d}_0\rr{p}_\geqslant\;+\;\rr{d}_1\;\rr{v}_{+}\;,
$$
where  $\rr{v}_+:=\rr{s}_{{\rm I},1}\rr{p}_\geqslant=\rr{s}_{b_+,1}\rr{p}_\geqslant$ and $\rr{p}_\geqslant:=\rr{p}_0+\rr{p}_+$. It is worth remembering that $[\rr{d}_i,\rr{p}_0]=[\rr{d}_i,\rr{p}_+]=0$ for $i=0,1$.
A direct computation shows that
\begin{equation}\label{eq:lif_proj}
\rr{q}_{+}^2\;=\;\rr{q}_{+}\;-\;\rr{d}_1^2\rr{p}_0\;=\;\rr{q}_{+}\;+\;(\rr{d}_0^2-\rr{d}_0)\rr{L} \rr{p}_0\;.
\end{equation}
The first equality in \eqref{eq:lif_proj} is justified by the relations 
$$
\begin{aligned}
\rr{d}_1\rr{v}_{+}\rr{d}_1\rr{v}_{+}\;&=\;
\rr{v}_{+}(\rr{s}_{b_+,1}^*\rr{d}_1\rr{s}_{b_+,1}\rr{d}_1)\rr{v}_{+}\;=\;0\;,\\
\rr{d}_0\rr{p}_\geqslant\rr{d}_1\rr{v}_{+}+\rr{d}_1\rr{v}_{+}\rr{d}_0\rr{p}_\geqslant\;
&=\;(\rr{d}_0\rr{d}_1+\rr{d}_1\rr{s}_{b_+,1}\rr{d}_0\rr{s}_{b_+,1}^*)\rr{v}_{+}\;=\;\rr{d}_1\rr{v}_{+}\;,\\
\rr{d}_0^2\rr{p}_\geqslant+\rr{v}_{+}^*\rr{d}_1\rr{d}_1\rr{v}_{+}+
\rr{d}_1\rr{v}_{+}\rr{v}_{+}^*\rr{d}_1\;&=\;
\rr{p}_\geqslant(\rr{d}_0^2+\rr{s}_{b_+,1}^*\rr{d}_1^2\rr{s}_{b_+,1}+\rr{d}_1^2)\rr{p}_\geqslant-\rr{d}_1^2\rr{p}_0\\
&=\;\rr{d}_0\rr{p}_\geqslant-\rr{d}_1^2\rr{p}_0
\end{aligned}
$$
deduced from \eqref{eq:conPRproj}. The second equality in 
\eqref{eq:lif_proj} follows from \eqref{eq:ulteram}
where $\rr{L}:=\rr{L}(\rr{d}_1)$ is the  support projection of $\rr{d}_1$
(in the von Neumann algebra generated by $\s{A}_{b_+}$).
An inductive argument, based on the identities
$$
\rr{q}_{+}\rr{p}_0\;=\;\rr{d}_0\rr{p}_1+\rr{d}_1\rr{s}_{b_+,1}\rr{p}_0\;,\qquad \rr{d}_1\rr{s}_{b_+,1}\rr{L}\;=\;0
$$
and the commutation relations $[\rr{L},\rr{d}_i]=0=[\rr{L},\rr{p}_0]$,
provides
\begin{equation}\label{eq:lif_proj-2}
\rr{q}_{+}^N\;=\;\rr{q}_{+}\;+\;(\rr{d}_0^N-\rr{d}_0)\rr{L} \rr{p}_0\;=\;
(\rr{q}_{+}-\rr{d}_0\rr{L} \rr{p}_0)\;+\;(\rr{d}_0\rr{L})^N \rr{p}_0\;.
\end{equation}
Equation \eqref{eq:lif_proj-2} facilitates the computation of the exponential of  $\rr{q}_{+}$. Indeed, one immediately gets
$$
\begin{aligned}
\expo{-\ii 2\pi \rr{q}_+}\;&=\;\sum_{N=0}^{+\infty}\frac{(-\ii 2\pi)^N}{N!}\rr{q}_+^N\\
&=(\expo{-\ii 2\pi}-1)(\rr{q}_+-\rr{d}_0\rr{L}\rr{p}_0)\;+\;\expo{-\ii 2\pi \rr{d}_0\rr{L}}\rr{p}_0\;+\;({\bf 1}-\rr{p}_0)\\
&=\rr{p}_-\;+\;\expo{-\ii 2\pi \rr{d}_0\rr{L}}\rr{p}_0\;+\;\rr{p}_+\;.
\end{aligned}
$$ 
Finally,  by using the homotopy   $\expo{-\ii 2\pi \rr{d}_0\rr{L}}\sim\rr{s}_{b_+,2}^*=\rr{s}_{{\rm I},2}^*$ described in Lemma \ref{lemma:homot_1} one obtains
${\rm exp}([(0,\rr{p}_{\theta_+})])=-[ \rr{p}_-+\rr{s}_{{\rm I},2}\rr{p}_0+\rr{p}_+]$.
The proof for $(\rr{p}_{\theta_-},0)$ proceeds in a similar way by considering the the Power-Rieffel projection\footnote{Observe that the set of self-adjoint operators $\{\rr{d}_0,\rr{d}_1\}$ which defines $\rr{p}_{\theta_+}$ is in principle different from the set of   self-adjoint operators $\{\rr{d}_0',\rr{d}_1'\}$ which defines $\rr{p}_{\theta_-}$.}
$$\rr{p}_{\theta_-}\;=\;\rr{s}_{b_-,1}^*\;\rr{d}_1'\;+\;\rr{d}'_0\;+\;\rr{d}'_1\;\rr{s}_{b_-,1}\;\in\;\s{A}_{b_-}$$
and the lift 
$$\rr{q}_{-}\;=\;\rr{v}_{-}^*\;\rr{d}_1'\;+\;\rr{d}_0'\rr{p}_-\;+\;\rr{d}_1'\;\rr{v}_{-}\;,$$
where  $\rr{v}_-:=\rr{p}_-\rr{s}_{{\rm I},1}=\rr{p}_-\rr{s}_{b_-,1}$. This time, a direct computation provides
$$
\rr{q}_{-}^2\;=\;\rr{q}_{-}\;-\;(\rr{s}_{b_-,1}^*\rr{d}_1'\rr{s}_{b_-,1})^2\rr{p}_{-1}\;=\;\rr{q}_{-}\;+\;({\rr{d}_0'}^2-\rr{d}_0')\rr{L}' \rr{p}_{-1}\;.
$$
where now $\rr{L}':=\rr{L}'(\rr{s}_{b_-,1}^*\rr{d}_1'\rr{s}_{b_-,1})$ is the  support projection of $\rr{s}_{b_-,1}^*\rr{d}_1'\rr{s}_{b_-,1}$. 
After an induction one gets
$$
\rr{q}_{-}^N\;=\;\rr{q}_{-}\;+\;({\rr{d}'_0}^N-\rr{d}_0')\rr{L}' \rr{p}_{-1}\;=\;
(\rr{q}_{-}-\rr{d}_0'\rr{L}' \rr{p}_{-1})\;+\;(\rr{d}_0'\rr{L}')^N \rr{p}_{-1}\;.
$$
and the exponential of $\rr{q}_{-}$ is given by 
$$
\expo{-\ii 2\pi \rr{q}_-}\;=\;\expo{-\ii 2\pi \rr{d}_0'\rr{L}'}\rr{p}_{-1}\;+\;({\bf 1}-\rr{p}_{-1})\;.
$$
The homotopy argument  $\expo{-\ii 2\pi \rr{d}_0'\rr{L}'}\simeq\rr{s}_{b_-,2}=\rr{s}_{{\rm I},2}$ provided in Lemma \ref{lemma:homot_2} provides 
${\rm exp}([(\rr{p}_{\theta_-},0)])=[ \rr{s}_{{\rm I},2}\rr{p}_{-1}+({\bf 1}-\rr{p}_{-1})]$.
To finish the proof, let us consider the operator $\rr{r}:= \rr{s}_{{\rm I},1}\rr{p}_{-1}+ \rr{s}_{{\rm I},1}^*\rr{p}_0+({\bf 1}-\rr{p}_{-1}-\rr{p}_{0})$. This is an unitary involution in $\s{I}^+$, \ie $\rr{r}=\rr{r}^{-1}=\rr{r}^*$. 
This implies that $[\rr{r}]=[{\bf 1}]$ is the trivial element of $K_1(\s{I})\simeq\Z$ which is torsion-free. As a consequence
$$
\begin{aligned}
{[ \rr{s}_{{\rm I},2}\rr{p}_{-1}+({\bf 1}-\rr{p}_{-1})]}\;&=\;[\rr{r}]+[ \rr{s}_{{\rm I},2}\rr{p}_{-1}+({\bf 1}-\rr{p}_{-1})]+[\rr{r}]\\
&=\;[ \rr{r}(\rr{s}_{{\rm I},2}\rr{p}_{-1}+({\bf 1}-\rr{p}_{-1}))\rr{r}]\\
&=\;[ \expo{\ii b_0}\rr{s}_{{\rm I},2}\rr{p}_{0}+({\bf 1}-\rr{p}_{0})]\\
&=\;[ \rr{s}_{{\rm I},2}\rr{p}_{0}+({\bf 1}-\rr{p}_{0}))]\\
\end{aligned}
$$
where we used  $\rr{r}\rr{p}_{-1}\rr{r}=\rr{p}_{0}$,  $\rr{r}\rr{s}_{{\rm I},2}\rr{p}_{-1}\rr{r}=f_B\rr{s}_{{\rm I},2}\rr{p}_0=\expo{\ii b_0}\rr{s}_{{\rm I},2}\rr{p}_0$ and 
the fact that $\expo{\ii b_0}\rr{s}_{{\rm I},2}$ is connected to 
$\rr{s}_{{\rm I},2}$ by the  homotopy $[0,1]\ni t\mapsto \expo{\ii (1-t) b_0}\rr{s}_{{\rm I},2}$.
\qed

\medskip

The surjectivity of the index map (Proposition \ref{prop:indeximage}) and of the exponential map (Proposition \ref{prop:expimage}) implies that the two maps $\imath_*$ in the diagram \eqref{eq:6-ex-seq} are just the zero maps. After replacing $\imath_*=0$ in  \eqref{eq:6-ex-seq} one obtains the  short exact sequences
$$
\begin{aligned}
0&\;\stackrel{}{\longrightarrow}\;K_0(\s{A}_{\rm I})\;\stackrel{{\rm ev}_\ast}{\longrightarrow}\;K_0(\s{A}_{\rm bulk})\;\stackrel{{\rm exp}}{\longrightarrow}\;K_1(\s{I})\;\stackrel{}{\longrightarrow}\;0\;,\\
0&\;\stackrel{}{\longrightarrow}\;K_1(\s{A}_{\rm I})\;\stackrel{{\rm ev}_\ast}{\longrightarrow}\;K_1(\s{A}_{\rm bulk})\;\stackrel{{\rm ind}}{\longrightarrow}\;K_0(\s{I})\;\stackrel{}{\longrightarrow}\;0\;.
\end{aligned}
$$
As a result, one 
gets further information about the structure of the  $K$-theory of the  Iwatsuka $C^*$-algebra.
\begin{theorem}[$K$-theory of the Iwatsuka $C^*$-algebra II]\label{prop:PV_2_II}
It holds true that
$$
\begin{aligned}
K_0(\s{A}_{\rm bulk})\;&=\;{\rm ev}_*\big(K_0(\s{A}_{\rm I})\big)\;\oplus\;\psi_{\rm exp}\big(K_1(\s{I})\big)\;,\\
K_1(\s{A}_{\rm bulk})\;&=\;{\rm ev}_*\big(K_1(\s{A}_{\rm I})\big)\;\oplus\;\psi_{\rm ind}\big(K_0(\s{I})\big)\;
\end{aligned}
$$
where $\psi_{\rm exp}$ and $\psi_{\rm ind}$ are suitable lifts of the exponential map and of the index map, respectively.
\end{theorem}
\proof
The two  short exact sequences are of the form 
$$
0\;\stackrel{}{\longrightarrow}\;\Z^3\;\stackrel{\alpha}{\longrightarrow}\;\Z^4\;\stackrel{\beta}{\longrightarrow}\;\Z\;\stackrel{}{\longrightarrow}\;0\\
$$
meaning that $\Z^4$ is an abelian extension of $\Z$ by $\Z^3$.
The possible extensions are classified by ${\rm Ext}_\Z(\Z,\Z^3)=0$ \cite[Chapter III]{hilton-stammbach-97}, meaning that only the trivial extension is possible. This in particular ensures the existence of the lifts $\psi_{\rm exp}$ and $\psi_{\rm ind}$.
\qed

\begin{remark}\label{rk:betterK_theo}
We can provide a more precise presentation of $K_1(\s{A}_{\rm bulk})$ by combining Theorem \ref{prop:PV_2_II} with the computation of the map ${\rm ev}_*$ and Proposition \ref{prop:indeximage}. One gets that
$$
\begin{aligned}
{\rm ev}_*\big(K_1(\s{A}_{\rm I})\big)\;&=\;\Z[(\rr{s}_{b_-,2},{\bf 1})]\;+\;\Z[({\bf 1},\rr{s}_{b_+,2})]\;+\;\Z[(\rr{s}_{b_-,1},\rr{s}_{b_+,1})]\;,\\
\psi_{\rm ind}\big(K_0(\s{I})\big)\;&=\;\Z[(\rr{s}_{b_-,1},{\bf 1})]\;,
\end{aligned}
$$
where $[(\rr{s}_{b_-,1},\rr{s}_{b_+,1})]=[(\rr{s}_{b_-,1},{\bf 1})]+[({\bf 1},\rr{s}_{b_+,1})]$ in the sense of the $K_1$-group and the (non-unique) lift $\psi_{\rm ind}$ has been chosen as  
$\psi_{\rm ind}([\rr{p}_0]):=[(\rr{s}_{b_-,1},{\bf 1})]$. A similar analysis for $K_0(\s{A}_{\rm bulk})$ provides
$$
\begin{aligned}
{\rm ev}_*\big(K_0(\s{A}_{\rm I})\big)\;&=\;\Z[({\bf 1},0)]\;+\;\Z[(0,{\bf 1})]\;+\;\Z[(\rr{p}_{\theta_-},\rr{p}_{\theta_+})]\;,\\
\psi_{\rm exp}\big(K_1(\s{I})\big)\;&=\;\Z[(0,\rr{p}_{\theta_+})]\;,
\end{aligned}
$$
where the (non-unique) lift $\psi_{\rm expo}$  is defined by  
$\psi_{\rm expo}([ \rr{p}_-+\rr{s}_{{\rm I},2}\rr{p}_0+\rr{p}_+]):=[(0,\rr{p}_{\theta_+})]$.
Finally, we are in position to say something more about  the 
Power-Rieffel-Iwatsuka projection $\rr{p}_{\rm I}\in \s{A}_{\rm I}\otimes{\rm Mat}_N(\C)$ introduced short after
Theorem \ref{prop:PV_2}. First consider $\rr{p}_{\rm I}\in\s{A}_{\rm I}\otimes\text{Mat}_N(\C)$ and $I_M$ the identity matrix in $\text{Mat}_M(\s{A}_{\rm bulk})$ with  $M\leq N$, such that
$${\rm ev}_*([\rr{p}_{\rm I}]-[I_M])\;=\;[(\rr{p}_{\theta_-},\rr{p}_{\theta_+})]\;.$$
This relation is satisfied in view of the surjectivity of $\rm ev _*$ and the standard picture of  $K_0$-group \cite[Proposition 6.2.7]{wegge-olsen-93}. It follows that
$$ {\rm ev}_*([\rr{p}_{\rm I}])\;=\;[(\rr{p}_{\theta_-},\rr{p}_{\theta_+})]+M[\bf 1]\;,$$
and consequently $\{[\rr{p}_-],[\rr{p}_+],[\rr{p}_{\rm I}]\}$ is a set of generators for $K_0(\s{A}_{\rm I})$. The fact that $[\rr{p}_{\rm I}]$ is the third generator of $K_0(\s{A}_I)$ can be used in the six-term exact sequence \ref{eq:6termexactseq_PV_Iwa2} which provides
$\partial_0[\rr{p}_{\rm I}]\in\{ \pm[\rr{s}_{\rm I,1}] \}$,
showing that $\rr{p}_{\rm I}$ is actually a \emph{PRI-projection}.
It is interesting to note that even though neither $[(\rr{p}_{\theta_-},0)]$ nor $[(0,\rr{p}_{\theta_+})]$ can be lifted into a projection, the existence of the \emph{PRI-projection} implies that the matrix ${\rm diag}((\rr{p}_{\theta_-},\rr{p}_{\theta_+}),I_M)\in\text{Mat}_{M+1}(\s{A}_{\rm bulk})$, can actually be lifted into a \emph{PRI-projection}.
\hfill $\blacktriangleleft$
 \end{remark}

\subsection{Bulk-interface correspondence for the Iwatsuka \texorpdfstring{$C^*$-}-algebra}\label{sect:bulk-int-Iwatsuka}
Let us start with a preliminary result  which is a direct consequence of Proposition \ref{prop:expimage}. 
\begin{lemma}\label{lemma:bId_I}
Let 
$\rr{p}=(\rr{p}_-,\rr{p}_+)\in\s{A}_{\rm bulk}$ be a projection and  $[\rr{p}]\in K_0(\s{A}_{\rm bulk})$ the related class in the $K_0$-group.
Let $N_\pm:={\rm Ch}_{b_\pm}(\rr{p}_{\pm})\in\Z$ be the Chern numbers of $\rr{p}_{\pm}$ defined by \eqref{eq:chern_03}. Then, 
$$
{\rm exp}([\rr{p}])\;=\;(N_--N_+)[\rr{w}_{\s{I}}]\;
$$
where $[\rr{w}_{\s{I}}]$ is the generator of $K_1(\s{I})$ defined in Proposition \ref{prop:k-group-Interface}.
\end{lemma}
\proof
 In terms of the generators of $K_0(\s{A}_{\rm bulk})$ one has that
 $$
[\rr{p}]\;=\; M_-\;[({\bf 1},0)]\;+\;M_+\;[(0,{\bf 1})]\;+\;N_-\;[(\rr{p}_{\theta_-},0)]\;+\;N_+\;[(0,\rr{p}_{\theta_+})]
 $$
 with $M_\pm,N_\pm\in\Z$ suitable integers. The discussion at the end of Appendix  \ref{app:Ktheo NCT}  justifies $N_\pm:={\rm Ch}_{b_\pm}(\rr{p}_{\pm})$. Finally, by using  that the map ${\rm exp}$ is a group homomorphism along with  formulas \eqref{eq:exp_form}, 
one gets the result.
\qed

\medskip

For the next result we need  the winding number $W_{\s{I}}$ defined by \eqref{eq:wind_numb}. The derivation and the trace on  $\s{I}$,
needed to build $W_{\s{I}}$, are described in Section \ref{sect:interf_iwatsuka_MT}
\begin{lemma}\label{lemma:bId_II}
Let $\rr{w}_{\s{I}}\in\s{I}^+$ be the unitary operator defined in Proposition \ref{prop:k-group-Interface}. Then, it holds true that
$$
W_{\s{I}}(\rr{w}_{\s{I}})\;=\;1
$$
\end{lemma}
\proof
An explicit computation provides
$$
\begin{aligned}
(\rr{w}_{\s{I}}^*-{\bf 1})\;\nabla_{\s{I}}(\rr{w}_{\s{I}}-{\bf 1})\;&=\;\ii\rr{p}_0(\rr{s}_{{\rm I},2}^*-{\bf 1})\big[\rr{p}_0(\rr{s}_{{\rm I},2}-{\bf 1}),\rr{n}_2\big]\\
&=\;\ii\rr{p}_0(\rr{s}_{{\rm I},2}^*-{\bf 1})\big[\rr{s}_{{\rm I},2},\rr{n}_2\big]\\
&=\;-\ii\rr{p}_0(\rr{s}_{{\rm I},2}^*-{\bf 1})\rr{s}_{{\rm I},2}\\
&=\;\ii\rr{p}_0(\rr{s}_{{\rm I},2}-{\bf 1})\;.
\end{aligned}
$$
By applying formula \eqref{eq:tra_int} one gets
$$
\bb{T}_{\s{I}}\big((\rr{w}_{\s{I}}^*-{\bf 1})\nabla_{\s{I}}(\rr{w}_{\s{I}}-{\bf 1})\big)\;=\;-\ii\bb{T}_{\s{I}}(\rr{p}_0\rr{q}_0)\;=\;-\ii
$$
since $\rr{q}_0\rr{s}_{{\rm I},2}\rr{q}_0=0$. this completes the proof.
\qed

\medskip

We are now in position to provide our main result, namely the proof of equation \ref{eq:int_01}.
\begin{theorem}[Bulk-interface duality for the Iwatsuka magnetic field]\label{theo_main_Iw}
Let $\hat{\rr{h}}\in\s{A}_{\rm I}$ be a magnetic Hamiltonian  with \emph{non-trivial bulk gap} detected by $\Delta$ (\cf Definition \ref{def_bulk_gap}). Let $g:\R\to[0,1]$ be a non-decreasing (smooth)  function such that $g=0$ below $\Delta$ and $g=1$ above $\Delta$ and consider the unitary operator $\rr{u}_\Delta:=\expo{\ii2\pi g(\hat{\rr{h}})}$ and the associated {interface conductance} (\cf Definition \ref{def:cond_inter})
\begin{equation*}
\sigma_{\s{I}}(\Delta)\;=\; \frac{ e^2}{h}\;W_{\s{I}}(\rr{u}_\Delta)\;.
\end{equation*}
Let $\rr{h}:={\rm ev }(\hat{\rr{h}})=(\rr{h}_-,\rr{h}_+)\in \s{A}_{\rm bulk}$ be the bulk Hamiltonian and for a given Fermi energy inside the bulk gap $\mu\in \Delta$ let $\rr{p}_\mu:=(\rr{p}_{\mu,-},\rr{p}_{\mu,+})$ with $\rr{p}_{\mu,\pm}:=\chi_{(-\infty,\mu]}(\rr{h}_\pm)$ be the associated Fermi projections. Denote with $N_\pm:={\rm Ch}(\rr{p}_{\mu,\pm})\in\Z$  the Chern numbers of such projectors. Then it holds true that 
\begin{equation*}
\sigma_{\s{I}}(\Delta)\;=\;\frac{e^2}{h}(N_+-N_-)\;.
\end{equation*}
\end{theorem}
\proof
We can compute $\sigma_{\s{I}}(\Delta)$ with the topological formula \eqref{ewq:int_condIIII}. From Lemma \ref{lemma:bId_I} and the bilinearity of the canonical pairing between $K_1(\s{I})$ and 
$HC^1(\s{I})$ one obtains 
$$
\prec {\rm exp}([\rr{p}_\mu]),[\eta_{\s{I}}]\succ\;=\;(N_--N_+)\prec [\rr{w}_{\s{I}}],[\eta_{\s{I}}]\succ\;.
$$
Then, equation \eqref{eq:chern_03_BIS} and Lemma \ref{lemma:bId_II}
provide
$$
\prec {\rm exp}([\rr{p}_\mu]),[\eta_{\s{I}}]\succ\;=\;(N_--N_+)W_{\s{I}} (\rr{w}_{\s{I}})\;=\;N_--N_+\;.
$$
This concludes the proof. \qed

\appendix

\section{Discrete Schwartz space}\label{app:dis_Swar}
Just because of the lack of references, the discrete Schwartz space over a Banach $\ast$-algebra will be defined and proved to be a Fréchet algebra. Let $\s{B}$ be a Banach $\ast$-algebra. 
We will use the notation
$$
\{a_m\}\;:=\;\{a_m\;|\; m\in\Z^n\}\;\subset\;\s{B}
$$
for $\s{B}$-valued sequences labelled by  $\Z^n$.
\begin{definition}[Discrete Schwartz Space]  The  discrete Schwartz space of dimension $n$ over  over $\s{B}$ is
	$$
	\s{S}\big(\Z^n,\s{B}\big)\;:=\;\big\{\{a_m\}\subset \s{B}\; \big|\; r_k\big(\{a_m\}\big)<\infty\;, \forall\;k\in\N_0\big\}\;,
	$$
	where
	$$r_k\big(\{a_m\}\big)\;:=\;\sup_{m\in\Z^n}\left(1+|m|^2\right)^{\frac{k}{2}}\;||a_m||_{\s{B}}\;,\qquad k\in\N_0\;.
	$$
\end{definition}

\begin{proposition}
One has that
$$
\s{S}\big(\Z^n,\s{B}\big)\;\simeq\; \s{S}\big(\Z^n\big)\;\otimes\;\s{B}
$$
is  a Fréchet space with respect to the system of norms $r_k$.
\end{proposition}
\proof
Let us recall that  $\s{S}(\Z^n):=\s{S}(\Z^n,\C)$ is a Fréchet nuclear space \cite[Theorem 51.5]{treves-67} and 
$\s{S}(\Z^n)\simeq\s{C}^\infty(\n{T}^n)$ via the Fourier transform \cite[Theorem 51.3]{treves-67}. Then,  the isomorphism stated above follows from \cite[Theorem 44.1]{treves-67}
along with the nuclearity of $\s{S}(\Z^n)$. Finally, let us  recall that the (projective) tensor product of 
Fréchet spaces is again Fréchet. 
\qed

\section{Magnetic Bloch-Floquet transform}\label{app:bloch-floquet}
Let $\s{A}_{A_B}$ be the magnetic $C^*$-algebra of the magnetic field $B$ as in Definition \ref{def:magnetic-algebra}. We are interested in the case where the magnetic field is constant along every vertical line on $\Z^2$, \ie $B(n_1,n_2)=B(n_1)$
for every $n:=(n_1,n_2)\in\Z^2$. The Iwatsuka magnetic field is of course contained in such a class of examples. Observe that these magnetic fields admit Landau-type potentials given by  
$$A_B(n,n-e_j)\;=\;\delta_{j,1}\;n_2\;B(n_1),\quad n\in Z^2,$$
and consequently the pair of magnetic translations which define   $\s{A}_{A_B}$ is
\begin{equation*}
\begin{array}{rl}
\rr{s}_{A_B,1}\;=\;\expo{\ii\rr{n}_2B(\rr{n}_1)}\rr{s}_1\;,\qquad
\rr{s}_{A_B,2}\;=\;\rr{s}_2\;,
\end{array}
\end{equation*}
where $\rr{n}_1,\rr{n}_2$ are the position operators introduced in Section \ref{sect:four_theo}.
Let  $V_f:=\expo{\ii f(\rr{n}_1)}\rr{s}_2\in \s{B}(\ell^2(\Z^2))$ for a given function $f:\Z\rightarrow\R$. The operator $V_f$ is unitary and commute with $\rr{s}_{A_B,2}$ by construction. On the other hand one can compute that
\begin{equation*}
\begin{array}{rl}
V_f\;\rr{s}_{A_B,1}\;V_f^*&\;=\;
\expo{\ii f(\rr{n}_1)}\;\expo{-\ii B(\rr{n}_1)}\;\expo{-\ii f(\rr{n}_1-1)}\;\rr{s}_{A_B,1}\;.
\end{array}
\end{equation*}
It follows that the commutation condition $V_f\rr{s}_{A_B,j}V_f^*=\rr{s}_{A_B,j}$ is guaranteed by
\begin{equation}\label{eq:bloch-floquet-condition}
f(m)\;-\;f(m-1)\;=\;B(m)\;,\qquad \forall\; a\in\Z\;.
\end{equation}
Note that equation \eqref{eq:bloch-floquet-condition} determines the function $f$ up to a constant, that is, by fixing $f(0)=a$ one gets 
$$
f(m)\;:=\;a\;+\;\delta_{m>0}\sum_{j=1}^{m}B(j)\;-\;\delta_{m<0}\sum_{j=0}^{|m|-1}B(-j)\;.
$$

\medskip

The map $\gamma\mapsto V_f^\gamma$ provides a unitary representation of $\Z$ on $\ell^2(\Z^2)$ which commutes with the  magnetic translations $\rr{s}_{A_B,j}$, and consequently with the magnetic algebra $\s{A}_{A_B}$. This fact can be used to define the \emph{magnetic Bloch-Floquet transform}\footnote{The  theory of the Bloch-Floquet transform is described in full generality in the classic monograph \cite{kuchment-93}. The results presented in this section are just an adaptation of the general theory to our specific case which includes a  magnetic field which is constant along one direction.} 
$\s{U}_B$
as follows:
$$
(\s{U}_B\psi)_k\;:=\;\sum_{\gamma\in\Z}\expo{-\ii\gamma k}\;V_f^\gamma\;\psi\;.
$$
The map $\s{U}_B$ is initially defined on the dense domain 
$\psi\in \s{C}_{\rm c}(\Z^2)\subset\ell^2(\Z^2)$ of the compactly supported sequence. From its very definition one gets
\begin{equation}
\begin{aligned}
(\mathcal{U}_B \psi)_{k+2\pi}(n_1,n_2)     &\;=\;(\mathcal{U}_B \psi)_{k}(n_1,n_2)   \\
(\mathcal{U}_B V_f^r \psi)_k     &\;=\;\expo{\ii k r}\;(\mathcal{U}_B\psi)_k\;.
\end{aligned}
\end{equation}
Observe  that the first equation tells us that the natural domain of the parameter $k$ is the circle $\n{S}^1:=\R/2\pi\Z=[0,2\pi)$. The second equation expresses the fact that 
the transformed vectors $(\mathcal{U}_B\psi)_k$ are generalized eigenvectors of $V_f$. The latter condition can be rewritten in the form
$$
(\mathcal{U}_B\psi)_k(n_1,n_2-r)\;=\;\expo{\ii r(k-f(n_1))}(\mathcal{U}_B\psi)_k(n_1,n_2).
$$
and shows that $(\mathcal{U}_B\psi)_k$ is  entirely determined by a single value of $n_2$, \eg by its value on the horizontal line $n_2=0$.
The latter observation allows to  think of $(\mathcal{U}_B\psi)_k$  (for every  fixed $k$) as an element of the \emph{fiber space} $\ell^2(\Z)$ obtained  by fixing $n_2=0$, \ie by setting
$(\mathcal{U}_B\psi)_k(m):=(\mathcal{U}_B\psi)_k(m,0)$ for every $m\in\Z$.
More precisely, one can show that $\mathcal{U}_B$ defined in this way provides a unitary equivalence\footnote{For the general theory of direct integrals of Hilbert spaces we refer to the standard monograph \cite[Part II, Chapter 1]{dixmier-81}. In particular, the isomorphism used in the right-hand side of  equation \eqref{eq:dir_int} is proved in the Corollary on \cite[p. 175]{dixmier-81}.}
\begin{equation}\label{eq:dir_int}
\mathcal{U}_B\;:\;\ell^2(\Z^2)\;\longrightarrow \;\int_{\n{S}^1}^\oplus \dd k\;\ell^2(\Z)\;\simeq\;L^2(\n{S}^1)\otimes\ell^2(\Z)\;,
\end{equation}
where $\dd k$ is the normalized Haar measure of $\n{S}^1$. In fact 
a standard computation shows that $\mathcal{U}_B$ is isometric on the dense domain $\s{C}_{\rm c}(\Z^2)$, hence extends to an isometry on $\ell^2(\Z^2)$. 
Moreover, the inverse map
$\mathcal{U}_B^{-1}$, defined by 
$$
(\mathcal{U}_B^{-1}\phi)(m,r)\;:=\;\int_{\n{S}^1}\dd k\;\expo{-\ii r(k-f(m))}\;\phi_k(m)\;,\qquad \{\phi_k\}\in\rr{H}
$$
satisfies $\mathcal{U}_B^{-1}\mathcal{U}_B={\bf 1}$ on $\s{C}_{\rm c}(\Z^2)$ and is isometric as well. Hence $\mathcal{U}_B^{-1}$ must be injective and as a consequence $\mathcal{U}_B$ must be surjective and thus unitary.

\medskip

Since the magnetic translations commute with the unitary $V_f$ they can be decomposed along the direct integral. A direct computation shows that
\begin{equation*}
\begin{aligned}
V_f         &\;\longmapsto\;v\;:=\:\mathcal{U}_B\; V_f\; \mathcal{U}_B^{-1}\; =\;\int_{\n{S}^1}^{\oplus}\dd k\;\expo{\ii k}{\bf 1}\;\simeq\;\expo{\ii k}\otimes{\bf 1} \\
\rr{s}_{A_B,1}&\;\longmapsto\;s\;:=\;\mathcal{U}_B\; \rr{s}_{A_B,1}\; \mathcal{U}_B^{-1}\; =\;\int_{\n{S}^1}^{\oplus}\dd k\;\rr{s}\;\simeq\;{\bf 1}\otimes \rr{s}\\
\rr{s}_{A_B,2}&\;\longmapsto\;t\;:=\;\mathcal{U}_B\; \rr{s}_{A_B,2}\; \mathcal{U}_B^{-1}\;=\;\int_{\n{S}^1}^{\oplus}\dd k\;\expo{\ii k}\expo{-\ii f(\rr{n})}\;\simeq\;\expo{\ii k}\otimes \expo{-\ii f(\rr{n})}\\
\end{aligned}
\end{equation*}
where $\rr{s}$ and $\rr{n}$ are the usual shift and position operator on $\ell^2(\Z)$. Moreover, one has that
$$
\mathcal{U}_B\; \rr{f}_B\; \mathcal{U}_B^{-1}\;=\;sts^*t^*\;\simeq\;{\bf 1}\otimes \expo{\ii B(\rr{n})}
$$
where $B(\rr{n})$ denotes the multiplication operator on 
$\ell^2(\Z)$ given by the restriction of the magnetic field $B(m):=B(m,0)$. As a consequence of the formulas above one gets that the magnetic Bloch-Floquet transform maps $\s{A}_{A_B}$ as a subalgebra of $\s{C}(\n{S}^1)\otimes\s{B}(\ell^2(\Z))$.

\section{Crossed product structure}\label{sect:crossedproduct}
In Section \ref{sect:magnetic_C_Al} we provided an explicit construction of the magnetic $C^*$-algebra $\s{A}_{A_B}$ associated to a vector potential $A_B$ for the magnetic field $B:\Z^2\to\R$. In this section we will show that the magnetic $C^*$-algebra $\s{A}_{A_B}$ is the \virg{concrete} realization of 
an abstract  \emph{twisted crossed product} $C^*$-algebra over $\Z^2$. For the general theory of the crossed product $C^*$-algebras we will refer to classic monographs \cite{pedersen-79,williams-07}. The special case of discrete crossed product $C^*$-algebras, which is the most related to our construction,  is discussed in detail in \cite[Chapter VIII]{davidson-96}.

\medskip

Let $(\s{C}(\Omega_B),\tau,\Z^2)$ be the $C^*$-dynamical system associated with the dynamical system
$(\Omega_B, \tau^*,\Z^2)$ described in Section \ref{sect:integr}. 
Consider a  pair of (abstract) unitary elements  $u_1,u_2$ 
and for every $\gamma:=(\gamma_1,\gamma_2)\in\Z^2$ set
$u_\gamma:=u_1^{\gamma_1}u_2^{\gamma_2}$. Let $\s{C}(\Omega_B)[\Z^2]$ be the set of  finite sums 
$$
G\;:=\;\sum_{\gamma\in\Lambda}g_\gamma\; u_\gamma
$$
with $g_\gamma\in \s{C}(\Omega_B)$ for all $\gamma\in \Lambda$
and $\Lambda\in \s{P}_0(\Z^2)$ a finite subset of $\Z^2$. The 
product in $\s{C}(\Omega_B)[\Z^2]$ is defined by the rules
$$
u_2\;u_1\;=\;\overline{f_B}\;u_1\;u_2\;,\qquad u_\gamma\;g\;u_{-\gamma}\;=\; \tau_\gamma(g)
$$
for every $\gamma\in\Z^2$ and $g\in\s{C}(\Omega_B)$,
where
$f_B:=\expo{\ii B}\in \s{C}(\Omega_B)$ is the magnetic phase associated to $B$.
The involution is provided by
$$
(g\; u_\gamma)^*\;:=\;u_{-\gamma}\;\overline{g}\;=\;\tau_{-\gamma}\big(\overline{g}\big)\;u_{-\gamma}\;.
$$
Endowed with these operations $\s{C}(\Omega_B)[\Z^2]$ acquires the structure of a unital $\ast$-algebra. Moreover it can be completed to a Banach $\ast$-algebra with respect to the norm
$$
\|G\|_1\;:=\;\sum_{\gamma\in\Lambda}\|g_\gamma\|_\infty\;.
$$
The (universal) enveloping $C^*$-algebra \cite[Section 2.7]{dixmier-77} obtained 
from this Banach $\ast$-algebra is called the \emph{$B$-twisted crossed product} of  $\s{C}(\Omega_B)$ and is denoted with $\s{C}(\Omega_B)\rtimes_{\tau,B}\Z^2$.

\medskip

In order to better understand  the twisted structure of the  crossed product   $\s{C}(\Omega_B)\rtimes_{\tau,B}\Z^2$
one can observe that the mapping $\gamma\mapsto u_\gamma$ provides a projective (abstract) unitary representation of $\Z^2$
defined by
$$
u_\gamma\;u_\xi\;=\;\Theta_B(\gamma,\xi)\;u_{\gamma+\xi}\;,\qquad \forall\; \gamma,\xi\in\Z^2\;,
$$
with phase  given by
$$
\Theta_B(\gamma,\xi)\;:=\;\prod_{\nu\in\Lambda(\gamma,\xi)}\overline{\tau_{\nu}({f}_B)}
$$
and the product is extended on the finite cell
$$
\Lambda(\gamma,\xi)\;:=\;\big([\gamma_1,\gamma_1+\xi_1-1]\times[0,\gamma_2-1]\big)\;\cap\;\Z^2\;.
$$
The map $\Theta_B: \Z^2\times \Z^2\to \s{U}(\s{C}(\Omega_B))$ takes value on the unitary elements of $\s{C}(\Omega_B)$ and a direct check shows that it satisfies the cocycle condition
$$
\Theta_B(\gamma+\xi,\zeta)\; \Theta_B(\gamma,\xi)\;=\;\tau_{\gamma}\big(\Theta_B(\xi,\zeta)\big)\;\Theta_B(\gamma, \xi+\zeta)
$$
for all $\gamma,\xi,\zeta\in\Z^2$.

\medskip

Given a vector potential $A_B$ for the magnetic field $B$
one can consider the representation $\pi_{A_B}:\s{C}(\Omega_B)\rtimes_{\tau,B}\Z^2\to\s{B}(\ell^2(\Z^2))$ defined by
$$
\begin{aligned}
\pi_{A_B}(u_\gamma)\;&:=\;(\rr{s}_{A_B,1})^{\gamma_1}\;(\rr{s}_{A_B,2})^{\gamma_2}\\
\pi_{A_B}(g)\;&:=\;\iota^{-1}(g)
\end{aligned}\;,\qquad \forall \gamma\in\Z^2\;,\;\;\forall g\in \s{C}(\Omega_B)
$$
where $\rr{s}_{A_B,1},\rr{s}_{A_B,2}$ are the magnetic translations defined in Section \ref{sect:gen_MT} and $\iota^{-1}(g)\in \s{F}_B$ is given by the isomorphism
defined in Lemma \ref{lemma:inv_prop}.
The map $\pi_{A_B}$ coincides with the tensor product of the isomorphism $\iota^{-1}$ with the $B$-twisted
 \emph{(left) regular} representation of $\Z^2$. This representation turns out to be faithful (\cf \cite[p. 218]{davidson-96}) and   as a consequence one gets
 $$
 \s{A}_{A_B}\;=\;\pi_{A_B}\left(\s{C}(\Omega_B)\rtimes_{\tau,B}\Z^2\right)\;.
 $$
 
 \medskip

It is also interesting to observe that $\s{A}_{A_B}$ can be represented as an iterated crossed product algebra as discussed in \cite[Section 3.1.1]{prodan-schulz-baldes-book}. Indeed one can check that
$$
\s{C}(\Omega_B)\rtimes_{\tau,B}\Z^2\;\simeq\;\s{Y}_{B,j}\rtimes_{\alpha_k}\Z\;,\qquad \s{Y}_{B,j}\;:=\;\s{F}_B\rtimes_{\alpha_j}\Z,
$$
where $\{j,k\}=\{1,2\}$,
the  crossed product algebra $\s{Y}_{B,j}$ is generated by $\s{F}_B$ and $u_j$ along with the relation 

\begin{equation*}
\alpha_j(g):=u_j gu_j^*\;:=\;\left\{
\begin{aligned}
&\tau_{(1,0)}(g)&\text{if }j=1\\
&\tau_{(0,1)}(g)&\text{if }j=2\\
\end{aligned}
\right\}\;=:\tau_j(g)
\end{equation*}
for all $g\in \s{C}(\Omega_B)$ and the crossed product algebra $\s{Y}_{B,j}\rtimes_{\alpha_k}\Z$ is generated by $\s{Y}_{B,j}$ and $u_k$ along with the relation
$\alpha_k(gu_j):=u_k(gu_j)u_k^*=\tau_k(g)f_Bu_j$
for all $g\in \s{F}_B$.

\medskip

In the special case the magnetic field $B$ is constant in the $n_2$-direction (as in the case of the Iwatsuka magnetic field) it follows that the $\alpha_2$-action is trivial, meaning that it reduces to the identity  $\alpha_2(g)=g$ for all $g\in \s{C}^*(\Omega_B)$. In this situation the crossed product algebra $\s{Y}_{B,2}$ acquires the following very simple structure
$$
\s{Y}_{B,2}\;\simeq\;\s{C}^*(\Omega_B)\otimes\;C^*_r(\Z)\;\simeq\;\s{C}^*(\Omega_B)\otimes \s{C}(\n{S}^1)\;\simeq\;\s{C}(\Omega_B\times\n{S}^1)
$$
The first isomorphism involves the  (reduced) group $C^*$-algebra $C^*_r(\Z)$ and is proved in \cite[Lemma 2.73]{williams-07} (along with the nuclearity of the  various $C^*$-algebras).
The isomorphism $C^*_r(\Z)\simeq \s{C}(\n{S}^1)$ is a consequence of the Pontryagin duality
\cite[Proposition VII.1.1]{davidson-96}.

\section{The \texorpdfstring{$K$-}-theory of the Iwatsuka magnetic hull}\label{app:K-theoIMH}
Let $\Omega_I=\Z\cup\{-\infty\}\cup\{+\infty\}$ be the Iwatsuka magnetic hull described in Example \ref{ex:hull_Iwatsuka} and consider the short exact sequence
\begin{equation}\label{eq:ex_seq_iwa_hull}
0\;\stackrel{}{\longrightarrow}\;\s{C}_0(\Z)\;\stackrel{\imath}{\longrightarrow}\;\s{C}(\Omega_I)\;\stackrel{\rm ev}{\longrightarrow}\C\oplus\C\;\stackrel{}{\longrightarrow}0\;
\end{equation}
where $\s{C}_0(\Z)$ is the $C^*$-algebra of sequences vanishing at infinity, $\imath$ is the inclusion homomorphism
and the evaluation homomorphism ${\rm ev}$ compute the left and right limits of elements in  $\s{C}(\Omega_I)$. The sequence \eqref{eq:ex_seq_iwa_hull} is split exact in view of the homomorphism
$$
\C\oplus\C\;\ni\;(\ell_-,\ell_+)\;\stackrel{\jmath}{\longrightarrow}\;c_{(\ell_-,\ell_+)}\;\in\;\s{C}(\Omega_I)
$$
where the element $c_{(\ell_-,\ell_+)}$ is specified by
$$
c_{(\ell_-,\ell_+)}(n)\;:=\;
\left\{
\begin{aligned}
&\ell_-&\quad&\text{if}\;\; n<0\\
&\ell_+&\quad&\text{if}\;\; n\geqslant0\;.
\end{aligned}
\right.
$$
Then, it follows that \cite[Corollary 8.2.2]{wegge-olsen-93}
$$
K_j\big(\s{C}(\Omega_I)\big)\;\simeq\;K_j\big(\s{C}_0(\Z)\big)\;\oplus\;K_j\big(\C\oplus\C\big)\;,\qquad j=1,2\;.
$$
The $K$-theory of $\C\oplus\C$ is easily calculated as $K_0(\C\oplus\C)=\Z\oplus\Z$ and $K_1(\C\oplus\C)=0$. The $K$-theory of $\s{C}_0(\Z)$ is given by $K_0(\s{C}_0(\Z))=\Z^{\oplus\Z}$ and $K_0(\s{C}_0(\Z))=0$. The latter fact follows from the isomorphism $K_j(\s{C}_0(\Z))\simeq K^j_{\rm top}(\Z)\simeq K^j_{\rm top}(\ast)^{\oplus\Z}$ between the algebraic and the topological $K$-theory \cite[Theorem 5]{baum-sanchez-garcia-11}. Another way of achieving the same result is to consider the Pontryagin duality $\n{S}^1=\widehat{\Z}$ and the isomorphism $\s{C}_0(\Z)\simeq C^*_r(\n{S}^1)$
where $C^*_r(\n{S}^1)$ is the (reduced) group algebra of the circle \cite[Proposition VII.1.1]{davidson-96}. Therefore, one has that $K_0(C^*_r(\n{S}^1))\simeq {\rm Rep}(\n{S}^1)\simeq\Z^{\oplus\Z}$ and $K_0(C^*_r(\n{S}^1))\simeq0$ where ${\rm Rep}(\n{S}^1)$ denotes the complex representation ring of $\n{S}^1$ \cite[Section 7]{baum-sanchez-garcia-11}. The generators of $K_j(\s{C}_0(\Z))$ are the classes $[\pi_i]$, $i\in\Z$, of the projections $\pi_i(n):=\delta_{i,n}$. After putting all the information together, we can describe the $K$-theory of the Iwatsuka magnetic hull as 
$$
K_0\big(\s{C}(\Omega_I)\big)\;\simeq\;\bigoplus_{i\in\Z}\Z[\pi_i]\;\oplus\;\Z[\pi_-]\;\oplus\;\Z[\pi_+]\;,\qquad K_1\big(\s{C}(\Omega_I)\big)\;=\;0\;.
$$
where  $\pi_-:=\jmath((1,0))$ and $\pi_+:=\jmath((0,1))-\pi_0$
are the projections \emph{at infinity}.

\section{The Pimsner-Voiculescu exact sequence}\label{app:PV}

In this section we will provide a  brief overview on the 
\emph{Pimsner-Voiculescu six-term exact sequence} which is the main tool to compute the $K$-theory for crossed product $C^*$-algebras by $\Z$. For the interested reader we refer to the original work \cite{pimsner-voiculescu-exact-seq} and the monograph \cite[Chapter V]{blackadar-98}.

\medskip

Let $\s{Y}$ be a $C^*$-algebra, $\alpha\in {\rm Aut}(\s{Y})$ and automorphism and $\s{Y}\rtimes_{\alpha}\Z$  the crossed product generated by $\s{Y}$ and the unitary $u$ with the relation 
$$
\alpha(a)\;=\:uau^*\;,\qquad\forall\; a\in\s{Y}\;.
$$
The first step of the construction is to define an appropriate short exact sequence of $C^*$-algebras. This is done by considering the tensor product $\s{Y}\otimes\s{K}$, where $\s{K}$ denotes the $C^*$-algebra of compact operators, and the $C^*$-algebra $\mathcal{T}_\alpha$ 
generated in $\s{Y}\otimes{C^*}(v)$ by $\s{Y}\otimes {\bf 1}$
and $V=u\otimes v$, with $v$ a non-unitary (abstract) isometry. It is useful to think at elements of $\s{K}$ as infinite matrices acting on $\ell^2(\N_0)$ with respect its canonical basis.
Let us consider the map $\varphi:\s{Y}\otimes\s{K}\to \s{T}_\alpha$ defined by
$$
\varphi(a\otimes e_{j,k})\;:=\;V^j(a\otimes({\bf 1}-vv^*))(V^*)^k\;=\;(\alpha^j(a)\otimes{\bf 1})V^jP(V^*)^k
$$
where $P$ is the (non-trivial) self-adjoint projection given by  $P:={\bf 1}-VV^*={\bf 1}\otimes({\bf 1}-vv^*)$ and $e_{j,k}$ are the rank one operators which generates $\s{K}$.
Then, there exists a short exact sequence of $C^*$-algebras
\begin{equation}\label{eq:actual_toeplitz_ext_app}
	0\;\longrightarrow\; \s{Y}\otimes \s{K} \;\stackrel{\varphi}{\longrightarrow}\; \mathcal{T}_{\alpha} \;\stackrel{\psi}{\longrightarrow}\; \s{Y}\rtimes_{\alpha}\Z \;\longrightarrow\; 0,
	\end{equation}
where the map $\psi:\s{T}_\alpha\to \s{Y}\rtimes_{\alpha}\Z$ defined by
$$
\psi(a\otimes {\bf 1})\;:=\;a\;,\qquad\quad \psi(V)\;:=\;u\;.
$$
For this reason $\s{T}_\alpha$ is called the Toeplitz extension of the stabilized algebra  $\s{Y}\otimes \s{K}$ by the crossed product  $\s{Y}\rtimes_{\alpha}\Z$.

\medskip

The Pimsner-Voiculescu (six-term) exact sequence is a cyclic sequence which connects the $K$-theory of $\s{Y}$ and $\s{Y}\rtimes_{\alpha}\Z$, and is given by
\begin{equation}\label{eq:6termexactseq_PV_gen_app}
\begin{array}{ccccc}
K_0(\s{Y})&\stackrel{\beta_{\ast}}{\longrightarrow}          
 &K_0(\s{Y}) &      \stackrel{{\imath}_*}{\longrightarrow}    
  &K_0(\s{Y}\rtimes_{\alpha}\Z)\\
  &&&&\\
{\partial_1}\Big\uparrow   &&&&\Big\downarrow{\partial_0}\\
   &&&&\\
  K_1(\s{Y}\rtimes_{\alpha}\Z)& \underset{{\imath}_*}{\longleftarrow}       
 &K_1(\s{Y}) &   \underset{\beta_{\ast}}{\longleftarrow}        
  &K_1(\s{Y})\\
\end{array}
\end{equation}
and it is worth pointing out that this is not exactly the standard six-term exact sequence associated with the  short exact sequence \eqref{eq:actual_toeplitz_ext_app}, although it is closely related.
The maps ${\imath}_*$ are induced by the canonical inclusion $\imath:\s{Y}\hookrightarrow\s{Y}\rtimes_{\alpha}\Z$ and the maps $\beta_{\ast}$ are induced by the map
$\beta :\s{Y}\to\s{Y}$ defined as $\beta:={\rm Id}-\alpha^{-1}$.
The vertical maps are related with the index and the exponential maps for the standard six-term exact sequence in $K$-theory emerging from the short exact sequence \eqref{eq:actual_toeplitz_ext_app} (\cf \cite[Theorem 9.3.2]{wegge-olsen-93}).
 More precisely one has that $\partial_0:=\kappa_0^{-1}\circ {\rm ind}$ and 
$\partial_1:=\kappa_0^{-1}\circ {\rm exp}$ where ${\rm ind}:K_1(\s{Y}\rtimes_{\alpha}\Z)\to K_0(\s{Y}\otimes\s{K})$ and ${\rm exp}:K_0(\s{Y}\rtimes_{\alpha}\Z)\to K_1(\s{Y}\otimes\s{K})$ are the usual index and the exponential maps related to the short exact sequence \eqref{eq:actual_toeplitz_ext_app} and $\kappa_0:  K_j(\s{Y})\to K_j(\s{Y}\otimes\s{K})$, with $j=0,1$, is the  stabilization isomorphism induced by $a\mapsto a\otimes e_{0,0}$ for every $a\in\s{Y}$.


\section{\texorpdfstring{$K$-}-theory for a constant magnetic field}\label{app:Ktheo NCT}

In this section the $K$-theory of the magnetic $C^*$-algebra $\s{A}_b$ associated with a constant
field of strength $b$ will be described. The key observation is that $\s{A}_b$ is a faithful representation of the noncommutative torus $\n{A}_{\theta_b}$ provided that $\theta_b:=b(2\pi)^{-1}$. Let us observe that $b$ enters in the definition of $\n{A}_{\theta_b}$ only modulo $2\pi$. For this reason, without loss of generality, we can assume $0< \theta_b <1$
as the {condition} for a \emph{non trivial} magnetic field.
The $K$-theory of the noncommutative torus $\n{A}_{\theta_b}$ has been investigated in 
\cite{embedding-rotation-alg,pimsner-voiculescu-exact-seq,power-rieffel-proj} and is  described in several textbooks like \cite[Section 12.3]{wegge-olsen-93} or \cite[Chapter 12]{gracia-varilly-figueroa-01}. As a consequence of the isomorphism $\s{A}_b\simeq \n{A}_{\theta_b}$
we get
\begin{equation}\label{eq:K-teo_NCT_ind}
\begin{aligned}
K_0(\s{A}_{b})\;&=\;\Z[{\bf 1}]\;\oplus\;\Z[\rr{p}_{\theta_b}]\;\simeq\;\Z^2
\;,\\
K_1(\s{A}_{b})\;&=\;\Z[\rr{s}_{b,1}]\;\oplus\;\Z[\rr{s}_{b,2}]\;\simeq\;\Z^2\;.
\\
\end{aligned}
\end{equation}

\medskip

The generators of the $K$-theory of $\s{A}_{b}$ are quite explicit except for the projection 
$\rr{p}_{\theta_b}\in \s{A}_{b}$ which is known as \emph{Power-Rieffel projection}. 
Our next task is to provide a presentation of $\rr{p}_{\theta_b}$ optimized for the aims of this work. We will set
$$
\rr{p}_{\theta_b}\;:=\rr{s}_{b,1}^*\;\rr{d}_1\;+\;\rr{d}_0\;+\;\rr{d}_1\;\rr{s}_{b,1}
$$
where $\rr{d}_1:=g(\rr{s}_{2})$ and $\rr{d}_0:=f(\rr{s}_{2})$
are suitable self-adjoint elements of $C^*(\rr{s}_{2})\subset\s{A}_{b}$. Here we are using the coincidence $\rr{s}_{2}=\rr{s}_{b,2}$  between the ordinary shift and magnetic translation in view of the election of the Landau gauge for the constant magnetic field. The requirement for $\rr{p}_{\theta_b}$ of being a projection is automatically satisfied it the following relations hold true:
\begin{equation}\label{eq:conPRproj}
\begin{aligned}
(\rr{s}_{b,1}^*\rr{d}_1\rr{s}_{b,1})\rr{d}_1\;&=\;0\;,\\
\rr{d}_1(\rr{d}_0+\rr{s}_{b,1}\rr{d}_0\rr{s}_{b,1}^*)\;&=\;\rr{d}_1\;,\\
\rr{d}_0^2+\rr{d}_1^2+(\rr{s}_{b,1}^*\rr{d}_1\rr{s}_{b,1})^2\;&=\;\rr{d}_0\;.
\end{aligned}
\end{equation}
The way of implementing these relation is by the isomorphism (induced by the Fourier transform)
$C^*(\rr{s}_{2})\simeq C_{\rm per}([0,1])$ where on the right-hand side one has the $C^*$-algebra of continuous function on $[0,1]$ with periodic boundary conditions, \ie $f(0)=f(1)$.
Under this isomorphism $\rr{s}_{2}\mapsto e$ where $e(k):=\expo{\ii 2\pi k}$
and $\rr{s}_{b,1}\rr{s}_{2}\rr{s}_{b,1}^*=\expo{\ii b} \rr{s}_{2}\mapsto e(\cdot+\theta_b)$
Consider a $0<\delta<\theta_b$ such that $\theta_b+\delta<1$ and the function $f$ such that
$$
f(e(k))\;:=\;\frac{k}{\delta}\chi_{[0,\delta]}(k)+\chi_{(\delta,\theta_b)}(k)+\left(1+\frac{\theta_b-k}{\delta}\right)\chi_{[\theta_b,\theta_b+\delta]}(k)\;.
$$
Define
$$
g(e(k))\;:=\;\sqrt{f(e(k))(1-f(e(k)))}\;\chi_{[0,\delta]}(k)\;=\;\sqrt{\frac{k}{\delta}\left(1-\frac{k}{\delta}\right)}\;\chi_{[0,\delta]}(k)\;.
$$
One can check that by using  $f$ and $g$ above to define  $\rr{d}_0$ and $\rr{d}_1$ respectively, then the conditions \eqref{eq:conPRproj} are automatically verified. The crucial identities are 
$\rr{s}_{b,1}^*\rr{d}_1\rr{s}_{b,1}\mapsto g(e(\cdot -\theta_b))$ (which is supported in $[\theta_b,\theta_b+\delta]$)  and $\rr{s}_{b,1}\rr{d}_0\rr{s}_{b,1}^*\mapsto f(e(\cdot +\theta_b))$ (which must be defined periodically on $[0,1]$). Let us point out that with a standard \virg{smoothing argument} it is possible to replace the continuous function $f$ and $g$ with smooth functions. This implies that it is possible to define the Power-Rieffel projection inside the smooth algebra $\s{A}_{b}^\infty$.

\medskip

The relations \eqref{eq:conPRproj} provide other useful identities. Let $\rr{L}:=\rr{L}(\rr{d}_1)$ be the support projection of $\rr{d}_1$
(in the von Neumann algebra generated by $\s{A}_{b}$). This is by definition the  smallest projection such that $\rr{L}\rr{d}_1=\rr{d}_1=\rr{d}_1\rr{L}$. It is immediate to conclude that $\rr{L}$ is mapped into the characteristic function on the support of $g\circ e$ under the isomorphism used above, \ie $\rr{L}\mapsto\chi_{[0,\delta]}$. Combining $\rr{L}$ with the first relation in \eqref{eq:conPRproj} one gets $(\rr{s}_{b,1}^*\rr{d}_1\rr{s}_{b,1})\rr{L}=0$. This relation combined with the third equation in \eqref{eq:conPRproj} provides
\begin{equation}\label{eq:ulteram}
\rr{d}_1^2\;=\;\rr{L}(\rr{d}_0-\rr{d}_0^2)\;=\;(\rr{d}_0-\rr{d}_0^2)\rr{L}\;.
\end{equation}
For the next result we need to recall that two unitary operators $\rr{u}_0,\rr{u}_1\in \s{A}_{b}$ are said to be \emph{homotopic equivalent}, denoted $\rr{u}_0\sim\rr{u}_1$,   if there is a continuous map $[0,1]\ni t\mapsto \rr{u}(t)\in  \s{A}_{b}$ such that $\rr{u}(0)=\rr{u}_0$,  $\rr{u}(1)=\rr{u}_1$ and $\rr{u}(t)$ is unitary for every $t\in [0,1]$. 
\begin{lemma}\label{lemma:homot_1}
The unitary operators $\expo{-\ii 2\pi \rr{d}_0\rr{L}}$ and $\rr{s}_{b,2}^*$ are  homotopic equivalent in $\s{A}_{b}$, \ie
$\expo{-\ii 2\pi \rr{d}_0\rr{L}}\sim\rr{s}_{b,2}^*$.
\end{lemma}
\proof
In view of the isomorphism $C^*(\rr{s}_{2})\simeq C_{\rm per}([0,1])$ it is  enough to find an homotopy between the functions
$t(k):=\expo{-\ii 2\pi\frac{k}{\delta}\chi_{[0,\delta]}(k)}$
and $e(k)^{-1}:=\expo{-\ii 2\pi k}$. Such an homotopy is explicitly given by 
$$
[0,1]\ni t\;\longmapsto\;{u}_t(k)\;:=\;\expo{-\ii 2\pi  \frac{k}{\delta+t(1-\delta)}\chi_{[0,\delta+t(1-\delta)]}(k)}\;
$$
and this completes the proof.
\qed

\medskip

Let $\rr{L}'$ be the support projection of the shifted operator
$\rr{s}_{b,1}^*\rr{d}_1\rr{s}_{b,1}$. It turns out that $\rr{L}'$
is isomorphically mapped into the characteristic function $\chi_{[\theta_b,\theta_b+\delta]}$.
\begin{lemma}\label{lemma:homot_2}
The unitary operators $\expo{-\ii 2\pi \rr{d}_0\rr{L}'}$ and $\rr{s}_{b,2}$ are  homotopic equivalent in $\s{A}_{b}$, \ie
$\expo{-\ii 2\pi \rr{d}_0\rr{L}'}\sim\rr{s}_{b,2}$.
\end{lemma}
\proof
As above it is  enough to find an homotopy between the functions
$t'(k):=\expo{-\ii 2\pi\left(1+\frac{\theta_b-k}{\delta}\right)\chi_{[\theta_b,\theta_b+\delta]}(k)}$
and $e(k):=\expo{\ii 2\pi k}$. Such an homotopy is explicitly given by 
$$
[0,1]\ni t\;\longmapsto\;{u}'_t(k)\;:=\;\expo{-\ii 2\pi  \left(
\frac{(1-t)(\theta_b-k+\delta-1)-t\delta k}{\delta}	
\right)\chi_{[(1-t)\theta_b,(1-t)(\theta_b+\delta-1)+1]}(k)}\;
$$
and this completes the proof.
\qed

\medskip

 Let us end this appendix with a more precise description of  the $K_0$-group of $\s{A}_{b}$. Let $[\rr{p}]\in K_0(\s{A}_{b})$. Then, from the first equation of \eqref{eq:K-teo_NCT_ind} one infers the existence of $M,N\in\Z$ such that
 $$
 [\rr{p}]\;\simeq\;M\;[{\bf 1}]\;+\;N\;[\rr{p}_{\theta_b}]\;.
 $$
The number $N$ can be deduced by using the pairing \eqref{eq:chern_03} along with $\prec[{\bf 1}],[\xi_b]\succ=0$ and 
$\prec [\rr{p}_{\theta_b}],[\xi_b]\succ=1$. This implies that $N={\rm Ch}_{b}(\rr{p})$.
The number $M$ can be deduced from the pairing $\tau:K_0(\s{A}_{b})\to \Z+\theta_b \Z$ induced by the trace, \ie $\tau([\rr{p}]):=\bb{T}(\rr{p})$. Since  $\tau([{\bf 1}])=1$ and 
$\tau([\rr{p}_{\theta_b}])=\theta_b$ one gets that $
M=\bb{T}(\rr{p})-N\theta_b$.

\end{document}